\documentclass[a4paper,UKenglish,cleveref,pdfa]{lipics-v2021}
\nolinenumbers
\ArticleNo{1}
\hideLIPIcs
\usepackage{hyperref}
\usepackage{booktabs} 
\usepackage{amssymb,amsmath,amsthm,mathtools}
\usepackage[inline]{enumitem}
\usepackage{tikz}
\usetikzlibrary{calc,positioning}
\usepackage{tikz-cd}
\usepackage{xcolor}
\usepackage{hyperref}
\usepackage{comment}
\usepackage{calc}
\usepackage{longtable}

\newcommand{\TA}{{\mathsf{TA}}}
\newcommand{\FOL}{{\mathsf{FOL}}}
\newcommand{\Sig}{\mathsf{Sig}}
\newcommand{\Mod}{\mathsf{Mod}}
\newcommand{\Sen}{\mathsf{Sen}}

\newcommand{\Set}{\mathbb{S}\mathsf{et}}
\newcommand{\A}{\mathfrak{A}}
\newcommand{\B}{\mathfrak{B}}

\newcommand{\card}{\mathsf{card}}
\newcommand{\act}{\mathfrak{t}}
\newcommand{\Space}{~~~~~}
\newcommand{\pos}[1]{{\langle}#1{\rangle}}
\newcommand{\Forall}[1]{\forall #1\,{\cdot}\,}
\newcommand{\Exists}[1]{\exists #1\,{\cdot}\,}
\newcommand{\proofrule}[2]{\displaystyle\frac{#1}{#2}}
\newcommand{\red}{\mathord{\upharpoonright}}
\usepackage{scalerel}

\newcommand{\bbsemicolon}{%
  \scalerel*{%
    \hbox{\usefont{U}{bbold}{m}{n} ;}%
  }{;}%
}
\newcommand{\comp}{\operatorname{\mathbin{\bbsemicolon}}}

\newcommand{\one}{\scalebox{0.75}{$\BS{1}$}}

\newcommand{\empact}{\scalebox{0.75}{$\BS{0}$}}
\newcommand{\precupact}{%
  \scalerel*{%
    \hbox{\raisebox{0.3ex}{\scalebox{0.5}{$\cup$}}}%
  }{;}%
}
\newcommand{\cupact}{\operatorname{\mathbin{\precupact}}}

\newcommand{\preimpact}{%
  \scalerel*{%
    \hbox{\raisebox{0.3ex}{\scalebox{0.5}{\rotatebox[origin=c]{180}{$\mathsf{C}$}}}}%
  }{;}%
}
\newcommand{\impact}{\operatorname{\mathbin{\preimpact}}}

\newcommand{\resi}{\operatorname{\mathbin{\triangleright}}}

\newcommand{\Co}{\complement}%
\newcommand{\Coo}{{{}^{{}_{\complement}}}}%

\setlist[enumerate,1]{
  font=\normalfont,
  label={\arabic*.},
  ref=\arabic*
}

\newlist{inlinenum}{enumerate*}{1}
\setlist[inlinenum,1]{
  font=\normalfont,
  label=(\emph{\alph*})
}

\setlist[itemize,1]{
}

\setlist[description]{
  font=\normalfont\em,
  leftmargin=\parindent
}

\newlist{plainlist}{itemize}{1}
\setlist[plainlist]{
  label={},
  leftmargin=0pt,
  itemsep=\parskip
}
\newcommand*{\pcformat}[1]{%
  [\;{\normalfont\itshape #1}\;]%
}

\newenvironment{proofcases}[1][]{%
  \description[font=\pcformat, leftmargin=\parindent, #1]%
}{\enddescription}

\usepackage{array,etoolbox}
\makeatletter

\newlength{\PS@lastparam}
\newlength{\PSlastparam}
\newcommand{\PSlp}{%
  \setlength{\PSlastparam}{\PS@lastparam}%
  \the\PSlastparam
}
\def\PS@sub@lastparam{}

\newcommand{\PS@numwidth}{99}
\newcommand{\PSnumwidth}[1]{%
  \renewcommand{\PS@numwidth}{#1}%
}

\newcommand{\PS@style}{\small}
\newcommand{\PS@numstyle}{\footnotesize}

\newlength{\PSindent}
\newlength{\PS@extraindent}
\newlength{\PSpre}
\newlength{\PSpost}
\newlength{\PS@Nwidth}
\newlength{\PS@Swidth}
\newlength{\PS@Ewidth}
\newlength{\PScolsep}
\newcommand{\PS@rownumber}{%
  \ifPS@subsubsteps
  \thePSsubstepc.%
  \the\numexpr\value{PSsubsubstepc}+1\relax
  \else
  \ifPS@substeps
  \thePSstepc.%
  \the\numexpr\value{PSsubstepc}+1\relax
  \else
  \the\numexpr\value{PSstepc}+1\relax
  \fi\fi
}
\newcommand{\PS@step}{%
  \ifPS@subsubsteps
  \refstepcounter{PSsubsubstepc}%
  \else
  \ifPS@substeps
  \refstepcounter{PSsubstepc}%
  \else
  \refstepcounter{PSstepc}
  \fi\fi%
}

\newif\ifPS@inprogress
\newif\ifPS@substeps
\newif\ifPS@subsubsteps
\newif\ifPS@continued
\newif\ifPS@subcontinued
\newcounter{PSc}
\newcounter{PSstepc}[PSc]
\newcounter{PSsubstepc}[PSstepc]
\renewcommand{\thePSsubstepc}{\thePSstepc.\arabic{PSsubstepc}}
\newcounter{PSsubsubstepc}[PSsubstepc]

\newenvironment{proofsteps}[1]{%
  \global\settowidth{\PS@lastparam}{\PS@style\hspace*{#1}}
  \ifPS@continued\else\refstepcounter{PSc}\fi
  \begingroup
  \setlength{\LTpre}{\PSpre}%
  \setlength{\LTpost}{\PSpost}%
  
  \setlength{\tabcolsep}{0pt}
  \noindent\PS@style
  \settowidth{\PS@Nwidth}{\PS@numstyle\PS@numwidth}%
  \setlength{\PS@Swidth}{#1}%
  \addtolength{\PS@Swidth}{-\PS@extraindent}%
  \setlength{\PS@Ewidth}{\linewidth}%
  \addtolength{\PS@Ewidth}{-\PSindent}%
  \addtolength{\PS@Ewidth}{-\PS@extraindent}%
  \addtolength{\PS@Ewidth}{-\PS@Nwidth}%
  \addtolength{\PS@Ewidth}{-\PScolsep}%
  \addtolength{\PS@Ewidth}{-\PS@Swidth}%
  \addtolength{\PS@Ewidth}{-\PScolsep}%
  \PS@inprogresstrue
  \longtable{%
    @{\hspace*{\PSindent}\hspace*{\PS@extraindent}\makebox[\PS@Nwidth][r]{\PS@rownumber}}%
    @{\hskip\PScolsep}>{\PS@step}p{\PS@Swidth}%
    @{\hskip\PScolsep}>{\footnotesize\raggedright\arraybackslash}p{\PS@Ewidth}%
  }%
}{%
  \ifPS@inprogress
  \addtocounter{table}{-1}%
  \endlongtable  
  \endgroup
  \PS@continuedfalse
  \PS@inprogressfalse
  \else\fi
}

\newcommand{\PSbreak}[1]{%
  \endproofsteps
  \par\medskip
  #1
  \medskip\par
  \PS@continuedtrue
  \proofsteps{\PS@lastparam}%
}

\newif\ifPS@sub@inprogress

\newif\ifPS@laststep
\newcommand{\laststep}{\global\PS@laststeptrue}
\newif\ifPS@lastsubstep
\newcommand{\lastsubstep}{\global\PS@lastsubsteptrue}

\newcommand{\adjustcol}[1]{%
  \global\advance\@colroom-#1%
}

\makeatother

\usepackage{tcolorbox} 
\usepackage{tablefootnote}
\usepackage[all]{xy}
\newcommand{\acta}{\mathfrak{a}}
\newcommand{\T}{{\mathsf{T}}}
\newcommand{\gene}{{\mathsf{gene}}}
\newcommand{\Dex}{\mathcal{D}}
\newcommand{\id}{\mathrm{id}}
\newcommand{\CAT}{\mathtt{CAT}}
\newcommand{\Init}{\mathsf{Init}}
\newcommand{\Cut}{\mathsf{Cut}}
\newcommand{\Modify}{\mathsf{Modify}}
\newcommand{\Atom}{\mathsf{Atom}}
\newcommand{\Ind}{\mathsf{Ind}}
\newcommand{\KEL}{\mathsf{KEL}}
\newcommand{\proofrulet}[2]{\genfrac{}{}{0pt}{0}{#1}{#2}}
\newcommand{\OL}[1]{\overline{#1}}
\newcommand{\UL}[1]{\underline{#1}}
\newcommand{\Iy}{\mathsf{I}}
\newcommand{\Ry}{\mathsf{R}}
\newcommand{\Sy}{\mathsf{S}}
\newcommand{\Ty}{\mathsf{T}}
\newcommand{\Fy}{\mathsf{F}}
\newcommand{\Py}{\mathsf{P}}
\newcommand{\INT}{{\mathsf{INT}}}
\newcommand{\Nat}{\mathtt{Nat}}
\newcommand{\Cod}{\mathrm{Cod}}
\newcommand{\BS}[0]{\boldsymbol}
\title{Induction rules for Transition Algebra}
\author{Go Hashimoto}%
{Kyushu University, Japan}%
{hashimoto.go.427@s.kyushu-u.ac.jp}%
{}%
{}
\authorrunning{G.~Hashimoto} %
\Copyright{G.~Hashimoto} %
\begin{document}

\keywords{
institutional model theory,
algebraic specification,
transition algebra,
joint consistency,
interpolation,
Kleene algebra,
induction,
completeness
}

\maketitle
\begin{abstract}
Transition Algebra ($\mathsf{TA}$) is a type of infinite logic introduced to discuss rewriting systems.
The natural deductive proof systems already introduced in $\mathsf{TA}$ satisfy completeness for countable signatures.
However, it lacks compactness, making it unsuitable for practical applications.
This time, we will create a compact proof system by restricting the $*$ rule to induction.
We will also use a sequent proof system instead of natural deduction.
Furthermore, we will introduce a semantics that makes this system complete, and its application (a model-theoretic proof of Craig interpolation).
\end{abstract}

\section{Introduction}
\subparagraph*{$\TA$}
$\TA$~\cite{go-icalp24} is an extension of first-order logic introduced for algebraic specification languages.
$\TA$ is based on institution theory and possesses the modularity of the logic that forms the basis of CafeOBJ~\cite{dia-caf}.
On the other hand, the introduction of actions gives it rich expressiveness, similar to the rewriting logic that forms the basis of Maude system~\cite{conf/maude/2007}.
\subparagraph*{Induction and its Semantics}
In \cite{go-icalp24}, a proof system for TA was introduced and its completeness was demonstrated using an extended forcing technique.
However, since $\TA$ is not compact, completeness does not hold true for ordinary proof systems, and a rule with infinite assumptions must be allowed.
That is unusable in practical applications.
\begin{center}\small
$(Star_E)~\proofrule{\Gamma\vdash_\Sigma t_{0}\stackrel{\acta^*}\Longrightarrow t_{1} ~~~ \Gamma\cup\{t_{0}\stackrel{\acta^n}\Longrightarrow t_{1}\}\vdash_\Sigma \phi\text{ for all } n \in \omega}{\Gamma\vdash_\Sigma \phi}$
\end{center}
A simple alternative is to weaken the system by introducing induction, as follows:%
\footnote{
In reality, it will be introduced in a slightly different way than shown below in order to correspond with Kleene algebras.
Then since the applicability of induction depends on the existence of actions, we also introduce the quantification of actions so that we can discuss these sufficiently.
}
\begin{center}\small
$(Ind)~\proofrule{\Gamma\vdash_\Sigma t_{0}\stackrel{\acta^*}\Longrightarrow t_{1} ~~~ \Gamma\cup\{x\stackrel{\acta}\Longrightarrow y\}\cup\{\phi(x)\}\vdash_{\Sigma[x,y]}\phi(y)}{\Gamma\cup\{\phi(t_{0})\}\vdash_{\Sigma} \phi(t_{1})}$
\end{center}
Of course, the resulting system is compact and therefore does not satisfy completeness.
This is inconvenient when using model-theoretic methods, so we introduce a Henkin semantic interpretation using Kleene algebras and show the completeness of the system. 
\subparagraph*{Interpolation}
As an application of completeness, we will prove the interpolation theorem (and joint consistency).
TA is, in a sense, strong enough to uniquely identify a model by its theory, and
properties such as joint consistency, which combines two theories, no longer hold true \cite{hashimoto2025}.
On the other hand, the system based on induction (and corresponding semantics) satisfy this property.
The proof is based on the extended forcing used in \cite{gaina2026forcinginterpolationfirstorderhybrid}, but to save space, it is not explicitly introduced and it takes a form similar to Henkin method.
Our result is a significant extension of the classical setting in the following sense:
\begin{itemize}\small
\item While classical results are based on signature inclusion, our results are based on signature morphism push-out.
Interpolation based on push-out is naturally consistent with the literature on algebraic specifications and is widely adopted in that context~\cite{Tarlecki2024Fragility}~\cite{Diaconescu2025}.
\item Results for many sorted logic.
In many sorted systems, these theorems do not hold true depending on the conditions of sorts.
This paper provides almost necessary and sufficient conditions for sorts.
\end{itemize}
\subparagraph*{Else}
This time, we add some action constructors. These are useful for writing Maude's strategy language~\cite{EKER2023100887} and so on.
We also expand the concept of substitution, that improve usability of translation.

\section{Transition Algebra}\label{sec:ta}

\begin{definition}[Category of signatures]
A signature $\Sigma$ is a tuple of disjoint components $\Sigma=\pos{S,F,L}$ where:
\begin{itemize}
\item \(\pos{S, F}\) is a many-sorted algebraic signature consisting of a set \(S\) of \emph{sorts} and a disjoint family \(F = \{ F_{\BS{s}\to s} \mid \BS{s} \in S^{*}, s \in S \}\) of \emph{function symbols}, and
\item \(L\) is a disjoint sets $\{L_{\BS{s}}\}_{\BS{s}\in S^{=}}\,(S^{=}:=\{\pos{s,s}\mid s\in S\})$ of symbols.%
\footnote{This time, we will not interpret the labels as $S$-sorted relations because we are quantifying over them.}
We call that elements \emph{transition labels}.
\end{itemize}
A signature morphism $\chi:\pos{S,F,L}\to\pos{S',F',L'}$
is a many-sorted algebraic signature morphism $\chi:\pos{S,F}\to\pos{S',F'}$ equipped with a label translation $L\to L'$.
Signature morphisms compose componentwise.
Their composition has identities and is associative, thus leading to a category $\Sig$ of signatures.
(Since the components are disjoint, it is acceptable to simply write them as $\Sigma_{\BS{s}\to s}$ or $\chi(\lambda)$ instead of using the symbols $S, F, L$.
Conversely, to emphasize the components, we may use the expressions $\Sigma=\pos{S_{\Sigma},F_{\Sigma},L_{\Sigma}}$, $\chi=\pos{S_{\chi},F_{\chi},L_{\chi}}$.)
\end{definition}

\begin{definition}[Terms and Actions]
We define a functor $\T:\Sig\to\Sig$ as follows:
\ \\
For $\Sigma=\pos{S_{\Sigma},F_{\Sigma},L_{\Sigma}}\in|\Sigma|$, we define $\T(\Sigma)=\pos{S_{\T(\Sigma)},F_{\T(\Sigma)},L_{\T(\Sigma)}}$ as follows:
\begin{itemize}
\item The elements of $S_{\T(\Sigma)}$ are defined by
$
\proofrule{s\in S_{\Sigma}}{"s"(=\pos{\mathrm{sort},s})\in S_{\T(\Sigma)}}
$.
\\ \
\item The elements of $F_{\T(\Sigma)}$ (called $\Sigma$-terms) are defined recursively 
(where $*$ is a join of strings.)\\
\raisebox{-5mm}{
$
\proofrule{\sigma\in F_{\BS{s}\to s}}{"\sigma"(=\pos{\mathrm{map},\sigma})\in{F_{\T(\Sigma)}}_{\BS{s}\to s}}
\
\proofrule{\pos{\mathrm{map},\sigma}\in {F_{\T(\Sigma)}}_{\BS{s}\to s}\ t_{0}\in {F_{\T(\Sigma)}}_{\BS{s}_{0}}\ t_{1}\in{F_{\T(\Sigma)}}_{\BS{s}_{1}}\ \cdots\ t_{n-1}\in{F_{\T(\Sigma)}}_{\BS{s}_{n-1}}}
{\sigma(t_{0},t_{1},\dots,t_{n-1})\,(=\pos{\mathrm{map},\sigma}*\pos{t_{0},t_{1},\dots,t_{n-1}})\in{F_{\T(\Sigma)}}_{s}}
$
}
\\ \
\item The elements of $L_{\T(\Sigma)}$ (called $\Sigma$-action) are defined recursively as follows:%
\footnote{
There are new operations.
$\impact$ has meaning like $\acta_{0}^{c}\cupact\acta_{1}$,
and $\resi$ has meaning like $\acta_{0}^{-1}\comp\acta_{1}$.
However, $\empact, \cupact, \impact, \resi$ are not essential to the discussion.
The semantics introduced here can be discussed even without these operations, with slight modifications.
The interpolation theorem also holds true without these operations.
Conversely, we can create actions "$\{\pos{t_{0},t_{1}}\}$" from terms, and add $\cup_{x}\acta$ to quantify the first-order variables within those terms.
}
\begin{center}
\(
\lambda\mid
\empact_{s}\mid\acta\cupact\acta\mid\acta\impact\acta\mid
\one_{s}\mid\acta\comp\acta\mid\acta\resi\acta\mid\acta^{*}
\)
\end{center}
We may
write $\acta^{\Coo}$ instead of $\acta\impact\empact$,
$\acta^{-1}$ instead of $\acta\resi\one$, and
$\acta^{n}$ instead of $\one\,(n=0)$, and $\acta^{n-1}\comp\acta\,(n\geq 1)$.
Strictly speaking, it is defined as follows:
\begin{align*}\hspace*{-0.5cm}
\proofrule{\lambda\in L_{\BS{s}}}{"\lambda"(=\pos{\mathrm{label},\lambda})\in{L_{\T(\Sigma)}}_{\BS{s}}}
\
\proofrule{\rho\in\{\empact,\one\}\,\,\,\,s\in S_{\T(\Sigma)}}{\rho_{s}\in{L_{\T(\Sigma)}}_{ss}}
\
\proofrule{\rho\in\{*\}\,\,\,\,\acta\in{L_{\T(\Sigma)}}_{\BS{s}}}{\acta^{\rho}\in{L_{\T(\Sigma)}}_{\BS{s}}}
\
\proofrule{\rho\in\{\cupact,\impact,\comp,\resi\}\,\,\,\,\acta_{0},\acta_{1}\in{L_{\T(\Sigma)}}_{\BS{s}}}{\acta_{0}\rho\acta_{1}\in{L_{\T(\Sigma)}}_{\BS{s}}}
\end{align*}
where
$\rho_{s}=\pos{\mathrm{op},\rho,s}$,
$\acta^{\rho}=\pos{\mathrm{op},\rho,\acta}$, and
$\acta_{0}\rho\acta_{1}=\pos{\mathrm{op},\rho,\acta_{0},\acta_{1}}$.
\end{itemize}
For $\chi:\Sigma\to\Sigma'$,
we define $\T(\chi):\T(\Sigma)\to\T(\Sigma')$
as follows:
\begin{itemize}
\item
$S_{\T(\chi)}(\pos{\mathrm{sort},s}):=\pos{\mathrm{sort},S_{\chi}(s)}$
\item
$F_{\T(\chi)}$ is defined recursively.
\(F_{\T(\chi)}(\sigma(t_{0},t_{1},\dots,t_{n-1})):=F_{\chi}(\sigma)(F_{\T(\chi)}(t_{0}),F_{\T(\chi)}(t_{1}),\dots,F_{\T(\chi)}(t_{n-1}))\)
\item
$L_{\T(\chi)}$ is defined recursively.
\(
L_{\T(\chi)}(\acta):=
\begin{cases}
\pos{\mathrm{label},L_{\chi}(\lambda)}&(\acta=\pos{\mathrm{label},\lambda})\\
\rho_{S_{\T(\chi)}(s)}&(\acta=\rho_{s},\,\rho\in\{\empact,\one\})\\
L_{\T(\chi)}(\acta_{0})^{\rho}&(\acta=\acta_{0}^{\rho},\,\rho\in\{*\})\\
L_{\T(\chi)}(\acta_{0})\rho L_{\T(\chi)}(\acta_{1})&(\acta=\acta_{0}\rho\acta_{1},\,\rho\in\{\cupact,\impact,\comp,\resi\})
\end{cases}
\)
\end{itemize}
\end{definition}

\begin{remark}
$\eta:=\{\pos{\pos{\mathrm{sort},\cdot},\pos{\mathrm{map},\cdot},\pos{\mathrm{label},\cdot}}:\Sigma\to\T(\Sigma)\}_{\Sigma\in|\Sig|}$ is a natural transformation.
Also, for $\chi:\Sigma\to\T(\Sigma')$, if we define $\T_{\chi}:\T(\Sigma)\to\T(\Sigma')$ as follows, then $\pos{\T,\eta,\T_{(\cdot)}}$ is a Kleisli triple.
\begin{itemize}
\item
$S_{\T_{\chi}}(\pos{\mathrm{sort},s}):=S_{\chi}(s)$
\item
\(F_{\T_{\chi}}(\sigma(t_{0},t_{1},\dots,t_{n-1})):=F_{\chi}(\sigma)*\pos{F_{\T_{\chi}}(t_{0}),F_{\T_{\chi}}(t_{1}),\dots,F_{\T_{\chi}}(t_{n-1})}\)
\item
\(L_{\T_{\chi}}(\acta):=\begin{cases}
L_{\chi}(\lambda)&(\acta=\pos{\mathrm{label},\lambda})\\
\rho_{S_{\T_{\chi}}(s)}&(\acta=\rho_{s},\,\rho\in\{\empact,\one\})\\
L_{\T_{\chi}}(\acta_{0})^{\rho}&(\acta=\acta_{0}^{\rho},\,\rho\in\{*\})\\
L_{\T_{\chi}}(\acta_{0})\rho L_{\T_{\chi}}(\acta_{1})&(\acta=\acta_{0}\rho\acta_{1},\,\rho\in\{\cupact,\impact,\comp,\resi\})
\end{cases}\)
\end{itemize}
If there is no risk of confusion, $\T_{\chi}$ is also represented as $\T(\chi)$. Or, more simply, it is written as $\chi$.
Symbols such as $\mathrm{sort}$, $\mathrm{map}$, $\mathrm{label}$, and $\mathrm{op}$ may also be omitted.
\end{remark}

\begin{definition}[Generalized signature morphisms]
The Kleisli category constructed from the above triples is denoted by $\Sig^{\gene}$, and its morphisms are called generalized signature morphisms.
Specifically, the objects are $|\Sig^{\gene}|=|\Sig|$, and the morphisms $\chi:\Sigma\to\Sigma'\in\Sig^{\gene}$ satisfy the following:
\begin{itemize}
\item $S_{\chi}:S_{\Sigma}\to S_{\T(\Sigma')}$ is a function,
\item $\{{F_{\chi}}_{\BS{s}\to s}:{F_{\Sigma}}_{\BS{s}\to s}\to{F_{\T(\Sigma')}}_{\chi(\BS{s})\to\chi(s)}\}_{\pos{\BS{s},s}\in S_{\Sigma}^{*}\times S_{\Sigma}}$ is a $S_{\Sigma}^{*}\times S_{\Sigma}$-sorted function, and
\item $\{{L_{\chi}}_{\BS{s}}:{L_{\Sigma}}_{\BS{s}}\to{L_{\T(\Sigma')}}_{\chi(\BS{s})}\}_{\BS{s}\in S_{\Sigma}^{=}}$ is a $S_{\Sigma}^{=}$-sorted function.
\end{itemize}
The difference from signature morphism is that the bottom two targets have been changed to composite ones.
The operations $\Dex$, $\Mod$, and $\Sen$ introduced later differ slightly between $\Sig$ and $\Sig^{\gene}$, but we will proceed without separating the notation.
When there is a risk of confusion, we may also write $\Dex^{\gene}$, $\Mod^{\gene}$, and $\Sen^{\gene}$.
\end{definition}

\begin{definition}[Variables]
Given a signature $\Sigma\in|\Sig|$.
A $\Sigma$-variable is a tuple $\pos{\mathrm{var},n,\Sigma}$.
Here, $n$ is an element of $\omega$ and is called the variable name.
(Including $\Sigma$ and $\mathrm{var}$ in the tuple is to treat them as new elements that are not present in $\Sigma$ or $\T(\Sigma)$.)
{\par}
$X$ is called a $\Sigma$ variable block if $X$ is a many-sorted disjoint sets $\{X_{s}\}_{s\in S\sqcup S^{=}}$ of variables.
In this case, we consider only the case where $X$ is finite, and this is simply called a block. The set of $\Sigma$ blocks is denoted by $\Sigma_{\Dex}$.
Any $\chi:\Sigma\to\Sigma'\in\Sig\,(\Sig^{\gene})$ defines a translation function $\chi_{\Dex}:\Sigma_{\Dex}\to\Sigma'_{\Dex}$ by
$\chi_{\Dex}(\{\cdots\pos{\mathrm{var},n,\Sigma}\cdots\}):=\{\cdots\pos{\mathrm{var},n,\Sigma'}\cdots\}$.
\end{definition}

\begin{definition}[Extension using variables]
Let $\Sigma=\pos{S,F,L}\in|\Sig|$, and $X\in\Sigma_{\Dex}$.
We define $\Sigma^{\Dex}(X):\Sigma\to\Sigma^{\Dex}[X]$ as follows:
\begin{itemize}
\item $\Sigma^{\Dex}[X]:=\pos{S,F\cup \{X_{s}\}_{s\in S},L\cup\{X_{\BS{s}}\}_{\BS{s}\in S^{=}}}$
\item $\Sigma^{\Dex}(X):\Sigma\to\Sigma^{\Dex}[X]$ is the unique inclusion.
(For $\gene$, $\Sigma^{\Dex^{\gene}}(X):=\eta_{\Sigma^{\Dex}[X]}\circ\Sigma^{\Dex}(X)$.)
\end{itemize}
For any $\chi:\Sigma\to\Sigma'\in\Sig\,(\Sig^{\gene})$, $X\in\Sigma_{\Dex}$,
we write $\chi^{\Dex}[X]:\Sigma^{\Dex}[X]\to\Sigma'^{\Dex}[X']\in\Sig\,(\Sig^{\gene})$ as $\chi$ with the correspondence determined by $x(=\pos{\mathrm{var},n,\Sigma})\mapsto x'(=\pos{\mathrm{var},n,\Sigma'})$ added to it.
(For $\gene$, $\chi^{\Dex^{\gene}}[X]:=\eta_{\Sigma^{\Dex}[X]}\circ\chi^{\Dex}[X]$.)
The following commutativity holds true:
\begin{center}
\(\xymatrix@R=20pt@C=70pt{ 
\Sigma^{\Dex}[X] \ar[r]|{\chi^{\Dex}[X]} & \Sigma'^{\Dex}[X']\\
\Sigma \ar[u]|{\Sigma^{\Dex}(X)} \ar[r]|{\chi} & \Sigma' \ar[u]|{\Sigma'^{\Dex}(X')}
}\)
\end{center}
\end{definition}

\begin{definition}[Model functor]\label{def:Model-functor}
Given a signature \(\Sigma\), a \emph{\(\Sigma\)-model} \(\mathfrak{A}\) is
an \(\pos{S, F}\)-algebra \(\mathfrak{A}\)
that interprets every label \(\lambda \in L_{ss}\) as a relation $\lambda^{\A}\subseteq\A_{s}\times\A_{s}$.
These also define the interpretation of terms and actions.
\begin{itemize}
\item
$\A_{\pos{\mathrm{sort},s}}:=\A_{s}$
\item
Terms are interpreted as
$
\sigma(t_{0},t_{1},\dots,t_{n-1})^{\A}:=\sigma^{\A}(t_{0}^{\A},t_{1}^{\A},\dots,t_{n-1}^{\A})
$.
\item
Actions are interpreted as binary transition relations in models.
Given a model $\A$ over a signature $\Sigma$, $s\in S$, and actions $\acta,\acta_1,\acta_2\in{L_{\T(\Sigma)}}_{ss}$, we have:
\begin{align*}
&\pos{label,\lambda}^{\A}:=\lambda^{\A}&
&\empact_{s}^{\A}:=\emptyset_{\A_{s}}& &(\acta_{0}\cupact\acta_{1})^{\A}:=\acta_{0}^{\A}\cup\acta_{1}^{\A}& &(\acta_{0}\impact\acta_{1})^{\A}:={\acta_{0}^{\A}}^{c}\cup\acta_{1}^{\A}& \\
&(\acta^*)^{\A}:=(\acta^{\A})^{*}& &\one_{s}^{\A}:=\id_{\A_{s}}& &(\acta_{0}\comp\acta_{1})^{\A}:=\acta_{0}^{\A}\comp\acta_{1}^{\A}& &(\acta_{0}\resi\acta_{1})^{\A}:={\acta_{0}^{\A}}^{-1}\comp\acta_{1}^{\A}&
\end{align*}
\end{itemize}
\emph{Homomorphism} would only complicate things, so this time we will only consider $\id_{\A}:\A\to\A\,(=\{\id_{\A_{s}}\}_{s\in S})$.
The category of the $\Sigma$ model is denoted by $\Mod(\Sigma)$.
{\par}
Every \(\chi:\Sigma\to\Sigma'\in\Sig\,(\Sig^{\gene})\) determines a \emph{model-reduct functor} $\Mod(\chi):\Mod(\Sigma')\to\Mod(\Sigma)\,(\A'\mapsto\A'\red_{\chi})$ such that:
\begin{itemize}
\item for every \(\Sigma'\)-model \(\A'\),
\((\A'\red_\chi)_s=\A'_{\chi(s)}\) for each sort \(s \in S\),
\(\sigma^{\A'\red_{\chi}} = \chi(\sigma)^{\A'}\) for each symbol \(\sigma \in F\), and
\(\lambda^{\A'\red_{\chi}} = \chi(\lambda)^{\A'}\)
for each label \(\lambda \in L\);%
\footnote{
Although the notation is the same, note that in the case of $\chi\in\Sig^{\gene}$, $\chi(\sigma)\in\T(\Sigma')$, so $\chi(\sigma)^{\A'}$ is also written in the sense of a term.
The same applies to labels.
}
and
\item for every \(\Sigma'\)-homomorphism \(h' \colon \A' \to \B'\), \((h' \red_{\chi})_{s} = h'_{\chi(s)}\) for each \(s \in S\).
\end{itemize}
The mapping \(\Mod:\Sig^{op}\,({\Sig^{\gene}}^{op})\to\CAT\) is a functor.
\end{definition}
For any signature morphism $\chi:\Sigma\to\Sigma'$, any $\Sigma$-model $\A$ and any $\Sigma'$-model $\A'$ if $\A=\A'\red_\chi$, we say that $\A$ is the \emph{$\chi$-reduct} of $\A'$, and that $\A'$ is a \emph{$\chi$-expansion} of  $\A$.
For example, let $\Sigma\in|\Sig|$, and $X\in\Sigma_{\Dex}$.
Then an expansion of a $\Sigma$-model $\A$ along $\Sigma^{\Dex}(X):\Sigma\to\Sigma^{\Dex}[X]$ can be seen as a pair $\pos{\A,g:X\to \A}$, where $g$ is a valuation of $X$ in $\A$.

\begin{definition}[Sentence functor]
The set $\Sen(\Sigma)$ of \emph{sentences} over $\Sigma$ is given by the following grammar:
\begin{center}
$\phi::=t_{0}=t_{1}\mid t_{0}\stackrel{\acta}{\Rightarrow}t_{1}\mid\phi\to\phi\mid\vee\Phi\mid\wedge\Phi\mid\Exists{X}\phi'\mid\Forall{X}\phi'$
\end{center}
where
(\emph{a})~$t_{0}$ and $t_{1}$ are \(\Sigma\)-terms of the same sort;
(\emph{b})~$\acta$ is a $\Sigma$-action;
(\emph{c})~$\Phi$ is a finite sequence of $\Sigma$-sentences; and
(\emph{d})~$X$ is a block for $\Sigma$ and $\phi'$ is a $\Sigma[X]$-sentence.
{\par}
The collection of only the equations and action sentences with the label $\lambda\in L_{\Sigma}$ is denoted by $\Sen^{b}(\Sigma)$.
When \(\Phi = \pos{}\), we may write \(\bot\,(\top)\) instead of \(\vee\Phi\,(\wedge\Phi)\).
When \(\Phi = \pos{\phi_{0}, \phi_{1}, \dotsc, \phi_{n-1}}\,(n\geq1)\), we may write \(\phi_{0}\vee\phi_{1}\vee\dotsb\vee\phi_{n-1}\,(\phi_{0}\wedge\phi_{1}\wedge\dotsb\wedge\phi_{n-1})\) instead of \(\vee\Phi\,(\wedge\Phi)\).
Also we may use
$\neg\phi$ instead of $\phi\to\bot$, and
$\phi\leftrightarrow\phi'$ instead of $(\phi\to\phi')\wedge(\phi'\to\phi)$.
{\par}
Any $\chi\colon\Sigma\to\Sigma'\in\Sig\,(\Sig^{\gene})$ can be canonically extended to a \emph{sentence-translation function} $\Sen(\chi):\Sen(\Sigma)\to\Sen(\Sigma')$ (or simply $\chi:\Sen(\Sigma)\to\Sen(\Sigma')$) given by:%
\footnote{
$\chi(t),\chi(\acta)$ means $\T(\chi)(t),\T(\chi)(\acta)$ when $\chi$ is treated as $\chi\in\Sig$, and $\T_{\chi}(t),\T_{\chi}(\acta)$ when $\chi\in\Sig^{\gene}$. 
}
\begin{align*}
&\chi(t_{0}=t_{1}):=(\chi(t_{0})=\chi(t_{1}))&
&\chi(t_{0}\stackrel{\acta}{\Rightarrow}t_{1}):=\chi(t_{0})\stackrel{\chi(\acta)}{\Longrightarrow}\chi(t_{1})&\\
&\chi(\phi\to\phi'):=\chi(\phi)\to\chi(\phi')& &&\\
&\chi(\vee\Phi):=\vee\chi(\Phi)&
&\chi(\wedge\Phi):=\wedge\chi(\Phi)&\\
&\chi(\Exists{X}\phi'):=\Exists{\chi_{\Dex}(X)}\chi^{\Dex}[X](\phi')&
&\chi(\Forall{X}\phi'):=\Forall{\chi_{\Dex}(X)}\chi^{\Dex}[X](\phi')&
\end{align*}
$\Sen:\Sig\,(\Sig^{\gene})\to\Set$ is a functor.
\end{definition}
For the sake of simplicity, we may identify variables only by their name, provided that there is no danger of confusion.
Using this convention, each inclusion morphism $\iota\colon\Sigma\hookrightarrow\Sigma'$ determines an inclusion function $\iota\colon\Sen(\Sigma)\hookrightarrow\Sen(\Sigma')$.

\begin{definition}[Satisfaction relation]\label{def:Satisfaction-relation}
We write $d_{0}\stackrel{\pi}{\Longrightarrow}d_{1}$ as the condition that $\pos{d_{0},d_{1}}\in\A_{s}\times\A_{s}$ satisfies the relation $\pi$.
We define the \emph{satisfaction relation} between models and sentences as follows:
\begin{itemize}
\item $\A\models t_{0}=t_{1}$ iff $t_{0}^{\A}=t_{1}^{\A}$,
\item $\A\models t_{0}\stackrel{\acta}{\Longrightarrow}t_{1}$ iff $t_{0}^{\A}\stackrel{\acta^{\A}}{\Longrightarrow}t_{1}^{\A}$,
\item $\A\models\phi\to\phi'$ iff $\A\models\phi\text{ implies }\A\models\phi'$,
\item $\A\models\vee\phi$ iff $\A\models\phi_{i}$ for some $i\in n$ where $\phi:n\to\Sen(\Sigma)$,
\item $\A\models\wedge\phi$ iff $\A\models\phi_{i}$ for each $i\in n$ where $\phi:n\to\Sen(\Sigma)$,
\item $\A\models\Exists{X}\phi'$ iff $\A'\models\phi'$ for some $\Sigma^{\Dex}(X)$-expansion $\A'$ of $\A$, and
\item $\A\models\Forall{X}\phi'$ iff $\A'\models\phi'$ for each $\Sigma^{\Dex}(X)$-expansion $\A'$ of $\A$.
\end{itemize}
This definition satisfies the Satisfaction condition.
That is, for all signature morphisms $\chi:\Sigma\to\Sigma'$,
all $\Sigma'$-models $\A$ and
all sentences $\phi\in\Sen(\Sigma)$
we have:
$\A\red_{\chi}\models\phi
\quad\text{iff}\quad
\A\models\chi(\phi)$.
\end{definition}
In other words, $\TA$ is an institution.
As we have confirmed, it is also an institution on $\Sig^{\gene}$, but since its properties as a category differ from $\Sig$,
we will still use $\Sig$ for the category of signatures of $\TA$.
The reason for introducing generalized signature morphisms is that substitutions become easier to handle.
\begin{definition}[Substitutions]
Let $\Sigma=\pos{S,F,L}\in|\Sig|$, and $X,Y\in\Sigma_{\Dex}$.
A substitution is defined as a function of $\theta\circ\Sigma^{\Dex}(X)=\Sigma^{\Dex}(Y)$ where $\theta:\Sigma^{\Dex}[X]\to\Sigma^{\Dex}[Y]\in\Sig^{\gene}$.
It is sometimes expressed symbolically as $\theta:X\to Y$. In particular, $\theta:X\to \emptyset$ is called $\Sigma^{\Dex}(X)$-substitution.
The equation $\theta\circ\Sigma^{\Dex}(X)=\id_{\Sigma}$ holds true.
\end{definition}

\section{Sequent calculus}\label{sec:sequent}

\begin{definition}[Sequent]
A sequent on $\TA$ is a family of relations $\vdash=\{\vdash_{\Sigma}\subseteq\mathcal{P}(\Sen(\Sigma))\times\mathcal{P}(\Sen(\Sigma))\}_{\Sigma\in|\Sig|}$ that satisfies the following "structural rules".
When discussing on fragments, we consider it restricted, for example, to $\mathcal{P}(\Sen^{b}(\Sigma))\times\mathcal{P}(\Sen^{b}(\Sigma))$.
\par\noindent 
\begin{center}\small
\textbf{The structural rules}\\
\begin{tabular}{lll}
\\
&$\Modify(\chi)~\proofrule{\,\Gamma\vdash_{\Sigma}\Delta\,}{\,\Gamma'\vdash_{\Sigma'}\Delta'\,}$
$\Space\chi:\Sigma\to\Sigma'\in\Sig^{\gene},\,\,\Gamma'\supseteq\chi(\Gamma),\,\,\Delta'\supseteq\chi(\Delta)$\\
&\\
&$\Init_{\psi}~\proofrule{}{\,\{\psi\}\cup\Gamma\vdash_{\Sigma}\Delta\cup\{\psi\}\,}$
$\Space\Cut^{\psi}~\proofrule{\,\Gamma\vdash_{\Sigma}\Delta\cup\{\psi\}\,\Space\,\{\psi\}\cup\Gamma'\vdash_{\Sigma}\Delta'\,}{\,\Gamma\cup\Gamma'\vdash_{\Sigma}\Delta\cup\Delta'\,}$
\\
\end{tabular}
\end{center}
\medskip
In intuitionistic logic, only one sentence can appear to the right of a sequent, and the three rules are changed from "If the above is true, then the below is true" to "If the above is true (at most one sentence appears to the right in both the top and bottom), then the below is true."
To shorten the notation, $\Gamma\cup\Gamma'$ may be abbreviated as $\Gamma\,\Gamma'$, and $\{\phi\}$ as $\phi$.
Also, $\Gamma\,\Gamma\vdash_{\Sigma}\Delta\,\Delta'$ may be written as $\Gamma'\vdash_{\Gamma\mid\Sigma\mid\Delta}\Delta'$.
\end{definition}

\begin{definition}[Semantic sequent]\label{def:Semantic-sequent}
The semantic sequent $\vDash$ is defined as follows:
\\
\(
\Gamma\vDash\Delta\text{ iff}\)\(\text{ there is no model }\A\text{ such that }\A\models\gamma\text{ for all }\gamma\in\Gamma\text{ and }\A\not\models\gamma\text{ for all }\gamma\in\Delta
\)
\\
In standard notation, $\Gamma\vDash\Delta$ is often used with meanings like $\wedge\Gamma\Rightarrow\wedge\Delta$, but here, in accordance with the term sequent, it is interpreted as $\wedge\Gamma\Rightarrow\vee\Delta$.
\end{definition}

\begin{definition}[Soundness, Completeness, and Compactness]
Let $\vdash$ be the sequent of $\TA$.
$\vdash$ is sound (complete) on $\Sigma\in|\Sig|$ if $\vdash_{\Sigma}\subseteq\vDash_{\Sigma}(\vdash_{\Sigma}\supseteq\vDash_{\Sigma})$ holds true.
$\vdash$ is compact on $\Sigma\in|\Sig|$ if, for any set of sentences $\Gamma,\Delta$,
if $\Gamma\vdash_{\Sigma}\Delta$ holds true, then there exists a finite part
$\iota:\Sigma_{f}\subseteq_{\omega}\Sigma$,
$\Gamma_{f}\subseteq_{\omega}\iota^{-1}(\Gamma)$,
$\Delta_{f}\subseteq_{\omega}\iota^{-1}(\Delta)$,
such thet $\Gamma_{f}\vdash_{\Sigma_{f}}\Delta_{f}$ holds true.
(Here, $A\subseteq_{\alpha}B$ means that $A\subseteq B$ and $\card(A)<\alpha$.)
\end{definition}

\begin{definition}\label{def:basic-proof}
We define the sequent $\vdash^{b}$ for $\Sen^{b}$ as the smallest relation that satisfies the following proof rule.%
\footnote{
Strictly speaking, it is defined as a proof tree.
Also, the $\UL{(\cdot)}$ and $\OL{(\cdot)}$ that appear in the proof rules represent the sentence used and the sentence derived, respectively.
Details will be discussed within the proof of well-definedness.
}
\begin{center}
\textbf{The basic proof rules}\\
\begin{tabular}{l l l l}
\\
$\Init_{\phi}~\proofrule{}{\,\OL{\phi}\vdash_{\Gamma\mid\Sigma\mid\Delta}\OL{\phi}\,}$
&
$\Ry~\proofrule{}{\,\vdash_{\Gamma\mid\Sigma\mid\Delta}\OL{t=t}\,}$ 
&
$\Sy~\proofrule{\,\vdash_{\Gamma\mid\Sigma\mid\Delta}\UL{t=t'}\,}{\,\vdash_{\Gamma\mid\Sigma\mid\Delta}\OL{t'=t}\,}$ 
& 
$\Ty~\proofrule{\,\vdash_{\Gamma\mid\Sigma\mid\Delta}\UL{t=t'}\,\Space\,\vdash_{\Gamma'\mid\Sigma\mid\Delta'}\UL{t'=t''}\,}{\,\vdash_{\Gamma\,\Gamma'\mid\Sigma\mid\Delta\,\Delta'}\OL{t=t''}\,}$\\
&& \\
\multicolumn{4}{l}{
$\Fy~\proofrule{\,\pos{\vdash_{\Gamma\mid\Sigma\mid\Delta}\UL{\BS{t}_{n}=\BS{t}'_{n}}}_{n\in N}\,}{\,\vdash_{\Gamma\mid\Sigma\mid\Delta}\OL{\sigma(\BS{t})=\sigma(\BS{t}')}\,}$ \Space
$\Py~\proofrule{\,\vdash_{\Gamma\mid\Sigma\mid\Delta}\UL{t_{0}=t'_{0}}\,\Space\,\vdash_{\Gamma'\mid\Sigma\mid\Delta'}\UL{t_{1}=t'_{1}}\,\Space\,\vdash_{\Gamma''\mid\Sigma\mid\Delta''}\UL{t_{0}\stackrel{\pi}{\Longrightarrow}t_{1}}\,}
{\,\vdash_{\Gamma\,\Gamma'\,\Gamma''\mid\Sigma\mid\Delta\,\Delta'\,\Delta''}\OL{t'_{0}\stackrel{\pi}{\Longrightarrow}t'_{1}}\,}$}\\
\end{tabular}
\end{center}
In intuitionistic logic, we use the symbol $\vdash^{b,\INT}$ and set $\Delta,\Delta',\Delta''=\emptyset$ for all rules.
\end{definition}

\begin{lemma}[basic soundness, completeness, compactness]\label{lemma:basic-soundness-completeness-compactness}%
Regarding the set of atomic sentences $\Gamma,\Delta\subseteq\Sen^{b}(\Sigma)$,
\(
\Gamma\vdash^{b}_{\Sigma}\Delta\iff\Gamma\vDash_{\Sigma}\Delta
\)
holds true.
In particular, when $\Delta=\{\phi\}$, $\Gamma\vdash^{b,\INT}_{\Sigma}\Delta\iff\Gamma\vdash^{b}_{\Sigma}\Delta\iff\Gamma\vDash_{\Sigma}\Delta$.
(Note that when $\Gamma\vDash_{\Sigma}\Delta$, there exists a sentence $\phi\in\Delta$ that satisfies $\Gamma\vDash_{\Sigma}\phi$, so the proof power of $\vdash^{b}$ and $\vdash^{b,\INT}$ is "the same".)
Furthermore, these are compact for atomic sentences.
\end{lemma}

\begin{definition}[Syntactic sequent]\label{def:Syntactic-sequent}
We define the sequent $\vdash^{*}$ of $\TA$ as the family of smallest relations satisfying the following: %
\par\noindent 
\begin{center}\small
\textbf{The proof rule for atomic sentences}\\
\begin{tabular}{c}
\\
$\Atom\,\proofrule{}{\,\OL{\Gamma_{b}}\vdash_{\Gamma\mid\Sigma\mid\Delta}\OL{\Delta_{b}}\,}\,(\Gamma_{b}\vdash^{b}_{\Sigma}\Delta_{b})$
\end{tabular}
\end{center}
\medskip
\par\noindent 
\begin{center}\small
\textbf{The proof rules for actions}\\
\begin{tabular}{ll}
\\
${\empact}_{L}~\proofrule{}{\,\OL{t_{0}\stackrel{\empact_{s}}{\Longrightarrow}t_{1}}\vdash_{\Gamma\mid\Sigma\mid\Delta}\,}$
&
\\
\\
${\cupact}_{L}~\proofrule{\,\UL{t_{0}\stackrel{\acta_{0}}{\Longrightarrow}t_{1}}\vdash_{\Gamma_{0}\mid\Sigma\mid\Delta_{0}}\,\Space\,\UL{t_{0}\stackrel{\acta_{1}}{\Longrightarrow}t_{1}}\vdash_{\Gamma_{1}\mid\Sigma\mid\Delta_{1}}}{\,\OL{t_{0}\stackrel{\acta_{0}\cupact\acta_{1}}{\Longrightarrow}t_{1}}\vdash_{\Gamma_{0}\,\Gamma_{1}\mid\Sigma\mid\Delta_{0}\,\Delta_{1}}\,}$
&
${\cupact}_{R}^{i}~\proofrule{\,\vdash_{\Gamma\mid\Sigma\mid\Delta}\UL{t_{0}\stackrel{\acta_{i}}{\Longrightarrow}t_{1}}\,}{\,\vdash_{\Gamma\mid\Sigma\mid\Delta}\OL{t_{0}\stackrel{\acta_{0}\cupact\acta_{1}}{\Longrightarrow}t_{1}}\,}\,(i\in2)$
\\
\\
${\impact}_{L}~\proofrule{\,\vdash_{\Gamma_{0}\mid\Sigma\mid\Delta_{0}}\UL{t_{0}\stackrel{\acta_{0}}{\Longrightarrow}t_{1}}\,\Space\,\UL{t_{0}\stackrel{\acta_{1}}{\Longrightarrow}t_{1}}\vdash_{\Gamma_{1}\mid\Sigma\mid\Delta_{1}}}{\,\OL{t_{0}\stackrel{\acta_{0}\impact\acta_{1}}{\Longrightarrow}t_{1}}\vdash_{\Gamma_{0}\,\Gamma_{1}\mid\Sigma\mid\Delta_{0}\,\Delta_{1}}\,}$
&
${\impact}_{R}~\proofrule{\,\UL{t_{0}\stackrel{\acta_{0}}{\Longrightarrow}t_{1}}\vdash_{\Gamma\mid\Sigma\mid\Delta}\UL{t_{0}\stackrel{\acta_{1}}{\Longrightarrow}t_{1}}\,}{\,\vdash_{\Gamma\mid\Sigma\mid\Delta}\OL{t_{0}\stackrel{\acta_{0}\impact\acta_{1}}{\Longrightarrow}t_{1}}\,}$
\\
\\
${\one}_{L}~\proofrule{\,\UL{t_{0}=t_{1}}\vdash_{\Gamma\mid\Sigma\mid\Delta}\,}{\,\OL{t_{0}\stackrel{\one_{s}}{\Longrightarrow}t_{1}}\vdash_{\Gamma\mid\Sigma\mid\Delta}\,}$
&
${\one}_{R}~\proofrule{\,\vdash_{\Gamma\mid\Sigma\mid\Delta}\UL{t_{0}=t_{1}}\,}{\,\vdash_{\Gamma\mid\Sigma\mid\Delta}\OL{t_{0}\stackrel{\one_{s}}{\Longrightarrow}t_{1}}\,}$
\\
\\
${\comp}_{L}~\proofrule{\,\UL{t_{0}\stackrel{\acta_{0}}{\Longrightarrow}x}\,\UL{x\stackrel{\acta_{1}}{\Longrightarrow}t_{1}}\vdash_{\Sigma^{\Dex}(x)(\Gamma)\mid\Sigma^{\Dex}[x]\mid\Sigma^{\Dex}(x)(\Delta)}\,}{\,\OL{t_{0}\stackrel{\acta_{0}\comp\acta_{1}}{\Longrightarrow}t_{1}}\vdash_{\Gamma\mid\Sigma\mid\Delta}}$
&
${\comp}_{R}^{t}~\proofrule{\,\vdash_{\Gamma\mid\Sigma\mid\Delta}\UL{t_{0}\stackrel{\acta_{0}}{\Longrightarrow}t}\,\Space\,\vdash_{\Gamma'\mid\Sigma\mid\Delta'}\UL{t\stackrel{\acta_{1}}{\Longrightarrow}t_{1}}\,}{\,\vdash_{\Gamma\,\Gamma'\mid\Sigma\mid\Delta\,\Delta'}\OL{t_{0}\stackrel{\acta_{0}\comp\acta_{1}}{\Longrightarrow}t_{1}}\,}$
\\
\\
${\resi}_{L}~\proofrule{\,\UL{x\stackrel{\acta_{0}}{\Longrightarrow}t_{0}}\,\UL{x\stackrel{\acta_{1}}{\Longrightarrow}t_{1}}\vdash_{\Sigma^{\Dex}(x)(\Gamma)\mid\Sigma^{\Dex}[x]\mid\Sigma^{\Dex}(x)(\Delta)}\,}{\,\OL{t_{0}\stackrel{\acta_{0}\resi\acta_{1}}{\Longrightarrow}t_{1}}\vdash_{\Gamma\mid\Sigma\mid\Delta}}$
&
${\resi}_{R}^{t}~\proofrule{\,\vdash_{\Gamma\mid\Sigma\mid\Delta}\UL{t\stackrel{\acta_{0}}{\Longrightarrow}t_{0}}\,\Space\,\vdash_{\Gamma'\mid\Sigma\mid\Delta'}\UL{t\stackrel{\acta_{1}}{\Longrightarrow}t_{1}}\,}{\,\vdash_{\Gamma\,\Gamma'\mid\Sigma\mid\Delta\,\Delta'}\OL{t_{0}\stackrel{\acta_{0}\resi\acta_{1}}{\Longrightarrow}t_{1}}\,}$
\\
\\
${*}_{L}~\proofrule{\pos{\,\UL{t_{0}\stackrel{\acta^{n}}{\Longrightarrow}t_{1}}\vdash_{\Gamma_{n}\mid\Sigma\mid\Delta_{n}}}_{n\in\omega}\,}{\,\OL{t_{0}\stackrel{\acta^{*}}{\Longrightarrow}t_{1}}\vdash_{\cup_{n\in\omega}\Gamma_{n}\mid\Sigma\mid\cup_{n\in\omega}\Delta_{n}}\,}$
&
${*}_{R}^{n}~\proofrule{\,\vdash_{\Gamma\mid\Sigma\mid\Delta}\UL{t_{0}\stackrel{\acta^{n}}{\Longrightarrow}t_{1}}\,}{\,\vdash_{\Gamma\mid\Sigma\mid\Delta}\OL{t_{0}\stackrel{\acta^{*}}{\Longrightarrow}t_{1}}\,}\,(n\in\omega)$
\\
\end{tabular}
\end{center}
\medskip
\par\noindent 
\begin{center}\small
\textbf{The proof rules for complex sentences}\\
\begin{tabular}{ll}
\\
${\to}_{L}~\proofrule{\,\vdash_{\Gamma_{0}\mid\Sigma\mid\Delta_{0}}\UL{\phi_{0}}\,\Space\,\UL{\phi_{1}}\vdash_{\Gamma_{1}\mid\Sigma\mid\Delta_{1}}\,}{\,\OL{\phi_{0}\to\phi_{1}}\vdash_{\Gamma_{0}\,\Gamma_{1}\mid\Sigma\mid\Delta_{0}\,\Delta_{1}}\,}$
&
${\to}_{R}~\proofrule{\,\UL{\phi_{0}}\vdash_{\Gamma\mid\Sigma\mid\Delta}\UL{\phi_{1}}\,}{\,\vdash_{\Gamma\mid\Sigma\mid\Delta}\OL{\phi_{0}\to\phi_{1}}\,}$
\\
\\
${\vee}_{L}~\proofrule{\,\pos{\,\UL{\phi_{n}}\vdash_{\Gamma_{n}\mid\Sigma\mid\Delta_{n}}\,}_{n\in N}\,}{\,\OL{\vee\phi}\vdash_{\cup_{n\in N}\Gamma_{n}\mid\Sigma\mid\cup_{n\in N}\Delta_{n}}\,}\,(\phi\in\Sen(\Sigma)^{N})$
&
${\vee}_{R}^{n}~\proofrule{\,\vdash_{\Gamma\mid\Sigma\mid\Delta}\UL{\phi_{n}}\,}{\,\vdash_{\Gamma\mid\Sigma\mid\Delta}\OL{\vee\phi}\,}\,(\phi\in\Sen(\Sigma)^{N},\,n\in N)$
\\
\\
${\wedge}_{L}^{n}~\proofrule{\,\UL{\phi_{n}}\vdash_{\Gamma\mid\Sigma\mid\Delta}\,}{\,\OL{\wedge\phi}\vdash_{\Gamma\mid\Sigma\mid\Delta}\,}\,(\phi\in\Sen(\Sigma)^{N},\,n\in N)$
&
${\wedge}_{R}~\proofrule{\,\pos{\,\vdash_{\Gamma_{n}\mid\Sigma\mid\Delta_{n}}\UL{\phi_{n}}\,}_{n\in N}\,}{\,\vdash_{\cup_{n\in N}\Gamma_{n}\mid\Sigma\mid\cup_{n\in N}\Delta_{n}}\OL{\wedge\phi}\,}\,(\phi\in\Sen(\Sigma)^{N})$
\\
\\
$\exists_{L}~\proofrule{\,\UL{\phi}\vdash_{\Sigma(X)(\Gamma)\mid\Sigma^{\Dex}[X]\mid\Sigma(X)(\Delta)}\,}{\,\OL{\Exists{X}\phi}\vdash_{\Gamma\mid\Sigma\mid\Delta}\,}$
&
$\exists_{R}^{\theta}~\proofrule{\,\vdash_{\Gamma\mid\Sigma\mid\Delta}\UL{\theta(\phi)}\,}{\,\vdash_{\Gamma\mid\Sigma\mid\Delta}\OL{\Exists{X}\phi}\,}\,(\theta\text{ is a }\Sigma^{\Dex}(X)\text{-substitution})$
\\
\\
$\forall_{L}^{\theta}~\proofrule{\,\UL{\theta(\phi)}\vdash_{\Gamma\mid\Sigma\mid\Delta}\,}{\,\OL{\Forall{X}\phi}\vdash_{\Gamma\mid\Sigma\mid\Delta}\,}\,(\theta\text{ is a }\Sigma^{\Dex}(X)\text{-substitution})$
&
$\forall_{R}~\proofrule{\,\vdash_{\Sigma^{\Dex}(X)(\Gamma)\mid\Sigma^{\Dex}[X]\mid\Sigma^{\Dex}(X)(\Delta)}\UL{\phi}\,}{\,\vdash_{\Gamma\mid\Sigma\mid\Delta}\OL{\Forall{X}\phi}\,}$
\\
\end{tabular}
\end{center}
\medskip
In intuitionistic logic,
we use the symbol $\vdash^{*,\INT}$,
limit the number of sentences to the right of a sequent to one, and
change the rule from
"If the above is true, then the below is true" to
"If the above is true (at most one sentence appears to the right of both the above and below), then the below is true."
\end{definition}

\begin{lemma}[Soundness, Completeness]\label{lemma:Soundness-Completeness}
$\vdash^{*}$ is sound.
Let $\vdash^{*-\{\exists{\lambda}\}}$ be the subrelation of $\vdash^{*}$ restricted to $\Sen^{-\{\exists{\lambda}\}}$ (meaning that label quantification is not used at all).
Let $\Sigma\in\Sig$ be a countable signature.
Then, for any $\Gamma,\Delta\subseteq\Sen^{-\{\exists{\lambda}\}}(\Sigma)$,
\(
\Gamma\vdash^{*-\{\exists{\lambda}\}}_{\Sigma}\Delta
\iff
\Gamma\vDash_{\Sigma}\Delta.
\) holds true.
\end{lemma}
Furthermore, by induction on the proof structure, it can be shown that $\vdash^{*-\{\exists{\lambda}\}}$ is "$\omega_{1}$-compact".
On the other hand, $\omega_{1}$-compactness does not hold true for $\vDash^{-\{\exists{\lambda}\}}$, so in general, completeness does not hold true when the signature is uncountable.

\section{Induction rules for $*$}\label{sec:induction}
%
\begin{lemma}\label{lemma:induction-rules}
When $\vdash$ is the relation $\vdash^{*}\,(\vdash^{*,\INT})$, the following holds true, i.e. when the upper part of the line holds, the lower part also holds.%
\[\tiny
\Ind_{R}^{0}~\proofrule{\vdash_{\Gamma\mid\Sigma\mid\Delta}\UL{t_{0}=t_{1}}}{\vdash_{\Gamma\mid\Sigma\mid\Delta}\OL{t_{0}\stackrel{\acta^*}\Longrightarrow t_{1}}}\ \
\Ind_{R}^{+t}~\proofrule{\vdash_{\Gamma\mid\Sigma\mid\Delta}\UL{t_{0}\stackrel{\acta^*}\Longrightarrow t}\ \ \vdash_{\Gamma'\mid\Sigma\mid\Delta'}\UL{t\stackrel{\acta}\Longrightarrow t_{1}}}{\vdash_{\Gamma\,\Gamma'\mid\Sigma\mid\Delta\,\Delta'}\OL{t_{0}\stackrel{\acta^*}\Longrightarrow t_{1}}}\ \
\Ind_{R}^{-t}~\proofrule{\vdash_{\Gamma\mid\Sigma\mid\Delta}\UL{t_{0}\stackrel{\acta}\Longrightarrow t}\ \ \vdash_{\Gamma'\mid\Sigma\mid\Delta'}\UL{t\stackrel{\acta^{*}}\Longrightarrow t_{1}}}{\vdash_{\Gamma\,\Gamma'\mid\Sigma\mid\Delta\,\Delta'}\OL{t_{0}\stackrel{\acta^*}\Longrightarrow t_{1}}}
\]
$
\Ind_{L}^{+t\,\act}~
\proofrule{
\vdash_{\Gamma_{L}\mid\Sigma\mid\Delta_{L}}\UL{t\stackrel{\act}{\Longrightarrow}t_{0}}\Space
\UL{z\stackrel{\act}{\Longrightarrow}x\llap{\phantom{y}}}\,\UL{x\stackrel{\acta}{\Longrightarrow}y}\vdash_{\iota(\Gamma_{M})\mid\Sigma^{\Dex}[X]\mid\iota(\Delta_{M})}\UL{z\stackrel{\act}{\Longrightarrow}y}\Space
\UL{t\stackrel{\act}{\Longrightarrow}t_{1}}\vdash_{\Gamma_{R}\mid\Sigma\mid\Delta_{R}}}
{\OL{t_{0}\stackrel{\acta^*}\Longrightarrow t_{1}}\vdash_{\Gamma\mid\Sigma\mid\Delta}}\\
$$
\Ind_{L}^{-t\,\act}~
\proofrule{
\vdash_{\Gamma_{L}\mid\Sigma\mid\Delta_{L}}\UL{t_{1}\stackrel{\act}{\Longrightarrow}t}\Space
\UL{x\stackrel{\acta}{\Longrightarrow}y}\,\UL{y\stackrel{\act}{\Longrightarrow}z}\vdash_{\iota(\Gamma_{M})\mid\Sigma^{\Dex}[X]\mid\iota(\Delta_{M})}\UL{x\stackrel{\act}{\Longrightarrow}z\llap{\phantom{y}}}\Space
\UL{t_{0}\stackrel{\act}{\Longrightarrow}t}\vdash_{\Gamma_{R}\mid\Sigma\mid\Delta_{R}}}
{\OL{t_{0}\stackrel{\acta^*}\Longrightarrow t_{1}}\vdash_{\Gamma\mid\Sigma\mid\Delta}}\\
\text{where }\pos{\Gamma,\Delta}=\pos{\Gamma_{L}\cup\Gamma_{M}\cup\Gamma_{R},\Delta_{L}\cup\Delta_{M}\cup\Delta_{R}},\ X=\{x,y,z\},\ \iota:=\Sigma^{\Dex}(X)
$
\end{lemma}

\begin{definition}\label{def:sequent-with-ind}
We write $\vdash^{\Ind}$ as the relation obtained
by removing $(*_{L})$ and $(*_{R})$ from $\vdash^{*}$
and adding $(\Ind_{R}^{0})$, $(\Ind_{R}^{+})$, $(\Ind_{R}^{-})$, $(\Ind_{L}^{+})$, $(\Ind_{L}^{-})$ and the following $(\Init_{\acta^{*}})$, $(\Cut^{\acta^{*}})$.%
\footnote{The possibility of elimination of these rules remains unresolved.}
\begin{align*}
&\Init_{\acta^{*}}~\proofrule{}{\OL{t_{0}\stackrel{\acta^{*}}{\Longrightarrow}t_{1}}\vdash_{\Gamma\mid\Sigma\mid\Delta}\OL{t_{0}\stackrel{\acta^{*}}{\Longrightarrow}t_{1}}}&
&\Cut^{\acta^{*}}~\proofrule
{\,\vdash_{\Gamma\mid\Sigma\mid\Delta}\UL{t_{0}\stackrel{\acta^{*}}{\Longrightarrow}t_{1}}\,\Space\,\UL{t_{0}\stackrel{\acta^{*}}{\Longrightarrow}t_{1}}\vdash_{\Gamma'\mid\Sigma\mid\Delta'}\,}
{\,\vdash_{\Gamma\,\Gamma'\mid\Sigma\mid\Delta\,\Delta'}\,}&
\end{align*}
Intuitionistic $\vdash^{\Ind,\INT}$ is defined similarly from $\vdash^{*,\INT}$.
\end{definition}

Induction is a method that utilizes the hypothesis for the previous points.
Therefore, raising the goal to be proven strengthens the inductive hypothesis, which can actually make the proof easier.
This modified goal can be understood as the meaning of $\act$ appearing in the rule.
Below are some claims that can be proven by induction.
(See the proofs for examples of modifying goals.)
\begin{example}\label{example:natural-Ind}
Let
$\Sigma:=\pos{\{\Nat\},\{0:\to\Nat,s:\Nat\to\Nat\},\{suc\}}$,
$\gamma(z)\in\Sen(\Sigma[z])$, and let
$\Gamma:=\{def_{suc},\Exists{\pi}def_{\pi,\gamma}\}\,(def_{suc}:=\Forall{x,y}y\stackrel{suc}{\Longrightarrow}x\leftrightarrow y=s(x),\,def_{\pi,\gamma}:=\Forall{x,y}x\stackrel{\pi}{\Longrightarrow}y\leftrightarrow x\in\mathbb{N}\wedge\gamma(x)\,(t\in\mathbb{N}:=t\stackrel{suc^{*}}{\Longrightarrow}0))$
The following holds true. (In other words, $\vdash^{\Ind}$ includes the usual induction of natural numbers.)
\[\gamma(0)\,\Forall{x\in\mathbb{N}}\gamma(x)\to\gamma(s(x))\vdash^{\Ind,\INT}_{\Gamma\mid\Sigma\mid}\Forall{z\in\mathbb{N}}\gamma(z)\]
\end{example}

Let
$\acta^{\bot}:=(\one\impact(\acta\comp\acta^{-1}))^{\Co}$ and
$\acta^{\top}:=\one\cap(\acta\comp\acta^{-1})$.
(where $\acta_{0}\cap\acta_{1}:=(\acta^{\Co}\cup\acta^{\Co})^{\Co}$ and $\acta_{0}^{-1}:=\acta_{0}\resi\one$.)
$\acta^{\bot}$ returns $a$ itself if there is no $b$ such that $a\stackrel{\acta}{\Longrightarrow}b$ for the input $a$ on the left; otherwise, it returns nothing.
Conversely, $\acta^{\top}$ returns $a$ itself if there is a $b$ such that $a\stackrel{\acta}{\Longrightarrow}b$.
For example, $\acta^{*}\comp\acta^{\bot}$ means repeating $\acta$ until it can no longer be applied.
Other expressions such as if statements, while statements, and some of Maude's strategy language can also be expressed.
Below, we will introduce another expression about paths using the example of two people eating a meal.
\begin{example}\label{example:meal}
Consider the following signature and theory.
\begin{itemize}\small
\item $\Sigma:=\{\{\texttt{member},\texttt{state}\},\{a,b,0:\to\texttt{member},\pos{\cdot,\cdot}:\texttt{member}\,\texttt{member}\to\texttt{state}\},$\\$\phantom{0}\,\,\,\,\,\,\,\,\,\,\,\,\,\,\{get_{a},get_{b},eat_{a},eat_{b}:\to\texttt{state}\,\texttt{state}\}\cup\{[t]:\texttt{state}\,\texttt{state}\mid t\in\T(\Sigma)_{\texttt{state}}\}\}$
\item $\Gamma_{0}:=\{\pos{m,0}{\stackrel{get_{i}}{\Longrightarrow}}\pos{m,i}, \pos{0,m}{\stackrel{get_{i}}{\Longrightarrow}}\pos{i,m}\mid m\in\{a,b,0\}, i\in\{a,b\}\}$ $\cup$\\
$\phantom{0}\,\,\,\,\,\,\,\,\,\,\,\,\,\{\pos{i,i}\stackrel{get_{i}}{\Longrightarrow}\pos{0,0}\mid i\in\{a,b\}\}$
$\cup$ $\{t\stackrel{[t]}{\Longrightarrow}t\mid t\in\T(\Sigma)_{\texttt{state}}\}$
\end{itemize}
$a$ and $b$ represent the people eating, and $\pos{\cdot,\cdot}$ indicates who is holding which chopstick.
If neither person is holding the chopsticks, it is represented as $0$.
$get_{i}$ is the operation where $i$ acquires the chopstick, and $eat_{i}$ is the operation where $i$ eats and puts down the chopsticks.
{\par}
Let $T_{\Sigma,\Gamma_{0}}$ be the initial model of $\Gamma_{0}$. (That is, it does not satisfy conditions outside of $\Gamma_{0}$.)
Let $\Gamma$ be the set of all sentences that $T_{\Sigma,\Gamma_{0}}$ satisfies and do not contain $*$ or a second-order quantifier.
This is a bit rough, but since the purpose here is to discuss $*$, for simplicity we proceed assuming that we already know the information about $\Gamma$.
The following hold true.
\[
\vdash^{\Ind}_{\Gamma\mid\Sigma\mid}\pos{0,0}\stackrel{(\acta\impact\act)^{\Co+\Co}}{===\Rightarrow}\pos{0,0}
\Space
\vdash^{\Ind}_{\Gamma\mid\Sigma\mid}\pos{0,0}\stackrel{(\acta'^{*}\comp{\acta'^{*}}^{\Co})^{\Co}}{====\Rightarrow}\pos{0,0}
\Space
\vdash^{\Ind}_{\Gamma\mid\Sigma\mid}\pos{0,0}\stackrel{\acta^{*}\comp eat_{i}\comp\acta^{*}}{====\Rightarrow}\pos{0,0}\,(i\in\{a,b\})
\]
Here,
$\acta:=get_{a}\cupact get_{b}\cupact eat_{a}\cupact eat_{b}$,
$\acta':=(\acta^{\top}\comp\acta)\cupact(\acta^{\bot}\comp(get_{b}\resi get_{a}))$,
$\act:=(eat_{a}\cupact eat_{b})$,
$\acta_{0}^{+}:=\acta_{0}\comp\acta_{0}^{*}$, which means that the action is applied one or more times.
{\par}
$(\acta\impact\act)^{\Co}$ semantically means $\acta\setminus\act$, that is, applying $\acta$ without going through $\act$.
Since $\Co$ is a negation, the left side means that in order to apply $\acta$ to $\pos{0,0}$ at least once and return to $\pos{0,0}$, it is necessary to pass through $\act$ along the way (one of them must "eat").
The right side means that there is a path to obtain food for both $a$ and $b$.
Unlike temporal logic, $\TA$ cannot discuss paths of infinite length, but it can capture events that occur along paths of finite length.
The middle part means that a deadlock will not occur if, for example, $b$ gives a chopstick to $a$ when they get stuck.
{\par}
These claims themselves are trivial.
However, we might use left introduction of $*$, because $*$ is in a "negative position".
In previous systems, proofs required an infinite number of assumptions.
\end{example}


\begin{lemma}\label{lemma:ind-is-sound-on-ta}
$\vdash^{\Ind}$ is sound (not complete), compact.
\end{lemma}

\section{Kleene algebraic semantics}\label{sec:complete}

\begin{definition}
For $\Sigma=\pos{S,F,L}\in|\Sig|$ and $\A\in|\Mod(\Sigma)|$,
let $\A^{All}_{ss}$ be the set of all relations on $\A_{s}$.
$\pos{\A^{All}_{ss};\empact,\cupact,\impact,\one,\comp,\resi}$ is a relation algebra.
The meaning of the symbols is the same as in the previous interpretation.
\end{definition}

\begin{definition}[Kleene algebraic $\TA$]\label{def:Kleene-algebraic-TA}
The Kleene algebraic $\TA$ (abbreviated as $\TA_{k}$) is defined by changing the part $\Mod$ and $\models$ as follows:
\begin{itemize}
\item[]$[\,\Mod^{\TA_{k}}\,]$
$\A=\pos{\A,\{*_{\BS{s}}^{\A}:\A_{\BS{s}}\to\A_{\BS{s}}\}_{\BS{s}\in S^{=}}}$ is a $\Sigma$-model if
\begin{itemize}
\item $\A\in|\Mod^{\TA}(\Sigma)|$
\item $\A_{\BS{s}}\subseteq\A^{All}_{\BS{s}}$ is a subalgebra (In other words, it is a set that is closed under relation algebraic operations.) and includes all interpretations of $\lambda\in L$.
\item $\{*^{\A}_{\BS{s}}:\A_{\BS{s}}\to\A_{\BS{s}}\}_{\BS{s}\in S^{=}}$%
is a $S^{=}$-sorted function such that $\pos{\A_{\BS{s}};\empact,\cupact,\one,\comp,*^{\A}_{\BS{s}}}\,(\BS{s}\in S^{=})$ is a Kleene algebra.
\end{itemize}
We consider only $\id_{\A}$ as a homomorphism.
For $\chi:\Sigma\to\Sigma'$, $\pos{\A,\{*_{\BS{s}}^{\A}:\A_{\BS{s}}\to\A_{\BS{s}}\}_{\BS{s}\in S'^{=}}}\in|\Mod^{\TA_{k}}(\Sigma')|$, the reduct is
\[
\Mod^{\TA_{k}}(\chi)(\pos{\A,\{*_{\BS{s}}^{\A}:\A_{\BS{s}}\to\A_{\BS{s}}\}_{\BS{s}\in S'^{=}}}):=
\pos{
\Mod^{\TA}(\chi)(\A),
\{*_{\chi(\BS{s})}^{\A}:\A_{\chi(\BS{s})}\to\A_{\chi(\BS{s})}\}_{\BS{s}\in S^{=}}
}.
\]
For $\Sig^{\gene}$, the $\Mod^{\TA}(\chi)(\A)$ part is considered to be the definition for $\gene$.
\item[] $[\,\models^{\TA_{k}}\,]$
We will change the $*$ part of the action interpretation from a transitive closure to interpretation of $*$ equipped in the model.
The rest remains the same.
\end{itemize}
\end{definition}

\begin{lemma}\label{lemma:IndNK}
We use the following abbreviations for actions.
\begin{align*}
&(\acta_{1}\leqq\acta_{2}):=\Forall{x,y}x\stackrel{\acta_{1}}{\Longrightarrow}y\to x\stackrel{\acta_{2}}{\Longrightarrow}y&
&(\acta_{1}\equiv\acta_{2}):=(\acta_{1}\leqq\acta_{2})\wedge(\acta_{2}\leqq\acta_{2})&
\end{align*}
Then the following holds true for $\vdash^{\Ind,\INT}$.
\footnote{
When considering fragments where some of $\empact$, $\cupact$, $\impact$, and $\resi$ are missing,
the rules concerning these does not need to be assumed.
}
\begin{itemize}
\item action sentences
\(\Space t_{0}=t'_{0}\wedge t_{1}=t'_{1}\wedge\acta\equiv\acta'\to(t_{0}\stackrel{\acta}{\Longrightarrow}t_{1}\leftrightarrow t'_{0}\stackrel{\acta'}{\Longrightarrow}t'_{1})\)
\item order and monotonicity
\begin{align*}\small
&\acta\leqq\acta\Space
\acta_{0}\leqq\acta_{1}\wedge\acta_{1}\leqq\acta_{2}\to\acta_{0}\leqq\acta_{1}&\\
&(\acta_{0}\leqq\acta'_{0})\wedge(\acta_{1}\leqq\acta'_{1})\to(\acta_{0}\rho\acta_{1})\leqq(\acta'_{0}\rho\acta'_{1})& &(\rho\in\{\cupact,\comp,\resi\})&\\
&(\acta'_{0}\leqq\acta_{0})\wedge(\acta_{1}\leqq\acta'_{1})\to(\acta_{0}\rho\acta_{1})\leqq(\acta'_{0}\rho\acta'_{1})& &(\rho\in\{\impact\})&\\
&(\acta_{0}\leqq\acta'_{0})\to(\acta_{0}^{\rho})\leqq({\acta'}_{0}^{\rho})& &(\rho\in\{*\})&
\end{align*}
Therefore $(\acta_{0}\equiv^{\Ind(,\INT)}_{\Gamma\mid\Sigma\mid\Delta}\acta_{1}:\Leftrightarrow\,\vdash^{\Ind(,\INT)}_{\Gamma\mid\Sigma\mid\Delta}\acta_{0}\equiv\acta_{1})$ is a congruence on $\pos{L_{\T(\Sigma)},\empact,\cupact,\impact,\one,\comp,\resi,*}$.
\item Kleene algebra
\begin{center}\hspace*{-0.19cm}
$\acta_{1}\comp\acta_{2}\leqq\acta_{1}\to\acta_{1}\comp\acta_{2}^{*}\leqq\acta_{1}\Space
\acta^{*}\comp\acta\leqq\acta^{*}\Space
\one\leqq\acta^{*}\Space
\acta\comp\acta^{*}\leqq\acta^{*}\Space
\acta_{1}\comp\acta_{2}\leqq\acta_{2}\to\acta_{1}^{*}\comp\acta_{2}\leqq\acta_{2}$
\end{center}
Conversely,
by removing $(*_{L})$ and $(*_{R})$ from $\vdash^{*}\,(\vdash^{*,\INT})$ and instead,
adding $(\Init_{\acta^{*}})$, $(\Cut^{\acta^{*}})$, and the rule $(\KEL_{\phi})$ (where $\phi$ is one of the five forms above) which eliminate sentences of the five forms from the left,
we obtain a system equivalent to $\vdash^{\Ind}\,(\vdash^{\Ind,\INT})$.
\end{itemize}
\end{lemma}

\begin{theorem}
\label{thm:Completeness-by-TAk}
For each $\Sigma\in|\Sig|$, and $\Gamma,\Delta\subseteq\Sen(\Sigma)$,
$
\Gamma\vDash^{\TA_{k}}_{\Sigma}\Delta\iff\Gamma\vdash^{\Ind}_{\Sigma}\Delta.
$
\end{theorem}

\section{Interpolation}\label{sec:interpolation}

We will prove Craig's interpolation theorem for the classical part $\vdash^{\Ind}$.
We will not use $\Sig^{\gene}$ in the interpolation theorem. (It would almost certainly not hold true.)

\begin{definition}
Let the following be a commutative diagram of signature morphisms.
\begin{center}
\(\xymatrix@R=20pt@C=70pt{ 
\Sigma_{0}\ar[r]|{\chi'_{0}}&\Sigma'
\\
\ar[u]|{\chi_{0}}\Sigma\ar[r]|{\chi_{1}}&\Sigma_{1}\ar[u]|{\chi'_{1}}
}\)
\end{center}
\begin{enumerate}
\item
The square is said to satisfy the elementary gluing condition if
for any $\A_{i}\in|\Mod^{\TA_{k}}(\Sigma_{i})|\,(i\in2)$, such that $\A_{0}\red_{\chi_{0}}$ and $\A_{1}\red_{\chi_{1}}$ are elementary equivalent,
there exists a $\A'\in|\Mod^{\TA_{k}}(\Sigma')|$ such that $\A'\red_{\chi'_{i}}$ and $\A_{i}$ are elementary equivalent for $i\in2$.
\item
The square is called the $\mathrm{RJC}$-square if
for any satisfiable theories $T_{i}\subseteq\Sen(\Sigma_{i})\,(i\in2)$ such that $\chi_{0}^{-1}(T_{0})=\chi_{1}^{-1}(T_{1})$ and these are complete theories,
$\cup_{i\in2}\chi'_{i}(T_{i})$ is satisfiable.
\item
If for any $\phi\in\Sen(\Sigma)$,
$\Gamma_{i}\not\vDash^{\TA_{k}}_{\Sigma_{i}}\Delta_{i}\,\chi_{i}(\phi)\text{ for each }i\in2$ or
$\chi_{i}(\phi)\,\Gamma_{i}\not\vDash^{\TA_{k}}_{\Sigma_{i}}\Delta_{i}\text{ for each }i\in2$ hold true,
we say that $\pos{\Gamma_{i},\Delta_{i}}_{i\in2}$ are agreeable.
{\par}
If $\cup_{i\in 2}\chi'_{i}(\Gamma_{i})\not\vDash^{\TA_{k}}_{\Sigma'}\cup_{i\in 2}\chi'_{i}(\Delta_{i})$ holds true
for any given agreeable $\pos{\Gamma_{i},\Delta_{i}}_{i\in2}$,
this square is called $\mathrm{CI'}$-square.
\item
If there is $\phi\in\Sen(\Sigma)$ such that $\Gamma_{0}\vdash^{\Ind}_{\Sigma_{0}}\Delta_{0}\,\chi_{0}(\phi)$ and $\chi_{1}(\phi)\,\Gamma_{1}\vdash^{\Ind}_{\Sigma_{1}}\Delta_{1}$ hold true
for each $\Gamma_{i},\Delta_{i}\subseteq\Sen(\Sigma_{i})\,(i\in2)$ such that $\cup_{i\in 2}\chi'_{i}(\Gamma_{i})\vdash^{\Ind}_{\Sigma'}\cup_{i\in 2}\chi'_{i}(\Delta_{i})$,
this square is called $\mathrm{CI}$-square.
\end{enumerate}
\end{definition}

\begin{lemma}\label{lemma:1234}
$1,2$ are equivalent.
$3,4$ are equivalent.
$3,4$ imply $1,2$.
The converse holds true if the signature is countable.
\end{lemma}

\begin{lemma}[exactness]\label{fact:exactness}
Let $\chi_{i}:\Sigma\to\Sigma_{i} \ (i\in2)$ be signature morphisms, $\{\chi'_{i}:\Sigma_{i}\to\Sigma'\}_{i\in2}$ is a colimit of $\{\chi_{i}:\Sigma\to\Sigma_{i}\}_{i\in2}$.
Let $\A_i\in|\Mod^{\TA_{k}}(\Sigma_i)| \ (i\in2)$.
The following are equivalent:
\begin{itemize}
\item $\A_{0}\upharpoonright_{\chi_0}=\A_{1}\upharpoonright_{\chi_1}$
\item there is a $\Sigma'$-model $\A'$ such that $\A'\upharpoonright_{\chi'_0}=\A_{1}$ and $\A'\upharpoonright_{\chi'_1}=\A_{0}$ hold true
\end{itemize}
\end{lemma}

\begin{definition}[disjoint]
Let $f_{i}:S\to S_{i}\,(i\in2)$ be functions.
Let $M_{i}=\{s\in S\mid \text{ there is } s'\neq s \text{ which satisfies } \chi_i(s)=\chi_i(s') \}\ (i\in2)$.
When $M_{0}\cap M_{1}=\emptyset$, we say that $f_{0}$ and $f_{1}$ are disjoint.
Let $\chi_{i}:\Sigma\to\Sigma_{i}\,(i\in2)$ be signature morphisms.
$\chi_0$ and $\chi_1$ are disjoint (about sort confluence) iff $S_{\chi_{0}}$ and $S_{\chi_{1}}$ are disjoint.
\end{definition}

\begin{theorem}\label{theorem:CISQ}
Pushout for disjoint pairs is a $\mathrm{CI}$-square.
Conversely, for a pair of non-disjoint functions $f_{i}:S\to S_{i}\,(i\in2)$, there exists a pair of signature morphisms $\chi_{i}:\Sigma\to\Sigma_{i}\,(i\in2)$ such that
$S_{\chi_{i}}=f_{i}\,(i\in2)$ and the pushout is not $\mathrm{CI}$-square (not even $\mathrm{RJC}$-square).
\end{theorem}

\section{Conclusions and Developments}
Since $\TA$'s proof system is not compact, we need to go through the stage of introducing some kind of subsystem for practical applications.
\subparagraph*{Completeness}
We introduced induction rules and corresponding semantics to demonstrate completeness.
Completeness allows us to understand the properties of proofs through models.
While we dealt with the interpolation theorem as an application,
this semantics could potentially be used for other compactness-based methods, such as ultrapower.
\subparagraph*{Interpolation}
It remains unresolved whether $(\Init_{\acta^{*}})$ and $(\Cut^{\acta^{*}})$ can be removed within a proof system using induction.
\begin{description}\small
\item[$(\Init)$]
For example, in Kleene three-valued logic, in addition to $\pos{T,F}=\pos{+1,-1}$ as truth values, $U=0$ is introduced,
and sentences are interpreted as $\neg\phi=-\phi$, $\vee\phi=\max_{i\in N}(\phi_{i})$, $\wedge\phi=\min_{i\in N}(\phi_{i})$, and $\phi_{0}\to\phi_{1}=\neg\phi_{0}\vee\phi_{1}$.
In this case, $0\to0=0$, so "$(\Init)$" does not necessarily hold true.
Such semantics might negate the possibility of removing this rule.
\item[$(\Cut)$]
In general, the cut-elimination theorem for proofs involving induction does not necessarily hold true unconditionally;
for example, in \cite{MCDOWELL200091}, $\Cut$-elimination is performed by restricting the logic to intuitionistic logic and utilizing the “complexity” of the unique conclusion.
However, since the sign exists in $\vdash^{\Ind}$ here, for example, if we try to remove the $(\Cut)$ of $\acta^{*}$ introduced by $(\Ind_{R}^{\pm})$ and $(\Ind_{L}^{\mp})$,
the $(\Cut)$ for the intermediate goal $\act$ will be duplicated, which risks actually increasing complexity.
\end{description}
Interpolation theorems are usually constructed by directly constructing interpolation using proof-theoretic methods,
but since the $\Cut$-elimination cannot be used as described above, we performed a semantic proof this time.
We obtained conditions that are almost necessary and sufficient for sort.
\\ \ \\
Considering the above circumstances, the interpolation theorem actually provides valuable information for predicting the intermediate goal for the application of $(\Cut)$.
Remember that $(\Ind_{L})$ also involves modifying the goal to $\act$.
As an extension of this discussion, if there is a corresponding interpolation theorem for induction, it could be useful in semi-automating this interactive process.
\subparagraph*{Extension}
This paper dealt with induction on $*$, but by introducing induction on terms and recursively defined actions (for example, the smallest action closed by the transitions of an action and its application to subterms), it becomes possible to apply this to a wider range of problems.
For example, in \cite{hashimoto2025}, the completeness of a constructor-based proof system is demonstrated as an application of the Omitting Types Theorem.
This is similar to the proof system in \cite{go-icalp24}, requiring infinite assumptions, but it is thought that induction and appropriate semantics can be introduced, as in this case.
\bibliographystyle{plainurl}
\bibliography{citab}
\appendix

\section{Proof of the results presented in Section~\ref{sec:ta}}

\begin{proof}[Proof of well-definedness of Definition~\ref{def:Model-functor}]
We will verify that the definition is valid for $\chi\in\Sig^{\gene}$.
First, as preparation, we will verify that $t^{\A\red_{\chi}}=\T_{\chi}(t)^{\A}$ and $\acta^{\A\red_{\chi}}=\T_{\chi}(\acta)^{\A}$.
\begin{proofcases}
\item[$t^{\A\red_{\chi}}=\T_{\chi}(t)^{\A}$]
I will show this using induction on the structure.
Let $t=\pos{\sigma,t_{0},t_{1},\dots,t_{n-1}}$, and $\pos{\sigma'}=\chi(\sigma)$.
\begin{align*}
&\pos{\sigma,t_{0},t_{1},\dots,t_{n-1}}^{\A\red_{\chi}}
=\sigma^{\A\red_{\chi}}(t_{0}^{\A\red_{\chi}},t_{1}^{\A\red_{\chi}},\dots,t_{n-1}^{\A\red_{\chi}})
=\sigma'^{\A}(t_{0}^{\A\red_{\chi}},t_{1}^{\A\red_{\chi}},\dots,t_{n-1}^{\A\red_{\chi}})\\
&\stackrel{IH}{=}\sigma'^{\A}(\T_{\chi}(t_{0})^{\A},\T_{\chi}(t_{1})^{\A},\dots,\T_{\chi}(t_{n-1})^{\A})
=\pos{\sigma',\T_{\chi}(t_{0}),\T_{\chi}(t_{1}),\dots,\T_{\chi}(t_{n-1})}^{\A}\\
&=\chi(\sigma)*\pos{\T_{\chi}(t_{0}),\T_{\chi}(t_{1}),\dots,\T_{\chi}(t_{n-1})}^{\A}
=\T_{\chi}({\sigma,t_{0},t_{1},\dots,t_{n-1}})^{\A}
\end{align*}
\item[$\acta^{\A\red_{\chi}}=\T_{\chi}(\acta)^{\A}$]
We will demonstrate this using induction on structure.
We will divide the cases based on form.
\begin{proofcases}
\item[$\acta=\pos{\mathrm{label},\lambda}$]
$\pos{\mathrm{label},\lambda}^{\A\red_{\chi}}=\lambda^{\A\red_{\chi}}=\chi(\lambda)^{\A}=\T_{\chi}(\pos{\mathrm{label},\lambda})^{\A}$
\item[$\acta=\empact_{s}$]
$\empact^{\A\red_{\chi}}_{s}=\emptyset_{(\A\red_{\chi})_{s}}=\emptyset_{\A_{\T_{\chi}(s)}}=\empact^{\A}_{\T_{\chi}(s)}=\T_{\chi}(\empact_{s})^{\A}$
\item[$\acta=\one_{s}$]
$\one^{\A\red_{\chi}}_{s}=\id_{(\A\red_{\chi})_{s}}=\id_{\A_{\T_{\chi}(s)}}=\one^{\A}_{\T_{\chi}(s)}=\T_{\chi}(\one_{s})^{\A}$
%
%
%
\item[$\acta=\acta_{0}^{*}$]
$(\acta_{0}^{*})^{\A\red_{\chi}}
=(\acta_{0}^{\A\red_{\chi}})^{*}
\stackrel{IH}{=}(\T_{\chi}(\acta_{0})^{\A})^{*}
=(\T_{\chi}(\acta_{0})^{*})^{\A}
=\T_{\chi}(\acta_{0}^{*})^{\A}$
\item[$\acta=\acta_{0}\cupact\acta_{1}$]
$(\acta_{0}\cupact\acta_{1})^{\A\red_{\chi}}
=(\acta_{0}^{\A\red_{\chi}})\cup(\acta_{1}^{\A\red_{\chi}})
\stackrel{IH}{=}(\T_{\chi}(\acta_{0})^{\A})\cup(\T_{\chi}(\acta_{1})^{\A})
=(\T_{\chi}(\acta_{0})\cupact\T_{\chi}(\acta_{1}))^{\A}
=\T_{\chi}(\acta_{0}\cupact\acta_{1})^{\A}$
\item[$\acta=\acta_{0}\impact\acta_{1}$]
$(\acta_{0}\impact\acta_{1})^{\A\red_{\chi}}
=(\acta_{0}^{\A\red_{\chi}})^{c}\cup(\acta_{1}^{\A\red_{\chi}})
\stackrel{IH}{=}(\T_{\chi}(\acta_{0})^{\A})^{c}\cup(\T_{\chi}(\acta_{1})^{\A})
=(\T_{\chi}(\acta_{0})\impact\T_{\chi}(\acta_{1}))^{\A}
=\T_{\chi}(\acta_{0}\impact\acta_{1})^{\A}$
\item[$\acta=\acta_{0}\comp\acta_{1}$]
$(\acta_{0}\comp\acta_{1})^{\A\red_{\chi}}
=(\acta_{0}^{\A\red_{\chi}})\comp(\acta_{1}^{\A\red_{\chi}})
\stackrel{IH}{=}(\T_{\chi}(\acta_{0})^{\A})\comp(\T_{\chi}(\acta_{1})^{\A})
=(\T_{\chi}(\acta_{0})\comp\T_{\chi}(\acta_{1}))^{\A}
=\T_{\chi}(\acta_{0}\comp\acta_{1})^{\A}$
\item[$\acta=\acta_{0}\resi\acta_{1}$]
$(\acta_{0}\resi\acta_{1})^{\A\red_{\chi}}
=(\acta_{0}^{\A\red_{\chi}})^{-1}\comp(\acta_{1}^{\A\red_{\chi}})
\stackrel{IH}{=}(\T_{\chi}(\acta_{0})^{\A})^{-1}\comp(\T_{\chi}(\acta_{1})^{\A})
=(\T_{\chi}(\acta_{0})\resi\T_{\chi}(\acta_{1}))^{\A}
=\T_{\chi}(\acta_{0}\resi\acta_{1})^{\A}$
\end{proofcases}
\end{proofcases}
Since it is clear that $\Mod(\id_{\Sigma})=\id_{\Mod(\Sigma)}$, we will verify the conservation of composition.
\begin{align*}
&\sigma^{\A\red_{\chi_{1}\circ\chi_{0}}}=
\T_{\chi_{1}\circ\chi_{0}}(\sigma)^{\A}=
\T_{\chi_{1}}(\T_{\chi_{0}}(\sigma))^{\A}=
\T_{\chi_{0}}(\sigma)^{\A\red_{\chi_{1}}}=
\sigma^{\A\red_{\chi_{1}}\red_{\chi_{0}}}
\\
&\lambda^{\A\red_{\chi_{1}\circ\chi_{0}}}=
\T_{\chi_{1}\circ\chi_{0}}(\lambda)^{\A}=
\T_{\chi_{1}}(\T_{\chi_{0}}(\lambda))^{\A}=
\T_{\chi_{0}}(\lambda)^{\A\red_{\chi_{1}}}=
\lambda^{\A\red_{\chi_{1}}\red_{\chi_{0}}}
\end{align*}
\end{proof}

\begin{proof}[Proof of well-definedness of Definition~\ref{def:Satisfaction-relation}]
We will verify the satisfaction condition.
When we defined $\Mod$, we verified that
$t^{\A\red_{\chi}}=\T_{\chi}(t)^{\A}$,
$\acta^{\A\red_{\chi}}=\T_{\chi}(\acta)^{\A}$
Using this, it is easy to see that the satisfaction condition holds true for action sentences and equations.
We can prove this by induction for other operators as well, but since we are using second-order quantification this time, we will only verify that there are no problems in that part.
First, as preparation, We will show that
"for each $\chi:\Sigma\to\Sigma'\in\Sig^{\gene}$, $X\in\Sigma_{\Dex}$, $\hat{\A}\in|\Mod(\Sigma^{\Dex}[X])|$, $\A'\in|\Mod(\Sigma')|$ such that$\hat{\A}\red_{\Sigma^{\Dex}(X)}=\A'\red_{\chi}$, there is unique $\hat{\A}'\in|\Mod(\Sigma^{\Dex}[X'])|\,(X':=\chi_{\Dex}(X))$ that satisfies $\hat{\A}'\red_{\Sigma^{\Dex}(X')}=\A'$ and $\hat{\A}'\red_{\chi^{\Dex}[X]}=\hat{\A}$".
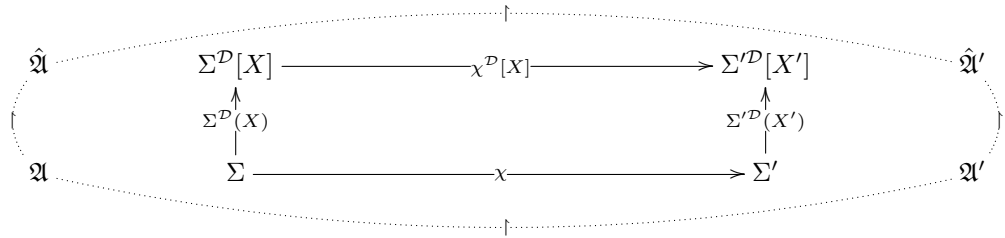
\[\xymatrix@R=25pt@C=50pt{ 
\hat{\A}\ar@/^+20pt/@{.}[rrrrr]|{\red}&\Sigma^{\Dex}[X] \ar[rrr]|{\chi^{\Dex}[X]} &&& \Sigma'^{\Dex}[X']&\hat{\A}' \\
\ar@/^10pt/@{.}[u]|{\red}\A\ar@/^-20pt/@{.}[rrrrr]|{\red}&\Sigma \ar[u]|{\Sigma^{\Dex}(X)} \ar[rrr]|{\chi} &&& \Sigma' \ar[u]|{\Sigma'^{\Dex}(X')}&\A'\ar@/^-10pt/@{.}[u]|{\red}
}\]
First, from $\hat{\A}'\red_{\Sigma^{\Dex}(X')}=\A'$ and $\hat{\A}'\red_{\chi^{\Dex}[X]}=\hat{\A}$, we can see that it cannot be done unless the definition is at least as follows.
\begin{itemize}
\item
For symbols of $\Sigma'=\pos{S',F',L'}$,
$\hat{\A}'_{s}:=\A'_{s}\,(s\in S')$,
$\sigma^{\hat{\A}'}:=\sigma^{\A'}\,(\sigma\in F')$, and
$\lambda^{\hat{\A}'}:=\lambda^{\A'}\,(\lambda\in L')$.
\item
For variables
$\pos{\mathrm{var},n,\Sigma'}^{\hat{\A}'}:=\pos{\mathrm{var},n,\Sigma}^{\hat{\A}}$.
\end{itemize}
Conversely, the above model can be defined%
\footnote{Using a many-sorted relation variable would make it impossible to guarantee the existence in this part.}
and satisfies $\hat{\A}'\red_{\Sigma^{\Dex}(X')}=\A'$ and $\hat{\A}'\red_{\chi^{\Dex}[X]}=\hat{\A}$.
\ \\
The proof of the satisfaction condition is done by induction.
(We will just check the $\exists$ part.)
\begin{align*}
&\A'\models\chi(\Exists{X}\phi)\,(=\Exists{X'}\chi^{\Dex}[X](\phi))\\
&\iff\hat{\A}'\models\chi^{\Dex}[X](\phi)\text{ for some }\Sigma^{\Dex}(X')\text{-expansion }\hat{\A}'\text{ of }\A'\\
&\stackrel{IH}{\iff}\hat{\A}'\red_{\chi^{\Dex}[X]}\models\phi\text{ for some }\Sigma^{\Dex}(X')\text{-expansion }\hat{\A}'\text{ of }\A'\\
&\iff\hat{\A}\models\phi\text{ for some }\Sigma^{\Dex}(X')\text{-expansion }\hat{\A}'\text{ of }\A'\text{ and }\hat{\A}=\hat{\A}'\red_{\chi^{\Dex}[X]}\\
&\iff\hat{\A}\models\phi\text{ for some }\hat{\A}\,\,\&\,\,\hat{\A}'\text{ such that }\hat{\A}'\red_{\Sigma^{\Dex}(X')}=\A'\text{ and }\hat{\A}'\red_{\chi^{\Dex}[X]}=\hat{\A}\\
&\stackrel{*}{\iff}\hat{\A}\models\phi\text{ for some }\hat{\A}\text{ such that }\hat{\A}\red_{\Sigma^{\Dex}(X)}=\A'\red_{\chi}\\
&\iff\hat{\A}\models\phi\text{ for some }\Sigma^{\Dex}(X)\text{-expansion }\hat{\A}\text{ of }\A'\red_{\chi}\\
&\iff\A'\red_{\chi}\models\Exists{X}\phi
\end{align*}
The equivalence of the $*$ part uses the existence of the prepared extension.
\end{proof}

\section{Proof of the results presented in Section~\ref{sec:sequent}}

\begin{proof}[Proof of well-definedness of Definition~\ref{def:Semantic-sequent}]
We confirm that this satisfies the structural rules.
\begin{proofcases}
\item[$\Modify$] Let $\A'\in|\Mod(\Sigma')|$.
Assuming the premise holds true, $\A'\red_{\chi}$ cannot satisfies $\Gamma\cup\{\neg\delta\mid\delta\in\Delta\}$.
Also by satisfaction condition, $\A'$ cannot satisfies $\chi(\Gamma)\cup\{\neg\delta\mid\delta\in\chi(\Delta)\}$.
So $\A'$ cannot satisfies $\Gamma'\cup\{\neg\delta\mid\delta\in\Delta'\}$.
Since the choice of $\A'$ was arbitrary, this means that the conclusion holds true.
\item[$\Init$] Because $\A\models\phi$ and $\A\not\models\phi$ cannot be true at the same time.
\item[$\Cut$] We show that if the conclusion is not true, then one of the premises must not be true.
From the assumption, there exists a model $\A$ that satisfies $\Gamma\cup\Gamma'$ and completely negates $\Delta\cup\Delta'$.
This satisfies either $\A\models\psi$ or $\A\not\models\psi$, so one of the premises cannot be held.
\end{proofcases}
\end{proof}

\begin{proof}[Proof of well-definedness of Definition~\ref{def:basic-proof}]
Strictly speaking, proof trees are defined recursively in the form
\[\pos{\text{Rule name}^{\{\text{Information about choice}\}},\proofrule{\text{Sequence of premise proof trees}}{\text{Root}},\text{Applied part}}.\]
\begin{itemize}
\item In the rule display, only the root portion of "Sequence of premises proof trees" is shown.
\item
The "Applied part" is $\proofrule{\pos{\Gamma'_{i}\vdash_{\Sigma_{i}}\Delta'_{i}}_{i\in I}}{\Gamma'\vdash_{\Sigma}\Delta'}$
for rules written in the form $\proofrule{\pos{\UL{\Gamma'_{i}}\vdash_{\Gamma_{i}\mid\Sigma_{i}\mid\Delta_{i}}\UL{\Delta'_{i}}}_{i\in I}}{\OL{\Gamma'}\vdash_{\Gamma\mid\Sigma\mid\Delta}\OL{\Delta'}}$.
\\
Note that in this paper, formally, sequent is defined using a set rather than a sequence of sentences, so it may be difficult to determine what the rule was applied to from the surrounding structure.
This is redundant in notation, so it will generally be omitted.
\item
"Information about choice" doesn't appear here, but it's an example of $\theta$ for $(\exists_{R}^{\theta})$ which will be introduced later.
If it exists, it's written as a subscript to the upper right of the rule.
This is information to avoid non-deterministic choices in recursive operations on proof trees.
(For example, note that if $\chi(\phi)=\chi(\phi')=\psi$, then the translation of $\phi\vee\phi'$ is $\psi\vee\psi$.)
\end{itemize}
{\par}
We will prove the well-definedness of $\Sen^{b}$. We will verify that it satisfies three axioms.
\begin{proofcases}
\item[$\Modify$]
Suppose $\Gamma\vdash^{b}_{\Sigma}\Delta$ has been proven.
We have a proof tree. Since proof rules hold true when translated, translating all nodes with $\chi:\Sigma\to\Sigma'$ gives us the proof of $\chi(\Gamma)\vdash^{b}_{\Sigma'}\chi(\Delta)$.
Next, adding $\Gamma'$ and $\Delta'$ to the left and right of all nodes respectively gives us the proof tree of $\Gamma'\vdash^{b}_{\Sigma'}\Delta'$.
\item[$\Init$]
It exists as a rule.
\item[$\Cut$]
Assume that $\Gamma\vdash^{b}_{\Sigma}\Delta\,\psi$ and $\psi\,\Gamma'\vdash^{b}_{\Sigma}\Delta'$ have been proven.
Hereafter, we will refer to $\psi\,\Gamma'\vdash^{b}_{\Sigma}\Delta'$ as the right-hand side.
We will show that $(\Cut)$ is possible by induction on the complexity of the proof diagram on the right.%
\footnote{
Note that the basic rules only deal with sentences that appear to the right of the sequent.
}
Since the others are the same, we will only check four cases.
\begin{proofcases}
\item[If the last rule on the right is $(\Init_{\phi})$]
We divide it into cases based on the form of $\phi$ used in the rule.
\begin{proofcases}
\item[if $\psi=\phi$]
In this case, the sequent obtained by $(\Cut^{\psi})$ is a modified version of the left ($\Gamma\vdash^{b}_{\Sigma}\Delta\cup\{\psi\}$).
Therefore, it can be proven.
\item[if $\psi\neq\phi$]
In this case, the sequent obtained by $(\Cut^{\psi})$ contains the sentence used in the rule, so it can be obtained directly by the $(\Init)$ rule.
\end{proofcases}
\item[if the last rule on the right is $(\Ry)$]
In this case, the sequent obtained by $(\Cut)$ contains the sentence used in the rule, so it can be obtained directly by the $(\Ry)$ rule.
\item[If the last rule on the right is $(\Sy)$]
This can be done simply by changing the order in which $(\Sy)$ and $(\Cut)$ are applied.
\item[If the last rule on the right is $(\Ty)$]
Similar to $(\Sy)$, we apply $(\Cut)$ to the premises instead of the conclusion, but
if only one of $\Gamma,\Gamma'$ contains $\psi$, then we apply $(\Modify(\id_{\Sigma}))$ to the other instead of $(\Cut)$.
Since we are simply adding a set of sentences to the entire proof tree, the complexity of the proof tree does not increase.
\end{proofcases}
\end{proofcases}
In intuitionistic logic, sentences are not added to the right, so
for $\Modify$, sentence additions to each node are only made to the left.
Otherwise, it's the same.
\end{proof}

\begin{proof}[Proof of Lemma~\ref{lemma:basic-soundness-completeness-compactness}]
The proof is almost the same as \cite{go-icalp24}.
\end{proof}

\begin{proof}[Proof of well-definedness of Definition~\ref{def:Syntactic-sequent}]
We prove that the proof is well-defined (by showing that it satisfies three axioms).
We verify that the proof can be constructed recursively.
\begin{proofcases}
\item[$\Modify$]
For proof tree $T$,
signature morphism $\chi:\Sigma\to\Sigma'$ whose domain is the root signature of $T$, and
sentence sets $\hat{\Gamma}',\hat{\Delta}'\subseteq\Sen(\Sigma')$,
such that $\hat{\Gamma}'\supseteq\chi(\hat{\Gamma})$, $\hat{\Delta}'\supseteq\chi(\hat{\Delta})$ hold true for root $\hat{\Gamma}\vdash_{\Sigma}\hat{\Delta}$ of $T$,
we define new proof tree $\Modify(\chi,T,\hat{\Gamma}',\hat{\Delta}')$ for $\hat{\Gamma}'\vdash_{\Sigma'}\hat{\Delta}'$ recursively with respect to the structure of the proofs.%
\footnote{
If we represent all signatures obtained by repeatedly adding variables to $\Sigma$ as $\Dex_{\Sigma}$, this becomes a set.
Therefore, all proof trees whose root signature is an element of $\Dex_{\Sigma}$ also become a set.
Strictly speaking, recursive definitions are only possible if the order satisfies set-likeness, so using a function that assigns a structure to perform recursion is not allowed.
However, in practice, we only use proof information that has $\Dex_{\Sigma}$ as its root, so we can consider set-likeness to be guaranteed.
}
{\par}
Generally, proof rules take the following form:
\[name^{choice}~\proofrule
{\pos{\UL{\Phi_{i}}\vdash_{\iota_{i}(\Gamma_{i})\mid\Sigma_{i}\mid\iota_{i}(\Delta_{i})}\UL{\Psi_{i}}}_{i\in I}}
{\OL{\Phi}\vdash_{\cup_{i\in I}\Gamma_{i}\mid\Sigma\mid\cup_{i\in I}\Delta_{i}}\OL{\Psi}}\]
where $\hat{\Gamma}=\cup_{i\in I}\Gamma_{i}\cup\Phi$, $\hat{\Delta}=\cup_{i\in I}\Delta_{i}\cup\Psi$.
\begin{align*}\hspace*{-1cm}
\Modify(\chi,
  &\pos{name^{choice},
  \proofrule{\pos{T_{i}}_{i\in I}}{\hat{\Gamma}\vdash_{\Sigma}\hat{\Delta}},
  \proofrule{\pos{\Phi_{i}\vdash_{\Sigma_{i}}\Psi_{i}}_{i\in I}}{\Phi\vdash_{\Sigma}\Psi}},
\hat{\Gamma}',\hat{\Delta}')&\\
:=&\pos{name^{choice'},
\proofrule{\pos{\Modify(\chi_{i},T_{i},\hat{\Gamma}'_{i},\hat{\Delta}'_{i})}_{i\in I}}{\hat{\Gamma}'\vdash_{\Sigma'}\hat{\Delta}'},
\proofrule{\pos{\chi_{i}(\Phi_{i})\vdash_{\Sigma'_{i}}\chi_{i}(\Psi_{i})}_{i\in I}}{\chi(\Phi)\vdash_{\Sigma'}\chi(\Psi)}}&\\
&\chi_{i}:\Sigma_{i}\to\Sigma'_{i}:=
\begin{cases}
\chi^{\Dex}[X]:\Sigma^{\Dex}[X]\to\Sigma'^{\Dex}[X']\,(X':=\chi_{\Dex}(X))&(\text{when the rule has variable adding})\\
\chi:\Sigma\to\Sigma'&(\text{else})
\end{cases}
&\\
&\pos{\hat{\Gamma}'_{i},\hat{\Delta}'_{i}}:=
\begin{cases}
\pos{\chi_{i}(\Phi_{i})\,\iota'_{i}(\hat{\Gamma}'),\iota'_{i}(\hat{\Delta}')\,\chi_{i}(\Psi_{i})}\,(\iota'_{i}:=\Sigma'^{\Dex}(\chi_{\Dex}(X)))&(\text{when the rule has variable adding})\\
\pos{\chi_{i}(\Phi_{i})\,\hat{\Gamma}',\hat{\Delta}'\,\chi_{i}(\Psi_{i})}&(\text{else})
\end{cases}
&\\
&choice':=
\begin{cases}
\chi(t)&(choice=t)\\
n&(choice=n)\\
\theta'&(choice=\theta:\Sigma^{\Dex}(X)\text{-substitution})
\end{cases}
&\\
&(\theta':\text{ unique }\Sigma'^{\Dex}(\chi_{\Dex}(X))\text{-substitution which satisfies }\chi\circ\theta=\theta'\circ\chi_{\Dex}[X])&
\end{align*}
The process is the same as with $\vdash^{b}$: translating the proof diagram and adding extra sentences to each node.
Since the rule structure is preserved, the result obtained by applying $\Modify$ is also a proof tree.
{\par}
In intuitionistic case, the operation of adding to the right changes.
Since we are not using $\neg\phi,\acta^{\Coo}$, the right side will not suddenly become empty in the middle of the proof.
The right side becomes empty when we did not add extra sentences to the right when introducing $\empact$ or $\bot=\vee\pos{}$ to the left, or when the proof retains them.
Therefore, we just need to go back up the proof and add the corresponding sentences to the nodes where the right side is empty.
\begin{align*}
&
\hat{\Delta}'_{i}:=
\begin{cases}
\chi_{i}(\Psi_{i})&(\text{if }\chi_{i}(\Psi_{i})\neq\emptyset)\\
\iota'_{i}(\hat{\Delta}')\,(\iota'_{i}:=\Sigma'^{\Dex}(\chi_{\Dex}(X)))&(\text{else if the rule involves variable adding})\\
\hat{\Delta}'&(\text{else})
\end{cases}
&
\end{align*}
\item[preparation]
Before discussing $\Init$ and $\Cut$, let's define the complexity of the sentences.
We define $C^{0}(\acta)\in\omega^{2}$ and $C(\phi)\in\omega^{3}$ as follows.
All $\alpha,\beta,\gamma$ and that with subscripts appear below are natural numbers.
\\
$\small
C^{0}(\acta):=
\begin{cases}
0&(\acta=\lambda)\\
1&(\acta=\empact,\one)\\
\omega(\beta+1)+\gamma&(\acta=\acta_{0}^{*}, C^{0}(\acta)=\omega\beta+\gamma)\\
\omega(\max_{i\in2}\beta_{i})+\sum_{i\in2}\gamma_{i}&(\acta=\acta_{0}\rho\acta_{1}\,(\rho\in\{\cupact,\impact,\comp,\resi\}),\,C^{0}(\acta_{i})=\omega\beta_{i}+\gamma_{i}\,(i\in2))
\end{cases}
$
\\
$\small
C(\phi):=
\begin{cases}
0&(\phi=(t_{0}=t_{1}))\\
C(\acta)&(\phi=t_{0}\stackrel{\acta}{\Longrightarrow}t_{1})\\
\omega^{2}(1+\sum_{i\in2}\alpha_{i})+\omega(\max_{i\in2}\beta_{i})+\sum_{i\in2}\gamma_{i}&(\phi=\phi_{0}\to\phi_{1},\,C(\phi_{i})=\omega^{2}\alpha_{i}+\omega\beta_{i}+\gamma_{i}\,(i\in2))\\
\omega^{2}(1+\sum_{i\in N}\alpha_{i})+\omega(\max_{i\in N}\beta_{i})+\sum_{i\in N}\gamma_{i}&(\phi=\vee\psi\,(\wedge\psi),\,C(\psi_{i})=\omega^{2}\alpha_{i}+\omega\beta_{i}+\gamma_{i}\,(i\in N))\\
\omega^{2}(\alpha+1)+\omega\beta+\gamma&(\phi=\Exists{X}\psi\,(\Forall{X}\psi),\,C(\psi)=\omega^{2}\alpha+\omega\beta+\gamma)\\
\end{cases}
$
\\
When $C(\phi)=\omega^{2}\alpha+\omega\beta+\gamma$,
$\alpha$ is the number of $\{\to,\vee,\wedge,\exists,\forall\}$ that appear,
$\beta$ is the "height" due to $*$,
$\gamma$ is the number of symbols other than $*$ that appear in the actions that appear, and $C$ evaluates these left-first.
The order defined from here satisfies, for example, $(\Exists{\lambda}t_{0}\stackrel{\lambda}{\Longrightarrow}t_{1})>(t_{0}\stackrel{(\pi\comp\acta)^{*}}{\Longrightarrow}t_{1})>(t_{0}\stackrel{(\pi\comp\acta)^{n}}{\Longrightarrow}t_{1})>(t_{0}=t_{1})$.
\item[$\Init$]
We will see that the function $\Init$ can be recursively defined for the complexity of the sentence $\psi$ defined above, which constructs a new proof tree $\Init(\Gamma,\Delta,\Sigma,\psi)$ for a sentence $\psi$, a signature $\Sigma$ such that $\psi\in\Sen(\Sigma)$ holds true, and a sentence set $\Gamma,\Delta\subseteq\Sen(\Sigma)$.
We will consider different cases based on the sentence's form. Since others are similar, we will only discuss some of them.
\begin{proofcases}
\item[$\psi$ is atomic]
Since $\psi\vDash^{b}_{\Sigma}\psi$, it can be defined by the following diagram.
\begin{align*}
\Atom~\proofrule{}{\,\psi\vdash_{\Gamma\mid\Sigma\mid\Delta}\psi\,}
\end{align*}
\item[$\psi=\vee\psi'$]
We draw it as a diagram.
\[
{\vee}_{L}~\proofrule
{\pos{{\vee}_{R}^{n}~\proofrule{\proofrulet{:}{\,\psi'_{n}\vdash_{\Gamma\mid\Sigma\mid\Delta}\psi'_{n}\,}}
{\,\psi'_{n}\vdash_{\Gamma\mid\Sigma\mid\Delta}\vee\psi'\,}}_{n\in N}}
{\,\vee\psi'\vdash_{\Gamma\mid\Sigma\mid\Delta}\vee\psi'\,}
\]
The part above $\psi'\,\Gamma\vdash_{\Sigma}\Delta\,\psi'$ is already defined by $\Init(\Gamma,\Delta,\Sigma,\psi')$.
This argument is independent of the size of $N$. Therefore, a similar argument can be made for $*$.
\item[$\psi=\Exists{X}\psi'$]
Let $X'$ be the result of translating the contents of $X$ using $\Sigma^{\Dex}(X):\Sigma\to\Sigma^{\Dex}[X]$.
We have $\Sigma^{\Dex}(X)^{\Dex}[X]:\Sigma^{\Dex}[X]\to\Sigma^{\Dex}[X]^{\Dex}[X']\,(X\mapsto X')$.
If we reverse the translation of the elements of $X$, we obtain a substitution $\theta$ such that $X'\ni x'\mapsto x\in X$.
Note that $\theta(\Sigma^{\Dex}(X)^{\Dex}[X](\psi'))=\psi'$,
\[
\exists_{L}~\proofrule
{\proofrule
{\exists_{R}^{\theta}~\proofrule
 {\proofrulet{:}{\psi'\vdash_{\Sigma^{\Dex}(X)(\Gamma)\mid\Sigma^{\Dex}[X]\mid\Sigma^{\Dex}(X)(\Delta)}\psi'}}
 {\psi'\vdash_{\Sigma^{\Dex}(X)(\Gamma)\mid\Sigma^{\Dex}[X]\mid\Sigma^{\Dex}(X)(\Delta)}\Exists{X'}\Sigma^{\Dex}(X)^{\Dex}[X](\psi')}}
{\psi'\vdash_{\Sigma^{\Dex}(X)(\Gamma)\mid\Sigma^{\Dex}[X]\mid\Sigma^{\Dex}(X)(\Delta)}\Sigma^{\Dex}(X)(\Exists{X}\psi')}}
{\Exists{X}\psi'\vdash_{\Gamma\mid\Sigma\mid\Delta}\Exists{X}\psi'}
\]
\normalsize
The part above $\psi'\,\Sigma(X)(\Gamma)\vdash_{\Sigma[X]}\Sigma(X)(\Delta)\,\psi'$ is already defined by $\Init(\Gamma,\Delta,\Sigma,\psi')$.
\end{proofcases}
In the case of intuitionistic, we can simply set $\Delta=\emptyset$.
\item[$\Cut$]
We will see that for two proof trees $T$ and $T'$ whose roots are $\Gamma\vdash_{\Sigma}\Delta\,\psi$ and $\psi\,\Gamma'\vdash_{\Sigma}\Delta'$ respectively,
a new proof tree $\Cut(\Gamma\,\Gamma'\vdash_{\Sigma}\Delta\,\Delta'$ whose root is $\Gamma\,\Gamma'\vdash_{\Sigma}\Delta\,\Delta'$
can be recursively defined for the lexicographic order of ($C(\psi)$ : the complexity of $\psi$) and (the Cartesian product order of proof structures).
\begin{align*}
&\Cut^{\psi}~\proofrule
{\,\proofrulet{T}{\vdash_{\Gamma\mid\Sigma\mid\Delta}\UL{\psi}}\,\Space\,\proofrulet{T'}{\UL{\psi}\vdash_{\Gamma'\mid\Sigma\mid\Delta'}}\,}
{\,\vdash_{\Gamma\,\Gamma'\mid\Sigma\mid\Delta\,\Delta'}\,}
\end{align*}
We will divide this into several cases.
In the following, the earlier cases will take precedence, and the lower cases will be those where the upper cases do not hold true.
\begin{proofcases}
\item[$\psi\in\Delta\cup\Delta'\cup\Gamma\cup\Gamma'$]
Since both are similar, we consider the case $\psi\in\Delta\cup\Delta'$.
The conclusion obtained by $(\Cut)$ is the conclusion of $T$ with an extra sentence added.
Therefore, we can obtain it by ignoring $T'$ and applying $\Modify$ to $T$.
(The same applies to intuitionistic case.)
\item[if $\psi$ is not the applied part of the rule]
Since we have already considered the above case, we can think of it as $\psi\not\in\Delta\cup\Delta'\cup\Gamma\cup\Gamma'$.
Since both are similar, we consider the case where the root rule of $T'$ is unrelated to $\psi$.
Generally, $T'$ takes the following form:
\[name^{choice}~\proofrule
{\pos{\proofrulet{T'_{i}}{\UL{\Phi'_{i}}\vdash_{\iota_{i}(\Gamma'_{i})\mid\Sigma_{i}\mid\iota_{i}(\Delta'_{i})}\UL{\Psi'_{i}}}}_{i\in I}}
{\OL{\Phi'}\vdash_{\cup_{i\in I}\Gamma'_{i}\mid\Sigma\mid\cup_{i\in I}\Delta'_{i}}\OL{\Psi'}}\]
where $\Gamma'\cup\{\psi\}=\cup_{i\in I}\Gamma'_{i}\cup\Phi'$, $\Delta'=\cup_{i\in I}\Delta'_{i}\cup\Psi'$.
\begin{proofcases}
\item[if $I=\emptyset$]
In other words, if $\psi$ is introduced as an extra sentence by $(\Atom)$ or "left introduction of $\empact,\bot$",
then instead of introducing $\psi$, we can add an extra sentence to $T'$ so that the result is the same as what we get with $(\Cut)$.
We ignore $T$.
(The same applies to intuitionistic case.)
\item[if $I\neq\emptyset$] There is $i\in I$ such that $\psi\in\Gamma'_{i}$.
We define a new proof diagram as follows:
\[name^{choice}~\proofrule
{\pos{\proofrulet{T''_{i}}{\UL{\Phi'_{i}}\vdash_{\Gamma''_{i}\mid\Sigma_{i}\mid\Delta''_{i}}\UL{\Psi'_{i}}}}_{i\in I}}
{\OL{\Phi'}\vdash_{\Gamma\cup(\cup_{i\in I}\Gamma'_{i}\setminus\{\psi\})\mid\Sigma\mid\Delta\cup(\cup_{i\in I}\Delta'_{i})}\OL{\Psi'}}\]
where
\begin{proofcases}
\item[if $\psi\in\Gamma'_{i}$] $\pos{\Gamma''_{i},\Delta''_{i}}:=\pos{\iota_{i}(\Gamma)\cup\iota_{i}(\Gamma'_{i}\setminus\{\psi\}),\iota_{i}(\Delta)\cup\iota_{i}(\Delta'_{i})}$, and
\\ \ \\
$
T''_{i}
:=
\Cut^{\iota_{i}(\psi)}~\proofrule{\proofrulet{\Modify(\iota_{i},T,\iota_{i}(\Gamma),\iota_{i}(\Delta)\cup\iota_{i}(\psi))}{\,\vdash_{\iota_{i}(\Gamma)\mid\Sigma_{i}\mid\iota_{i}(\Delta)}\UL{\iota_{i}(\psi)}}\Space\proofrulet{T'_{i}}{\UL{\iota_{i}(\psi)}\,\Phi'_{i}\vdash_{\iota_{i}(\Gamma'_{i}\setminus\{\psi\})\mid\Sigma_{i}\mid\iota_{i}(\Delta'_{i})}\Psi'_{i}}}
{\UL{\Phi'_{i}}\vdash_{\iota_{i}(\Gamma)\cup\iota_{i}(\Gamma'_{i}\setminus\{\psi\})\mid\Sigma_{i}\mid\iota_{i}(\Delta)\cup\iota_{i}(\Delta'_{i})}\UL{\Psi'_{i}}}
$
\\
\item[if $\psi\not\in\Gamma'_{i}$] $\pos{\Gamma''_{i},\Delta''_{i}}:=\pos{\iota_{i}(\Gamma'_{i}),\iota_{i}(\Delta'_{i})}$, and $T''_{i}:=T'_{i}$.
\end{proofcases}
(In the intuitionistic case, since $\psi\not\in\Delta\cup\Delta'\cup\Gamma\cup\Gamma'$, $\Delta=\emptyset$ holds true, and
the elements of $\iota_{i}(\Delta)\cup\iota_{i}(\Delta_{i})\cup\Psi_{i}$ remain at most one. In other words, that operation does not break the intuitionistic proof.)
\end{proofcases}
\item[If $\psi$ is the applied part in both rules]
Since we have already considered the above case, we can think of it as $\psi\not\in\Delta\cup\Delta'\cup\Gamma\cup\Gamma'$.
Since others are similar, we will only discuss some of them.
\begin{proofcases}
\item[$\psi$ is atomic]
The $\Atom$ rule involves adding an extra compound sentence to what has been proven by $\vdash^{b}$.
Since $\vdash^{b}$ satisfies $(\Cut)$, we can apply $(\Cut)$ first and then add the extra sentence.
\item[$\psi=\vee\phi$]
Since $\bot=\vee\pos{}$ cannot be introduced by ${\vee}_{R}$, we can assume that both are not empty sequences.
The root parts of $T$ and $T'$ are in the form
${\vee}_{R}^{n}~\proofrule
{\Gamma\vdash_{\Sigma}\Delta\,\phi_{n}}
{\Gamma\vdash_{\Sigma}\Delta\,\vee\phi}$ and
${\vee}_{L}~\proofrule
{\pos{\phi_{n}\,\Gamma'\vdash_{\Sigma}\Delta'}_{n\in N}}
{\vee\phi\,\Gamma'\vdash_{\Sigma}\Delta'}$ respectively.
Let the proofs of the premises be $T_{n}$, $T'_{i}\,(i\in N)$ respectively.
We can define it by $\Cut(\Gamma,\Gamma',\Delta,\Delta',\psi,T,T')
:=\Cut(\Gamma,\Gamma',\Delta,\Delta',\phi_{n},T_{n},T'_{n})$.
\item[$\psi=\Exists{X}\phi$]
The root parts of $T$ and $T'$ are in the form
$\exists_{R}^{\theta}~\proofrule
{\Gamma\vdash_{\Sigma}\Delta\,\theta(\phi)}
{\Gamma\vdash_{\Sigma}\Delta\,\Exists{X}\phi}$ and
$\exists_{L}~\proofrule
{\phi\,\Sigma^{\Dex}(X)(\Gamma')\vdash_{\Sigma^{\Dex}[X]}\Sigma^{\Dex}(X)(\Delta')}
{\Exists{X}\phi\,\Gamma'\vdash_{\Sigma}\Delta'}$ respectively.
Translating the premise $\exists_{L}$ in terms of $\theta$ gives
$\theta(\phi)\,\Gamma'\vdash_{\Sigma}\Delta'$.%
\footnote{The lengthy explanation at the beginning is preparation for using "translation by substitution" here.}
Let the proofs of the premises be $T_{0}$, $T'_{0}$ respectively.
We can define it by
\[
\Cut(\Gamma,\Gamma',\Delta,\Delta',\psi,T,T')
:=\Cut(\Gamma,\Gamma',\Delta,\Delta',\theta(\phi),T_{0},\Modify(\theta,T'_{0},\theta(\phi)\,\Gamma',\Delta')).
\]
\item[$\psi=t_{1}\stackrel{\acta_{1}\comp\acta_{2}}{\Longrightarrow}t_{2}$]
The root parts of $T$ and $T'$ are in the form
\\
${\comp}_{R}^{t}~\proofrule{\Gamma_{L}\vdash_{\Sigma}\Delta_{L}\,t_{1}\stackrel{\acta_{1}}{\Longrightarrow}t\Space\Gamma_{R}\vdash_{\Sigma}\Delta_{R}\,t\stackrel{\acta_{2}}{\Longrightarrow}t_{2}}{\Gamma\vdash_{\Sigma}\Delta\,t_{1}\stackrel{\acta_{1}\comp\acta_{2}}{\Longrightarrow}t_{2}}$
and \\
${\comp}_{L}~\proofrule{t_{1}\stackrel{\acta_{1}}{\Longrightarrow}x\,x\stackrel{\acta_{2}}{\Longrightarrow}t_{2}\,\Sigma^{\Dex}(x)(\Gamma')\vdash_{\Sigma^{\Dex}[x]}\Sigma^{\Dex}(x)(\Delta')}{t_{1}\stackrel{\acta_{1}\comp\acta_{2}}{\Longrightarrow}t_{2}\,\Gamma'\vdash_{\Sigma}\Delta'}$
respectively.\\
Translating the premise ${\comp}_{L}$ in terms of $\pos{x\mapsto t}$ gives the proof $T'_{t}$ of
$t_{1}\stackrel{\acta_{1}}{\Longrightarrow}t\,t\stackrel{\acta_{2}}{\Longrightarrow}t_{2}\,\Gamma'\vdash_{\Sigma}\Delta'$. 
If we let $T_{L}$ be the proof of the left premise of $T$ and $T_{R}$ be the proof of the right premise of $T$, then from the recursive hypothesis, we can use $(\Cut)$ twice to obtain the proof of $\Gamma\,\Gamma'\vdash_{\Sigma}\Delta\,\Delta'$.
Even though the complexity of the proof structure increases with the application of $(\Cut)$ to $T_{R},T'_{t}$, the format order is complexity-first of $\psi$, so $(\Cut)$ can be applied again to $T_{L}$.
\end{proofcases}
(The same applies to intuitionistic case.)
\end{proofcases}
\end{proofcases}
\end{proof}

\begin{proof}[Proof of Lemma~\ref{lemma:Soundness-Completeness}]
The proof is almost the same as \cite{go-icalp24}.
Although $(\to)$ was not originally a constructor,
classically, it is equivalent if $\phi_{0}\to\phi_{1}$ is restricted to $\phi_{1}=\bot$ and used instead of $\neg\phi$.
The same restriction can be applied to $(\impact)$, and if it is treated the same way as $\neg$, its completeness can be proven by forcing.
\end{proof}

\section{Proof of the results presented in Section~\ref{sec:induction}}

\begin{proof}[Proof of Lemma~\ref{lemma:induction-rules}]
Since the same applies to other things, I will only discuss a few.
	\begin{proofcases}
		\item[$\Ind_{R}^{0}$] By $(\Ry)$ and $(*_{R})$.
		\item[$\Ind_{R}^{+}$]
		$t_{0}\stackrel{\acta^{n}}\Longrightarrow t\,t\stackrel{\acta}\Longrightarrow t_{1}\vdash_{\Sigma}t_{0}\stackrel{\acta^{n+1}}\Longrightarrow t_{1}\vdash_{\Sigma} t_{0}\stackrel{\acta^{*}}\Longrightarrow t_{1}$ holds true for all $n\in\omega$.
		By $(*_{L})$, $t_{0}\stackrel{\acta^{*}}\Longrightarrow t\,t\stackrel{\acta}\Longrightarrow t_{1}\vdash_{\Sigma} t_{0}\stackrel{\acta^{*}}\Longrightarrow t_{1}$.
		By $(\Cut)$ and assumptions, $\Gamma\,\Gamma'\vdash_{\Sigma}\Delta\,\Delta'\,t_{0}\stackrel{\acta^{*}}\Longrightarrow t_{1}$
		\item[$\Ind_{L}^{+}$]
		Let $\phi(v):=t\stackrel{\act}\Longrightarrow v$.
		For any term $u\in F_{\T(\Sigma^{\Dex}[x,y])}$ let $\phi(u):=\theta(\phi)$ where $\theta$ is a substitution which satisfies $\theta(v)=u$.
		By the assumption and $(\Modify\,z\mapsto t)$, $\phi(x)\,x\stackrel{\acta}\Longrightarrow y\,\Sigma^{\Dex}(x,y)(\Gamma_{M})\vdash_{\Sigma^{\Dex}[x,y]}\Sigma^{\Dex}(x,y)(\Delta_{M})\,\phi(y)$.
		By induction on $n\in\omega$, we can show that
		$\phi(x)\,x\stackrel{\acta^{n}}\Longrightarrow y\,\Sigma^{\Dex}(x,y)(\Gamma_{M})\vdash_{\Sigma^{\Dex}[x,y]}\Sigma^{\Dex}(x,y)(\Delta_{M})\,\phi(y)$ for all $n\in\omega$.
		By $(\Modify)$, $\phi(t_{0})\,t_{0}\stackrel{\acta^{n}}\Longrightarrow t_{1}\,\Gamma_{M}\vdash_{\Sigma}\Delta_{M}\,\phi(t_{1})$.
		By $(*_{L})$, $\phi(t_{0})\,t_{0}\stackrel{\acta^{*}}\Longrightarrow t_{1}\,\Gamma_{M}\vdash_{\Sigma}\Delta_{M}\,\phi(t_{1})$.
		By the assumptions, and $(\Cut)$, $t_{0}\stackrel{\acta^{*}}\Longrightarrow t_{1}\,\Gamma\vdash_{\Sigma}\Delta$.
	\end{proofcases}
	The same applies to intuitionistic case.
\end{proof}

\begin{proof}[Proof of well-definedness of Definition~\ref{def:sequent-with-ind}]
We show the well-definedness.
\begin{proofcases}
\item[$\Modify$]
The process is the same as with $\vdash^{*}$: translating each node and adding extra sentence sets to the left and right of the node.
While it is necessary to verify the newly added rules, the details are omitted as they are almost identical to before.
\item[$\Init$]
This can also be done in the same way as before. (The $\acta^{*}$ part does not need to be discussed because $\Init_{\acta^{*}}$ is accepted.)
\item[$\Cut$]
This can also be done in the same way as before. (The $\acta^{*}$ part does not need to be discussed because $\Cut^{\acta^{*}}$ is accepted.)
\end{proofcases}
\end{proof}

\begin{proof}[Proof for Example~\ref{example:natural-Ind}]
To avoid inconvenience, we use $\phi:=\Forall{x\in\mathbb{N}}\gamma(x)\to\gamma(s(x))$,
and express it as $\gamma(0)\vdash^{\Ind,\INT}_{\Gamma\,\phi\mid\Sigma\mid}\Forall{z\in\mathbb{N}}\gamma(z)$.
It is sufficient to show that $\gamma(0)\,z\stackrel{suc^{*}}{\Longrightarrow}0\vdash^{\Ind,\INT}_{\Gamma\,\phi\mid\Sigma[z]\mid}\gamma(z)$.
We want to obtain this result by applying induction to $suc$ and $\gamma$,
but since we cannot directly apply induction to $\gamma$,
we strengthen our goal a bit and consider using $z\stackrel{\pi}{\Longrightarrow}t$ instead of $\gamma(z)$.
First, we remove the existence quantifier by $(\exists_{L})$.
$\gamma(0)\,z\stackrel{suc^{*}}{\Longrightarrow}0\vdash^{\Ind,\INT}_{\{def_{suc},def_{\pi,\gamma},\phi\}\mid\Sigma[z][\pi]\mid}\gamma(z)$
This holds true by induction if the following three conditions hold.
\begin{align*}
&\gamma(0)\vdash^{\Ind,\INT}_{\{def_{suc},def_{\pi,\gamma},\phi\}\mid\Sigma[z][\pi][wxy]\mid}0\stackrel{\pi}{\Longrightarrow}t&\\
&w\stackrel{suc}{\Longrightarrow}x\,x\stackrel{\pi}{\Longrightarrow}y\vdash^{\Ind,\INT}_{\{def_{suc},def_{\pi,\gamma}\,\phi\}\mid\Sigma[z][\pi][wxy]\mid}w\stackrel{\pi}{\Longrightarrow}y&\\
&z\stackrel{\pi}{\Longrightarrow}t\vdash^{\Ind,\INT}_{\{def_{suc},def_{\pi,\gamma},\phi\}\mid\Sigma[z][\pi][wxy]\mid}\gamma(z)&
\end{align*}
The first and last of these are obvious from $def_{\pi,\gamma}$.
The middle one can be shown from the following two.
\begin{align*}
&w\stackrel{suc}{\Longrightarrow}x\,x\stackrel{suc^{*}}{\Longrightarrow}0\vdash^{\Ind,\INT}_{\{def_{suc},def_{\pi,\gamma},\phi\}\mid\Sigma[z][\pi][wxy]\mid}w\stackrel{suc^{*}}{\Longrightarrow}0&\\
&w\stackrel{suc}{\Longrightarrow}x\,\gamma(x)\vdash^{\Ind,\INT}_{\{def_{suc},def_{\pi,\gamma},\phi\}\mid\Sigma[z][\pi][wxy]\mid}\gamma(w)&
\end{align*}
The above is obvious.\\
The lower part can be shown from $def_{suc}$ and
$w=s(x)\,\gamma(x)\vdash^{\Ind,\INT}_{\{def_{suc},def_{\pi,\gamma},\phi\}\mid\Sigma[z][\pi][wxy]\mid}\gamma(w)$, which can be shown from $\phi=\Forall{x\in\mathbb{N}}\gamma(x)\to\gamma(s(x))$.
\end{proof}

\begin{proof}[Proof for Example~\ref{example:meal}]\
\begin{proofcases}
\item[Right]
There is a transition $\pos{0,0}=\pos{0,0}\stackrel{get_{i}}{\Longrightarrow}\pos{i,0}\stackrel{get_{i}}{\Longrightarrow}\pos{i,i}\stackrel{eat_{i}}{\Longrightarrow}\pos{0,0}=\pos{0,0}\,(i\in\{a,b\})$.
The equalities become $\acta^{*}$ using $(\Ind_{R}^{0})$.
$get_{i}$ becomes $\acta$ using $({\cupact}_{R})$.
The two introduced $\acta$ are absorbed into $\acta^{*}$ by $(\Ind_{R}^{+})$.
By using $({\comp}_{R})$, we obtain
$\pos{0,0}\stackrel{\acta^{*}\comp eat_{i}\comp\acta^{*}}{====\Rightarrow}\pos{0,0}$.
\item[Left]
Since that is semantically correct on $\TA$, we can say that
$\Forall{xy}x\stackrel{(\acta\impact\act)^{\Co}}{\Longrightarrow}y\leftrightarrow x\stackrel{get_{a}\cupact get_{b}}{\Longrightarrow}y\in\Gamma$.
As shown later in Lemma~\ref{lemma:IndNK}, equivalent actions can be substituted for each other.
Therefore, we show $\vdash^{\Ind}_{\Gamma\mid\Sigma\mid}\pos{0,0}\stackrel{(get_{a}\cupact get_{b})^{+\Co}}{=====\Rightarrow}\pos{0,0}$.
Paying attention to the definitions of $\Co$ and $+$, it is sufficient to show that $\pos{0,0}\stackrel{(get_{a}\cupact get_{b})}{=====\Rightarrow}x\,x\stackrel{(get_{a}\cupact get_{b})^{*}}{=====\Rightarrow}\pos{0,0}\vdash^{\Ind}_{\Gamma\mid\Sigma[x]\mid}$.
Now we change the goal to apply the introduction of $*$ to the left.
Let $\lambda=[\pos{0,0}]\comp U\comp[\pos{0,0}]^{\Co}\,(U=\empact^{\Co})$.
The following hold true.
\begin{itemize}
\item $\pos{0,0}\stackrel{(get_{a}\cupact get_{b})}{=====\Rightarrow}x\vdash^{\Ind}_{\Gamma\mid\Sigma[x]\mid}\pos{0,0}\stackrel{\lambda}{\Longrightarrow}x$
\item $w'\stackrel{\lambda}{\Longrightarrow}x'\,x'\stackrel{(get_{a}\cupact get_{b})}{=====\Rightarrow}y'\vdash^{\Ind}_{\Gamma\mid\Sigma[x][w'x'y']\mid}w'\stackrel{\lambda}{\Longrightarrow}y'$
\item $\pos{0,0}\stackrel{\lambda}{\Longrightarrow}\pos{0,0}\vdash^{\Ind}_{\Gamma\mid\Sigma[x]\mid}$
\end{itemize}
$\lambda$ rewrites $\pos{0,0}$ to something else.
The application of $(get_{a}\cupact get_{b})$ is closed under the complement of $\pos{0,0}$, so the first two statements are true.
From the definition of $\lambda$, $\pos{0,0}\stackrel{\lambda}{\Longrightarrow}\pos{0,0}$ does not hold true, so the last statement is true.
Then, applying $(\Ind_{L}^{+})$ will give the desired result.
\item[Middle]
If we can prove that
$\pos{0,0}\stackrel{\acta'^{*}}{\Longrightarrow}x
\vdash^{\Ind}_{\Gamma\mid\Sigma[x]\mid}
x\stackrel{\acta'^{*}}{\Longrightarrow}\pos{0,0}$
then we can use $({\comp}_{L})$, $({\to}_{L})$, $({\to}_{R})$, etc., to obtain the conclusion.
We will change the goal.
Let $\lambda:=[\pos{0,0}]\comp U\,(U=\empact^{\Co})$
The following hold true.
\begin{itemize}
\item $\vdash^{\Ind}_{\Gamma\mid\Sigma[x]\mid}\pos{0,0}\stackrel{\lambda}{\Longrightarrow}\pos{0,0}$
\item $w'\stackrel{\lambda}{\Longrightarrow}x'\,x'\stackrel{\acta'}{\Longrightarrow}y'\vdash^{\Ind}_{\Gamma\mid\Sigma[x][w'x'y']\mid}x'\stackrel{\lambda}{\Longrightarrow}y'$
\item $\pos{0,0}\stackrel{\lambda}{\Longrightarrow}x\vdash^{\Ind}_{\Gamma\mid\Sigma[x]\mid}x\stackrel{\acta'^{*}}{\Longrightarrow}\pos{0,0}$
\end{itemize}
The first one can be said from the definition of $\lambda$.
From the definition of $\acta'$, it is possible to go to $\pos{0,0}$ from any state, so the latter two statements hold true.
By applying $(\Ind_{R}^{+})$, we get
$
\pos{0,0}\stackrel{\acta'^{*}}{\Longrightarrow}x
\vdash^{\Ind}_{\Gamma\mid\Sigma[x]\mid}
x\stackrel{\acta'^{*}}{\Longrightarrow}\pos{0,0}
$.
\end{proofcases}
\end{proof}


\begin{proof}[Proof of Lemma~\ref{lemma:ind-is-sound-on-ta}]
From the previous lemma, the rule for $\vdash^{\Ind}$ can be shown from $\vdash^{*}$. Since $\vdash^{*}$ is sound, $\vdash^{\Ind}$ is also sound.
Compactness can be easily proven by induction on the complexity of the proof tree. On the other hand, $\vDash$ is not compact, so completeness does not hold true.
\end{proof}

\section{Proof of the results presented in Section~\ref{sec:complete}}

\begin{proof}[Proof of well-definedness of Definition~\ref{def:Kleene-algebraic-TA}]
As we did before, $\acta^{\A\red_{\chi}}=\T(\chi)(\acta)^{\A}$ (for $\chi\in\Sig^{\gene}$, $\T_{\chi}(\acta)^{\A}$) is important.
The proof is induction on the structure, but only the $*$ part changes as follows.
\begin{proofcases}
\item[$\acta=\acta_{0}^{*}$]
$(\acta_{0}^{*})^{\A\red_{\chi}}
=(\acta_{0}^{\A\red_{\chi}})^{*_{\BS{s}}^{\A\red_{\chi}}}
\stackrel{IH}{=}(\T_{\chi}(\acta_{0})^{\A})^{*_{\BS{s}}^{\A\red_{\chi}}}
=(\T_{\chi}(\acta_{0})^{\A})^{*_{\chi(\BS{s})}^{\A}}
=(\T_{\chi}(\acta_{0})^{*})^{\A}
=\T_{\chi}(\acta_{0}^{*})^{\A}$
\end{proofcases}
The rest is the same.
\end{proof}

\begin{proof}[Proof of Lemma~\ref{lemma:IndNK}]
Note that
$t_{0}\stackrel{\acta_{0}}{\Longrightarrow}t_{1}\vdash^{\Ind,\INT}_{\acta_{0}\leqq\acta_{1}\mid\Sigma\mid}t_{0}\stackrel{\acta_{1}}{\Longrightarrow}t_{1}$
holds true from
$t_{0}\stackrel{\acta_{0}}{\Longrightarrow}t_{1}\vdash^{\Ind,\INT}_{\Sigma}t_{0}\stackrel{\acta_{0}}{\Longrightarrow}t_{1}$,
$t_{0}\stackrel{\acta_{1}}{\Longrightarrow}t_{1}\vdash^{\Ind,\INT}_{\Sigma}t_{0}\stackrel{\acta_{1}}{\Longrightarrow}t_{1}$,
$({\to}_{L})$, and
$({\forall}_{L})$.
\begin{proofcases}
\item[action sentence]
Since $t'_{0}\stackrel{\acta}{\Longrightarrow}t'_{1}\vdash^{\Ind,\INT}_{\acta\equiv\acta'\mid\Sigma\mid}t'_{0}\stackrel{\acta'}{\Longrightarrow}t'_{1}$,
if we show $t_{0}=t'_{0}\,t_{1}=t'_{1}\,t_{0}\stackrel{\acta}{\Longrightarrow}t_{1}\vdash^{\Ind,\INT}_{\Sigma}t'_{0}\stackrel{\acta}{\Longrightarrow}t'_{1}$,
all that remains is to apply $(\Cut)$, $(\wedge_{L})$, and $({\to}_{R})$.
Therefore, we will prove this by induction on the complexity of $\acta$.
\begin{proofcases}
\item[$\acta=\empact$] By $({\empact}_{L})$.
\item[$\acta=\acta_{0}\cupact\acta_{1}$]
From the recursion hypothesis
$t_{0}=t'_{0}\,t_{1}=t'_{1}\,t_{0}\stackrel{\acta_{i}}{\Longrightarrow}t_{1}\vdash^{\Ind,\INT}_{\Sigma}t'_{0}\stackrel{\acta_{i}}{\Longrightarrow}t'_{1}\,(i\in2)$.
By applying $({\cupact}_{R})$ to each of these once, we get
$t_{0}=t'_{0}\,t_{1}=t'_{1}\,t_{0}\stackrel{\acta_{i}}{\Longrightarrow}t_{1}\vdash^{\Ind,\INT}_{\Sigma}t'_{0}\stackrel{\acta_{0}\cupact\acta_{1}}{\Longrightarrow}t'_{1}\,(i\in2)$.
By combining these with $({\cupact}_{L})$, we get
$t_{0}=t'_{0}\,t_{1}=t'_{1}\,t_{0}\stackrel{\acta_{0}\cupact\acta_{1}}{\Longrightarrow}t_{1}\vdash^{\Ind,\INT}_{\Sigma}t'_{0}\stackrel{\acta_{0}\cupact\acta_{1}}{\Longrightarrow}t'_{1}$.
\item[$\acta=\acta_{0}\impact\acta_{1}$]
From the recursion hypothesis
$t_{0}=t'_{0}\,t_{1}=t'_{1}\,t'_{0}\stackrel{\acta_{0}}{\Longrightarrow}t'_{1}\vdash^{\Ind,\INT}_{\Sigma}t_{0}\stackrel{\acta_{0}}{\Longrightarrow}t_{1}$, and
$t_{0}=t'_{0}\,t_{1}=t'_{1}\,t_{0}\stackrel{\acta_{1}}{\Longrightarrow}t_{1}\vdash^{\Ind,\INT}_{\Sigma}t'_{0}\stackrel{\acta_{1}}{\Longrightarrow}t'_{1}$ hold true.
By combining these with $({\impact}_{L})$, we get
$t_{0}=t'_{0}\,t_{1}=t'_{1}\,t_{0}\stackrel{\acta_{0}\impact\acta_{1}}{\Longrightarrow}t_{1}\,t'_{0}\stackrel{\acta_{0}}{\Longrightarrow}t'_{1}\vdash^{\Ind,\INT}_{\Sigma}t'_{0}\stackrel{\acta_{1}}{\Longrightarrow}t'_{1}$.
By applying $({\impact}_{R})$ once, we get $t_{0}=t'_{0}\,t_{1}=t'_{1}\,t_{0}\stackrel{\acta_{0}\impact\acta_{1}}{\Longrightarrow}t_{1}\vdash^{\Ind,\INT}_{\Sigma}t'_{0}\stackrel{\acta_{0}\impact\acta_{1}}{\Longrightarrow}t'_{1}$.
\item[$\acta=\one$]
Since $t_{0}=t'_{0}\,t_{1}=t'_{1}\,t_{0}=t_{1}\vdash^{b}_{\Sigma}t'_{0}=t'_{1}$,
by $({\one}_{L})$, $({\one}_{R})$, we get
$t_{0}=t'_{0}\,t_{1}=t'_{1}\,t_{0}\stackrel{\one}{\Longrightarrow}t_{1}\vdash^{\Ind,\INT}_{\Sigma}t'_{0}\stackrel{\one}{\Longrightarrow}t'_{1}$.
\item[$\acta=\acta_{0}\comp\acta_{1}$]
From the recursion hypothesis
$t_{0}=t'_{0}\,t_{0}\stackrel{\acta_{0}}{\Longrightarrow}x\vdash^{\Ind,\INT}_{\Sigma[x]}t'_{0}\stackrel{\acta_{0}}{\Longrightarrow}x$, and
$t_{1}=t'_{1}\,x\stackrel{\acta_{1}}{\Longrightarrow}t_{1}\vdash^{\Ind,\INT}_{\Sigma[x]}x\stackrel{\acta_{1}}{\Longrightarrow}t'_{1}$ hold true.
By $({\comp}_{R})$,
$t_{0}=t'_{0}\,t_{1}=t'_{1}\,t_{0}\stackrel{\acta_{0}}{\Longrightarrow}x\,x\stackrel{\acta_{1}}{\Longrightarrow}t_{1}\vdash^{\Ind,\INT}_{\Sigma[x]}t'_{0}\stackrel{\acta_{0}\comp\acta_{1}}{\Longrightarrow}t'_{1}$ holds true.
By $({\comp}_{L})$, we get
$t_{0}=t'_{0}\,t_{1}=t'_{1}\,t_{0}\stackrel{\acta_{0}\comp\acta_{1}}{\Longrightarrow}t_{1}\vdash^{\Ind,\INT}_{\Sigma}t'_{0}\stackrel{\acta_{0}\comp\acta_{1}}{\Longrightarrow}t'_{1}$.
\item[$\acta=\acta_{0}\resi\acta_{1}$]
From the recursion hypothesis
$t_{0}=t'_{0}\,x\stackrel{\acta_{0}}{\Longrightarrow}t_{0}\vdash^{\Ind,\INT}_{\Sigma[x]}x\stackrel{\acta_{0}}{\Longrightarrow}t'_{0}$, and
$t_{1}=t'_{1}\,x\stackrel{\acta_{1}}{\Longrightarrow}t_{1}\vdash^{\Ind,\INT}_{\Sigma[x]}x\stackrel{\acta_{1}}{\Longrightarrow}t'_{1}$ hold true.
By $({\resi}_{R})$,
$t_{0}=t'_{0}\,t_{1}=t'_{1}\,x\stackrel{\acta_{0}}{\Longrightarrow}t_{0}\,x\stackrel{\acta_{1}}{\Longrightarrow}t_{1}\vdash^{\Ind,\INT}_{\Sigma[x]}t'_{0}\stackrel{\acta_{0}\resi\acta_{1}}{\Longrightarrow}t'_{1}$ holds true.
By $({\resi}_{L})$, we get
$t_{0}=t'_{0}\,t_{1}=t'_{1}\,t_{0}\stackrel{\acta_{0}\resi\acta_{1}}{\Longrightarrow}t_{1}\vdash^{\Ind,\INT}_{\Sigma}t'_{0}\stackrel{\acta_{0}\resi\acta_{1}}{\Longrightarrow}t'_{1}$.
\item[$\acta=\acta_{0}^{*}$]
First since $x=w\vdash^{b}_{\Sigma[wxyz]}w=x$, by $(\Ind_{R}^{0})$ we get $x=w\vdash^{\Ind,\INT}_{\Sigma[wxyz]}w\stackrel{\acta_{0}^{*}}{\Longrightarrow}x$.
By $(\Init_{\acta_{0}^{*}})$, $w'\stackrel{\acta_{0}^{*}}{\Longrightarrow}x'\vdash^{\Ind,\INT}_{\Sigma[wxyz][w'x'y'z']}w'\stackrel{\acta_{0}^{*}}{\Longrightarrow}x'$ holds true.
From the recursion hypothesis and $(\Init)$ for equation, we get $x'\stackrel{\acta_{0}}{\Longrightarrow}z'\vdash^{\Ind,\INT}_{\Sigma[wxyz][w'x'y'z']}x'\stackrel{\acta_{0}}{\Longrightarrow}z'$.
By applying $(\Ind_{R}^{+x'})$ to them we get $w'\stackrel{\acta_{0}^{*}}{\Longrightarrow}x'\,x'\stackrel{\acta_{0}}{\Longrightarrow}z'\vdash^{\Ind,\INT}_{\Sigma[wxyz][w'x'y'z']}w'\stackrel{\acta_{0}^{*}}{\Longrightarrow}z'$.
By $(\Init_{\acta_{0}^{*}})$, we get $w\stackrel{\acta_{0}^{*}}{\Longrightarrow}z\vdash^{\Ind,\INT}_{\Sigma[wxyz]}w\stackrel{\acta_{0}^{*}}{\Longrightarrow}z$.
\\
{\small
$x=w\vdash^{\Ind,\INT}_{\Sigma[wxyz]}w\stackrel{\acta_{0}^{*}}{\Rightarrow}x$,
$w'\stackrel{\acta_{0}^{*}}{\Rightarrow}x'\,x'\stackrel{\acta_{0}}{\Rightarrow}z'\vdash^{\Ind,\INT}_{\Sigma[wxyz][w'x'y'z']}w'\stackrel{\acta_{0}^{*}}{\Rightarrow}z'$,
$w\stackrel{\acta_{0}^{*}}{\Rightarrow}z\vdash^{\Ind,\INT}_{\Sigma[wxyz]}w\stackrel{\acta_{0}^{*}}{\Rightarrow}z$.
}
\\
By applying $(\Ind_{L}^{+t\,\acta_{0}^{*}})$ to the above three, we obtain $x=w\,x\stackrel{\acta_{0}^{*}}{\Longrightarrow}z\vdash^{\Ind,\INT}_{\Sigma[wxyz]}w\stackrel{\acta_{0}^{*}}{\Longrightarrow}z$.
{\par}
Similarly, by using $(\Ind_{R}^{-})$ and the above result, we get
\\
$y=z\vdash^{\Ind,\INT}_{\Sigma[wxyz]}y\stackrel{\acta_{0}^{*}}{\Longrightarrow}z$
$\Space$
$x'\stackrel{\acta_{0}}{\Longrightarrow}y'\,y'\stackrel{\acta_{0}^{*}}{\Longrightarrow}z'\vdash^{\Ind,\INT}_{\Sigma[wxyz][w'x'y'z']}x'\stackrel{\acta_{0}^{*}}{\Longrightarrow}z'$
$\Space$
$x=w\,x\stackrel{\acta_{0}^{*}}{\Longrightarrow}z\vdash^{\Ind,\INT}_{\Sigma[wxyz]}w\stackrel{\acta_{0}^{*}}{\Longrightarrow}z$
\\
By applying $(\Ind_{L}^{-z\,\acta_{0}^{*}})$ to the above three, we obtain $x=w\,y=z\,x\stackrel{\acta_{0}^{*}}{\Longrightarrow}y\vdash^{\Ind,\INT}_{\Sigma[wxyz]}w\stackrel{\acta_{0}^{*}}{\Longrightarrow}z$.
{\par}
Finally, we translate this process to obtain $t_{0}=t'_{0}\,t_{1}=t'_{1}\,t_{0}\stackrel{\acta_{0}^{*}}{\Longrightarrow}t_{1}\vdash^{\Ind,\INT}_{\Sigma}t'_{0}\stackrel{\acta_{0}^{*}}{\Longrightarrow}t'_{1}$.
\end{proofcases}
\item[order and monotonicity]
Using $(\Init)$ and $(\Cut)$, it is immediately clear that this is an order.
Therefore, we will only check for monotonicity in part.
\begin{proofcases}
\item[$\cupact$]
$(\acta_{i}\leqq\acta'_{i})\,x\stackrel{\acta_{i}}{\Longrightarrow}y\vdash^{\Ind,\INT}_{\Sigma[x,y]}x\stackrel{\acta'_{i}}{\Longrightarrow}y\,(i\in2)$ holds true.
By applying $({\cupact}_{R})$ to both, we get
$(\acta_{i}\leqq\acta'_{i})\,x\stackrel{\acta_{i}}{\Longrightarrow}y\vdash^{\Ind,\INT}_{\Sigma[x,y]}x\stackrel{\acta'_{0}\cupact\acta'_{1}}{\Longrightarrow}y\,(i\in2)$.
By combining these with $({\cupact}_{L})$, we get
$(\acta_{0}\leqq\acta'_{0})\,(\acta_{1}\leqq\acta'_{1})\,x\stackrel{\acta_{0}\cupact\acta_{1}}{\Longrightarrow}y\vdash^{\Ind,\INT}_{\Sigma[x,y]}x\stackrel{\acta'_{0}\cupact\acta'_{1}}{\Longrightarrow}y$.
Then, by applying $({\to}_{R})$ and $({\forall}_{R})$ several times, we can obtain
$\vdash^{\Ind,\INT}_{\Sigma}(\acta_{0}\leqq\acta'_{0})\wedge(\acta_{1}\leqq\acta'_{1})\to(\acta_{0}\cupact\acta_{1})\leqq(\acta'_{0}\cupact\acta'_{1})$.
\item[$\comp$]
$(\acta_{0}\leqq\acta'_{0})\,x\stackrel{\acta_{0}}{\Longrightarrow}w\vdash^{\Ind,\INT}_{\Sigma[w,x,y]}x\stackrel{\acta'_{0}}{\Longrightarrow}w$, and
$(\acta_{1}\leqq\acta'_{1})\,w\stackrel{\acta_{1}}{\Longrightarrow}y\vdash^{\Ind,\INT}_{\Sigma[w,x,y]}w\stackrel{\acta'_{1}}{\Longrightarrow}y$ hold true.
By combining these with $({\comp}_{R})$, we get
$(\acta_{0}\leqq\acta'_{0})\,(\acta_{1}\leqq\acta'_{1})\,
x\stackrel{\acta_{0}}{\Longrightarrow}w\,w\stackrel{\acta'_{0}}{\Longrightarrow}y
\vdash^{\Ind,\INT}_{\Sigma[w,x,y]}x\stackrel{\acta'_{0}\comp\acta'_{1}}{\Longrightarrow}y$.
By applying $({\comp}_{L})$, we get
$(\acta_{0}\leqq\acta'_{0})\,(\acta_{1}\leqq\acta'_{1})\,x\stackrel{\acta_{0}\comp\acta_{1}}{\Longrightarrow}y\vdash^{\Ind,\INT}_{\Sigma[x,y]}x\stackrel{\acta'_{0}\comp\acta'_{1}}{\Longrightarrow}y$.
Then, by applying $({\to}_{R})$ and $({\forall}_{R})$ several times, we can obtain
$\vdash^{\Ind,\INT}_{\Sigma}(\acta_{0}\leqq\acta'_{0})\wedge(\acta_{1}\leqq\acta'_{1})\to(\acta_{0}\comp\acta_{1})\leqq(\acta'_{0}\comp\acta'_{1})$.
\item[$\resi$]
$(\acta_{0}\leqq\acta'_{0})\,w\stackrel{\acta_{0}}{\Longrightarrow}x\vdash^{\Ind,\INT}_{\Sigma[w,x,y]}w\stackrel{\acta'_{0}}{\Longrightarrow}x$, and
$(\acta_{1}\leqq\acta'_{1})\,w\stackrel{\acta_{1}}{\Longrightarrow}y\vdash^{\Ind,\INT}_{\Sigma[w,x,y]}w\stackrel{\acta'_{1}}{\Longrightarrow}y$ hold true.
By combining these with $({\resi}_{R})$, we get
$(\acta_{0}\leqq\acta'_{0})\,(\acta_{1}\leqq\acta'_{1})\,
w\stackrel{\acta_{0}}{\Longrightarrow}x\,w\stackrel{\acta'_{0}}{\Longrightarrow}y
\vdash^{\Ind,\INT}_{\Sigma[w,x,y]}x\stackrel{\acta'_{0}\resi\acta'_{1}}{\Longrightarrow}y$.
By applying $({\resi}_{L})$, we get
$(\acta_{0}\leqq\acta'_{0})\,(\acta_{1}\leqq\acta'_{1})\,x\stackrel{\acta_{0}\resi\acta_{1}}{\Longrightarrow}y\vdash^{\Ind,\INT}_{\Sigma[x,y]}x\stackrel{\acta'_{0}\resi\acta'_{1}}{\Longrightarrow}y$.
Then, by applying $({\to}_{R})$ and $({\forall}_{R})$ several times, we can obtain
$\vdash^{\Ind,\INT}_{\Sigma}(\acta_{0}\leqq\acta'_{0})\wedge(\acta_{1}\leqq\acta'_{1})\to(\acta_{0}\resi\acta_{1})\leqq(\acta'_{0}\resi\acta'_{1})$.
\item[$\impact$]
$(\acta'_{0}\leqq\acta_{0})\,x\stackrel{\acta'_{0}}{\Longrightarrow}y\vdash^{\Ind,\INT}_{\Sigma[x,y]}x\stackrel{\acta_{0}}{\Longrightarrow}y$, and
$(\acta_{1}\leqq\acta'_{1})\,x\stackrel{\acta_{1}}{\Longrightarrow}y\vdash^{\Ind,\INT}_{\Sigma[x,y]}x\stackrel{\acta'_{1}}{\Longrightarrow}y$ hold true.
By combining these with $({\impact}_{L})$, we get
$(\acta'_{0}\leqq\acta_{0})\,(\acta_{1}\leqq\acta'_{1})\,x\stackrel{\acta_{0}\impact\acta_{1}}{\Longrightarrow}y\,x\stackrel{\acta'_{0}}{\Longrightarrow}y\vdash^{\Ind,\INT}_{\Sigma[x,y]}x\stackrel{\acta'_{1}}{\Longrightarrow}y$.
By applying $({\impact}_{R})$, we get
$(\acta'_{0}\leqq\acta_{0})\,(\acta_{1}\leqq\acta'_{1})\,x\stackrel{\acta_{0}\impact\acta_{1}}{\Longrightarrow}y\vdash^{\Ind,\INT}_{\Sigma[x,y]}x\stackrel{\acta'_{0}\impact\acta'_{1}}{\Longrightarrow}y$.
Then, by applying $({\to}_{R})$ and $({\forall}_{R})$ several times, we can obtain
$\vdash^{\Ind,\INT}_{\Sigma}(\acta'_{0}\leqq\acta_{0})\wedge(\acta_{1}\leqq\acta'_{1})\to(\acta_{0}\impact\acta_{1})\leqq(\acta'_{0}\impact\acta'_{1})$.
\item[$*$]
$(\acta_{0}\leqq\acta'_{0})\,x'\stackrel{\acta_{0}}{\Longrightarrow}y'\vdash^{\Ind,\INT}_{\Sigma[xy][w'x'y']}x'\stackrel{\acta'_{0}}{\Longrightarrow}y'$ holds true.
From $(\Init_{\acta^{*}})$, $w'\stackrel{{\acta'}_{0}^{*}}{\Longrightarrow}x'\vdash^{\Ind,\INT}_{\Sigma[xy][w'x'y']}w'\stackrel{{\acta'}_{0}^{*}}{\Longrightarrow}x'$ holds true.
By combining these with $(\Ind_{R}^{+})$ we get
$(\acta_{0}\leqq\acta'_{0})\,w'\stackrel{{\acta'}_{0}^{*}}{\Longrightarrow}x'\,x'\stackrel{\acta_{0}}{\Longrightarrow}y'\vdash^{\Ind,\INT}_{\Sigma[xy][w'x'y']}x'\stackrel{{\acta'}_{0}^{*}}{\Longrightarrow}y'$.
From $(\Ind_{R}^{0})$,
$\vdash^{\Ind,\INT}_{\Sigma[xy]}x\stackrel{{\acta'}_{0}^{*}}{\Longrightarrow}x$ holds true.
From $(\Init_{\acta^{*}})$,
$x\stackrel{{\acta'}_{0}^{*}}{\Longrightarrow}y\vdash^{\Ind,\INT}_{\Sigma[xy]}x\stackrel{{\acta'}_{0}^{*}}{\Longrightarrow}y$ holds true.
By combining these three with $(\Ind_{L}^{+})$, we get
$(\acta_{0}\leqq\acta'_{0})\,x\stackrel{\acta_{0}^{*}}{\Longrightarrow}y\vdash^{\Ind,\INT}_{\Sigma[xy]}x\stackrel{{\acta'}_{0}^{*}}{\Longrightarrow}y$.
Then, by applying $({\to}_{R})$ and $({\forall}_{R})$ several times, we can obtain
$\vdash^{\Ind,\INT}_{\Sigma}(\acta'_{0}\leqq\acta_{0})\to(\acta_{0}^{*})\leqq({\acta'}_{0}^{*})$.
\end{proofcases}
\item[Kleene algebra]\
\begin{proofcases}
\item[proof of axioms] We only check $\one\leqq\acta^{*}$, $\acta\comp\acta^{*}\leqq\acta^{*}$, $\acta_{1}\comp\acta_{2}\leqq\acta_{1}\to\acta_{1}\comp\acta_{2}^{*}\leqq\acta_{1}$.
\begin{proofcases}
\item[$\one\leqq\acta^{*}$]
\begin{proofsteps}{25.6em}
\label{ps:R0:1}
$x=y\vdash^{\Ind,\INT}_{\Sigma[x,y]}x=y$
&by $(\Iy)$
\\ \label{ps:R0:2}
$x\stackrel{\one}{\Longrightarrow}y\vdash^{\Ind,\INT}_{\Sigma[x,y]}x\stackrel{\acta^{*}}{\Longrightarrow}y$
&by \ref{ps:R0:1}, $({\one}_L)$, and $(\Ind_{R}^{0})$
\\ \label{ps:R0:3}
$\vdash^{\Ind,\INT}_{\Sigma}\Forall{x,y}x\stackrel{\one}{\Longrightarrow}y\to x\stackrel{\acta^{*}}{\Longrightarrow}y$
&by \ref{ps:R0:2}, $({\to}_R)$, and $(\forall_{R})$
\end{proofsteps}
\item[$\acta^{*}\comp\acta\leqq\acta^{*}$]
\begin{proofsteps}{25.6em}
\label{ps:R+:1}
$x\stackrel{\acta^{*}}{\Longrightarrow}y\vdash^{\Ind,\INT}_{\Sigma[x,y,z]}x\stackrel{\acta^{*}}{\Longrightarrow}y$
&by $(\Init)$
\\ \label{ps:R+:2}
$y\stackrel{\acta}{\Longrightarrow}z\vdash^{\Ind,\INT}_{\Sigma[x,y,z]}y\stackrel{\acta}{\Longrightarrow}z$
&by $(\Init)$
\\ \label{ps:R+:3}
$x\stackrel{\acta^{*}}{\Longrightarrow}y\,y\stackrel{\acta}{\Longrightarrow}z\vdash^{\Ind,\INT}_{\Sigma[x,y,z]}x\stackrel{\acta^{*}}{\Longrightarrow}z$
&by \ref{ps:R+:1}, \ref{ps:R+:2}, and $(\Ind_{R}^{+})$
\\ \label{ps:R+:4}
$x\stackrel{\acta^{*}\comp\acta}{\Longrightarrow}z\vdash^{\Ind,\INT}_{\Sigma[x,z]}x\stackrel{\acta^{*}}{\Longrightarrow}z$
&by \ref{ps:R+:3}, and $(\comp_{L})$
\\ \label{ps:R+:5}
$\vdash^{\Ind,\INT}_{\Sigma}\Forall{x,z}x\stackrel{\acta^{*}\comp\acta}{\Longrightarrow}z\to x\stackrel{\acta^{*}}{\Longrightarrow}z$
&by \ref{ps:R+:4}, $({\to}_R)$, and $(\forall_{R})$
\end{proofsteps}
\item[$\acta_{1}\comp\acta_{2}\leqq\acta_{1}\to\acta_{1}\comp\acta_{2}^{*}\leqq\acta_{1}$]
\begin{proofsteps}{28.0em}
\label{ps:L+:-10}
$x'\stackrel{\acta_{1}}{\Longrightarrow}z'
\vdash^{\Ind,\INT}_{\Sigma[xyz][x'y'z']}
x'\stackrel{\acta_{1}}{\Longrightarrow}z'$
&by $(\Init)$
\\ \label{ps:L+:-9}
$z'\stackrel{\acta_{2}}{\Longrightarrow}y'
\vdash^{\Ind,\INT}_{\Sigma[xyz][x'y'z']}
z'\stackrel{\acta_{2}}{\Longrightarrow}y'$
&by $(\Init)$
\\ \label{ps:L+:-8}
$x'\stackrel{\acta_{1}}{\Longrightarrow}z'\,
z'\stackrel{\acta_{2}}{\Longrightarrow}y'
\vdash^{\Ind,\INT}_{\Sigma[xyz][x'y'z']}
x'\stackrel{\acta_{1}\comp\acta_{2}}{\Longrightarrow}y'$
&by \ref{ps:L+:-10}, \ref{ps:L+:-9}, and $(\comp_{R})$
\\ \label{ps:L+:-7}
$x'\stackrel{\acta_{1}}{\Longrightarrow}y'
\vdash^{\Ind,\INT}_{\Sigma[xyz][x'y'z']}
x'\stackrel{\acta_{1}}{\Longrightarrow}y'$
&by $(\Init)$
\\ \label{ps:L+:-6}
$x\stackrel{\acta_{1}}{\Longrightarrow}z
\vdash^{\Ind,\INT}_{\Sigma[x,y,z]}
x\stackrel{\acta_{1}}{\Longrightarrow}z$
&by $(\Init)$
\\ \label{ps:L+:-5}
$\Forall{x,y}(x\stackrel{\acta_{1}\comp\acta_{2}}{\Longrightarrow}y\to x\stackrel{\acta_{1}}{\Longrightarrow}y)$\par
$x'\stackrel{\acta_{1}}{\Longrightarrow}z'\,
z'\stackrel{\acta_{2}}{\Longrightarrow}y'
\vdash^{\Ind,\INT}_{\Sigma[xyz][x'y'z']}
x'\stackrel{\acta_{1}}{\Longrightarrow}y'$
&by \ref{ps:L+:-8}, \ref{ps:L+:-7}, $({\to}_{L})$, and $(\forall_{L})$
\\ \label{ps:L+:-4}
$x\stackrel{\acta_{1}}{\Longrightarrow}y
\vdash^{\Ind,\INT}_{\Sigma[x,y,z]}
x\stackrel{\acta_{1}}{\Longrightarrow}y$
&by $(\Init)$
\\ \label{ps:L+:-3}
$\Forall{x,y}(x\stackrel{\acta_{1}\comp\acta_{2}}{\Longrightarrow}y\to x\stackrel{\acta_{1}}{\Longrightarrow}y)$\par
$x\stackrel{\acta_{1}}{\Longrightarrow}z\,
z\stackrel{\acta_{2}^{*}}{\Longrightarrow}y
\vdash^{\Ind,\INT}_{\Sigma[x,y,z]}
x\stackrel{\acta_{1}}{\Longrightarrow}y$
&by \ref{ps:L+:-6}, \ref{ps:L+:-5}, \ref{ps:L+:-4}, and $(\Ind_{L}^{+})$
\\ \label{ps:L+:-2}
$\Forall{x,y}(x\stackrel{\acta_{1}\comp\acta_{2}}{\Longrightarrow}y\to x\stackrel{\acta_{1}}{\Longrightarrow}y)\,
x\stackrel{\acta_{1}\comp\acta_{2}^{*}}{\Longrightarrow}y
\vdash^{\Ind,\INT}_{\Sigma[x,y]}
x\stackrel{\acta_{1}}{\Longrightarrow}y$
&by \ref{ps:L+:-3}, and $(\comp_{L})$
\\ \label{ps:L+:-1}
$\Forall{x,y}x\stackrel{\acta_{1}\comp\acta_{2}}{\Longrightarrow}y\to x\stackrel{\acta_{1}}{\Longrightarrow}y
\vdash^{\Ind,\INT}_{\Sigma}
\Forall{x,y}x\stackrel{\acta_{1}\comp\acta_{2}^{*}}{\Longrightarrow}y\to x\stackrel{\acta_{1}}{\Longrightarrow}y$
&by \ref{ps:L+:-2}, $({\to}_R)$, and $(\forall_{R})$
\\ \label{ps:L+:0}
$\vdash^{\Ind,\INT}_{\Sigma}
\Forall{x,y}x\stackrel{\acta_{1}\comp\acta_{2}}{\Longrightarrow}y\to x\stackrel{\acta_{1}}{\Longrightarrow}y\to
\Forall{x,y}x\stackrel{\acta_{1}\comp\acta_{2}^{*}}{\Longrightarrow}y\to x\stackrel{\acta_{1}}{\Longrightarrow}y$
&by \ref{ps:L+:-1}, and $({\to}_R)$
\end{proofsteps}
\end{proofcases}
\item[sufficiency of axioms] We only check $(\Ind_{R}^{0})$, $(\Ind_{R}^{+})$, $(\Ind_{L}^{+})$.
\begin{proofcases}
\item[$\Ind_{R}^{0}$]
Let $\vdash^{\Ind}_{\Gamma\mid\Sigma\mid\Delta}t_{1}=t_{2}$.
We show $\vdash^{\Ind}_{\Gamma\mid\Sigma\mid\Delta}t_{1}\stackrel{\acta^{*}}{\Longrightarrow}t_{2}$.
\begin{proofsteps}{25.6em}
\label{ps:R0':-4}
$\vdash^{\Ind}_{\Gamma\mid\Sigma\mid\Delta}t_{1}=t_{2}$
&by the assumption
\\ \label{ps:R0':-3}
$\vdash^{\Ind}_{\Gamma\mid\Sigma\mid\Delta}t_{1}\stackrel{\one}{\Longrightarrow}t_{2}$
&by \ref{ps:R0':-4}, and $({\one}_{R})$
\\ \label{ps:R0':-2}
$t_{1}\stackrel{\acta^{*}}{\Longrightarrow}t_{2}
\vdash^{\Ind}_{\Gamma\mid\Sigma\mid\Delta}t_{1}\stackrel{\acta^{*}}{\Longrightarrow}t_{2}$
&by $(\Init_{\acta^{*}})$
\\ \label{ps:R0':-1}
$\Forall{x,y}x\stackrel{\one}{\Longrightarrow}y\to x\stackrel{\acta^{*}}{\Longrightarrow}y
\vdash^{\Ind}_{\Gamma\mid\Sigma\mid\Delta}t_{1}\stackrel{\acta^{*}}{\Longrightarrow}t_{2}$
&by \ref{ps:R0':-3}, \ref{ps:R0':-2}, $({\to}_{L})$, and $(\forall_{L})$
\\ \label{ps:R0':0}
$\vdash^{\Ind}_{\Gamma\mid\Sigma\mid\Delta}t_{1}\stackrel{\acta^{*}}{\Longrightarrow}t_{2}$
&by \ref{ps:R0':-1}, and $(\KEL_{\one\leqq\acta^{*}})$
\end{proofsteps}
\item[$\Ind_{R}^{+}$]
Let $\vdash^{\Ind}_{\Gamma\mid\Sigma\mid\Delta}t_{1}\stackrel{\acta^{*}}{\Longrightarrow}t$ and
$\vdash^{\Ind}_{\Gamma\mid\Sigma\mid\Delta}t\stackrel{\acta}{\Longrightarrow}t_{2}$.
We show $\vdash^{\Ind}_{\Gamma\mid\Sigma\mid\Delta}t_{1}\stackrel{\acta^{*}}{\Longrightarrow}t_{2}$
\begin{proofsteps}{25.6em}
\label{ps:R+':-5}
$\vdash^{\Ind}_{\Gamma\mid\Sigma\mid\Delta}t_{1}\stackrel{\acta^{*}}{\Longrightarrow}t$
&by the assumption
\\ \label{ps:R+':-4}
$\vdash^{\Ind}_{\Gamma\mid\Sigma\mid\Delta}t\stackrel{\acta}{\Longrightarrow}t_{2}$
&by the assumption
\\ \label{ps:R+':-3}
$\vdash^{\Ind}_{\Gamma\mid\Sigma\mid\Delta}t_{1}\stackrel{\acta^{*}\comp\acta}{\Longrightarrow}t_{2}$
&by \ref{ps:R+':-5}, \ref{ps:R+':-4}, and $(\comp_{R})$
\\ \label{ps:R+':-2}
$t_{1}\stackrel{\acta^{*}}{\Longrightarrow}t_{2}
\vdash^{\Ind}_{\Gamma\mid\Sigma\mid\Delta}t_{1}\stackrel{\acta^{*}}{\Longrightarrow}t_{2}$
&by $(\Init_{\acta^{*}})$
\\ \label{ps:R+':-1}
$\Forall{x,z}x\stackrel{\acta^{*}\comp\acta}{\Longrightarrow}z\to x\stackrel{\acta^{*}}{\Longrightarrow}z
\vdash^{\Ind}_{\Gamma\mid\Sigma\mid\Delta}t_{1}\stackrel{\acta^{*}}{\Longrightarrow}t_{2}$
&by \ref{ps:R+':-3}, \ref{ps:R+':-2}, $({\to}_{L})$, and $(\forall_{L})$
\\ \label{ps:R+':0}
$\vdash^{\Ind}_{\Gamma\mid\Sigma\mid\Delta}t_{1}\stackrel{\acta^{*}}{\Longrightarrow}t_{2}$
&by \ref{ps:R+':-1}, and $(\KEL_{\acta^{*}\comp\acta\leqq\acta^{*}})$
\end{proofsteps}
\item[$\Ind_{L}^{+}$]
Let $\vdash^{\Ind}_{\Gamma\mid\Sigma\mid\Delta}t\stackrel{\acta_{1}}{\Longrightarrow}t_{1}$,
$x\stackrel{\acta_{1}}{\Longrightarrow}z\,z\stackrel{\acta_{2}}{\Longrightarrow}y\vdash^{\Ind}_{\Gamma\mid\Sigma[xyz]\mid\Delta}x\stackrel{\acta_{1}}{\Longrightarrow}y$, 
$t\stackrel{\acta_{1}}{\Longrightarrow}t_{2}\vdash^{\Ind}_{\Gamma\mid\Sigma\mid\Delta}$.
We show $t_{1}\stackrel{\acta_{2}^{*}}{\Longrightarrow}t_{2}\vdash^{\Ind}_{\Gamma\mid\Sigma\mid\Delta}$.
\begin{proofsteps}{25.6em}
\label{ps:L+':-9}
$x\stackrel{\acta_{1}}{\Longrightarrow}z\,z\stackrel{\acta_{2}}{\Longrightarrow}y
\vdash^{\Ind}_{\Gamma\mid\Sigma[xyz]\mid\Delta}
x\stackrel{\acta_{1}}{\Longrightarrow}y$
&by the assumption
\\ \label{ps:L+':-8}
$x\stackrel{\acta_{1}\comp\acta_{2}}{\Longrightarrow}y\vdash^{\Ind}_{\Gamma\mid\Sigma[xy]\mid\Delta}
x\stackrel{\acta_{1}}{\Longrightarrow}y$
&by \ref{ps:L+':-9}, and $(\comp_{L})$
\\ \label{ps:L+':-3}
$\vdash^{\Ind}_{\Gamma\mid\Sigma\mid\Delta}
\Forall{x,y}x\stackrel{\acta_{1}\comp\acta_{2}}{\Longrightarrow}y\to x\stackrel{\acta_{1}}{\Longrightarrow}y$
&by \ref{ps:L+':-8}, $({\to}_R)$, and $(\forall_{R})$
\\ \label{ps:L+':-7}
$\vdash^{\Ind}_{\Gamma\mid\Sigma\mid\Delta}
t\stackrel{\acta_{1}}{\Longrightarrow}t_{1}$
&by the assumption
\\ \label{ps:L+':-6}
$t_{1}\stackrel{\acta_{2}^{*}}{\Longrightarrow}t_{2}\vdash^{\Ind}_{\Gamma\mid\Sigma\mid\Delta}
t_{1}\stackrel{\acta_{2}^{*}}{\Longrightarrow}t_{2}$
&by $(\Init_{\acta_{2}^{*}})$
\\ \label{ps:L+':-5}
$t_{1}\stackrel{\acta_{2}^{*}}{\Longrightarrow}t_{2}\vdash^{\Ind}_{\Gamma\mid\Sigma\mid\Delta}
t\stackrel{\acta_{1}\comp\acta_{2}^{*}}{\Longrightarrow}t_{2}$
&by \ref{ps:L+':-7}, \ref{ps:L+':-6}, and $({\comp}_{R})$
\\ \label{ps:L+':-4}
$t\stackrel{\acta_{1}}{\Longrightarrow}t_{2}\vdash^{\Ind}_{\Gamma\mid\Sigma\mid\Delta}$
&by the assumption
\\ \label{ps:L+':-2}
$\Forall{x,y}(x\stackrel{\acta_{1}\comp\acta_{2}^{*}}{\Longrightarrow}y\to x\stackrel{\acta_{1}}{\Longrightarrow}y)\,
t_{1}\stackrel{\acta_{2}^{*}}{\Longrightarrow}t_{2}\vdash^{\Ind}_{\Gamma\mid\Sigma\mid\Delta}$
&by \ref{ps:L+':-5}, \ref{ps:L+':-4}, $({\to}_L)$, and $(\forall_{L})$
\\ \label{ps:L+':-1}
$\Forall{x,y}x\stackrel{\acta_{1}\comp\acta_{2}}{\Longrightarrow}y\to x\stackrel{\acta_{1}}{\Longrightarrow}y\to
\Forall{x,y}x\stackrel{\acta_{1}\comp\acta_{2}^{*}}{\Longrightarrow}y\to x\stackrel{\acta_{1}}{\Longrightarrow}y$
{\par}
$t_{1}\stackrel{\acta_{2}^{*}}{\Longrightarrow}t_{2}\vdash^{\Ind}_{\Gamma\mid\Sigma\mid\Delta}$
&by \ref{ps:L+':-3}, \ref{ps:L+':-2}, and $({\to}_L)$
\\ \label{ps:L+':0}
$t_{1}\stackrel{\acta_{2}^{*}}{\Longrightarrow}t_{2}\vdash^{\Ind}_{\Gamma\mid\Sigma\mid\Delta}$
&by \ref{ps:L+':-1}, and $(\KEL_{\acta_{1}\comp\acta_{2}\leqq\acta_{1}\to\acta_{1}\comp\acta_{2}^{*}\leqq\acta_{1}})$
\end{proofsteps}
In reality, we use $(\Modify)$ to rearrange the order of variables, but since elimination rules can be translated, it's okay to use $(\Modify)$ as well.
\end{proofcases}
The above discussion is also possible on $\vdash^{\Ind,\INT}$.
\end{proofcases}
\end{proofcases}
\end{proof}

\begin{proof}[Proof of Theorem~\ref{thm:Completeness-by-TAk}]\
\begin{proofcases}
\item[soundness]
From the definition, $\vDash^{\TA_{k}}$ satisfies the axioms of Kleene algebra.
Since $(\Cut)$ is also satisfied, the left elimination rule for these axioms holds true.
$(\Init)$ also holds true.
As stated in Lemma~\ref{lemma:IndNK}, the left elimination rule for axioms is equivalent to the rule of induction, so soundness holds true.
\item[completeness]
Let $\Sigma=\pos{S,F,L}\in\Sig$ and
$\Gamma,\Delta\subseteq\Sen(\Sigma)$.
Suppose $\Gamma\not\vdash^{\Ind}_{\Sigma}\Delta$.
We show $\Gamma\not\vDash^{\TA_{k}}_{\Sigma}\Delta$.
{\par}
Let $\alpha:=\card(\Sen(\Sigma))$.
We define $\BS\Sigma=\cup_{n\in\alpha}\BS\Sigma_{n}\,(\{\BS\Sigma_{n}\}_{n\in\alpha})$ as follows:
\begin{proofcases}
\item[$n=0$] $\BS\Sigma_{n}:=\{\Sigma\}$
\item[$n=n'+1$] $\BS\Sigma_{n}:=\BS\Sigma_{n'}\cup\{\Sigma^{\Dex}[X]\mid X\in\Sigma_{\Dex}\}$
\item[$n$ is limit] $\BS\Sigma_{n}:=\{\cup_{n'\in n}\BS\Sigma'_{n'}\mid \BS\Sigma'\in\Pi_{n'\in n}\BS\Sigma_{n'}\}$
\end{proofcases}
Prepare a choice function of bijections $\pos{\phi_{\Sigma'}:\card(\Sen(\Sigma'))\to\Sen(\Sigma')}_{\Sigma'\in\BS\Sigma}$.
Prepare a bijective $pair:\alpha\times\alpha\to\alpha$ that satisfies $pair(m,n)\geq m$.
We define $\pos{\Gamma_{n},\Sigma_{n},\Delta_{n}}_{n\in\alpha}$ that satisfies $\Gamma_{n}\not\vdash^{\Ind}_{\Sigma_{n}}\Delta_{n}\,(\text{ for each }n\in\alpha)$ as follows:
\begin{proofcases}
\item[$n=0$] $\pos{\Gamma_{n},\Sigma_{n},\Delta_{n}}=\pos{\Gamma,\Sigma,\Delta}$
\item[$n=n'+1$] Let $pair(l,m)=n'$.
\begin{proofcases}
\item[$\phi_{\Sigma_{l}}(m)\,\Gamma_{n'}\not\vdash^{\Ind}_{\Sigma_{n'}}\Delta_{n'}$]\
\begin{proofcases}
\item[$\phi_{\Sigma_{l}}(m)=\Exists{X}\psi$]
$\pos{\Gamma_{n},\Sigma_{n},\Delta_{n}}:=\pos{\{\psi\}\cup\Sigma^{\Dex}(X)(\{\phi_{\Sigma_{l}}(m)\}\cup\Gamma_{n'}),\Sigma_{n'}^{\Dex}[X],\Sigma^{\Dex}(X)(\Delta_{n'})}$
\\By $(\exists_{L})$, $\Gamma_{n}\not\vdash^{\Ind}_{\Sigma_{n}}\Delta_{n}$.
\item[else] $\pos{\Gamma_{n},\Sigma_{n},\Delta_{n}}:=\pos{\{\phi_{\Sigma_{l}}(m)\}\cup\Gamma_{n'},\Sigma_{n'},\Delta_{n'}}$
\end{proofcases}
\item[$\phi_{\Sigma_{l}}(m)\,\Gamma_{n'}\vdash^{\Ind}_{\Sigma_{n'}}\Delta_{n'}$] By $(\Init)$ and $\Gamma_{n'}\not\vdash^{\Ind}_{\Sigma_{n'}}\Delta_{n'}$, $\Gamma_{n'}\vdash^{\Ind}_{\Sigma_{n'}}\Delta_{n'}\,\phi_{\Sigma_{l}}(m)$
\begin{proofcases}
\item[$\phi_{\Sigma_{l}}(m)=\Forall{X}\psi$]
$\pos{\Gamma_{n},\Sigma_{n},\Delta_{n}}:=\pos{\Sigma^{\Dex}(X)(\Gamma_{n'}),\Sigma_{n'}^{\Dex}[X],\Sigma^{\Dex}(X)(\Delta_{n'}\cup\{\phi_{\Sigma_{l}}(m)\})\cup\{\psi\}}$
\\By $(\forall_{R})$, $\Gamma_{n}\not\vdash^{\Ind}_{\Sigma_{n}}\Delta_{n}$.
\item[else] $\pos{\Gamma_{n},\Sigma_{n},\Delta_{n}}:=\pos{\Gamma_{n'},\Sigma_{n'},\Delta_{n'}\cup\{\phi_{\Sigma_{l}}(m)\}}$
\end{proofcases}
\end{proofcases}
\item[$n$ is limit] $\pos{\Gamma_{n},\Sigma_{n},\Delta_{n}}:=\pos{\cup_{n'\in n}\Gamma_{n'},\cup_{n'\in n}\Sigma_{n'},\cup_{n'\in n}\Delta_{n'}}$
By the compactness of $\vdash^{\Ind}$, $\Gamma_{n}\not\vdash^{\Ind}_{\Sigma_{n}}\Delta_{n}$.
\end{proofcases}
Let $\pos{\Gamma_{\alpha},\Sigma_{\alpha},\Delta_{\alpha}}:=\pos{\cup_{n\in\alpha}\Gamma_{n},\cup_{n\in\alpha}\Sigma_{n},\cup_{n\in\alpha}\Delta_{n}}$.
We can define algebraic quotient model $\A'$ from $\Gamma_{\alpha}\cap\Sen^{b}(\Sigma_{\alpha})$.
{\par}
Since $\Gamma_{\alpha}$ is a maximally consistent set, by Lemma~\ref{lemma:IndNK}, $\Gamma_{\alpha}$ contains the axioms of Kleene algebra and congruence relations between actions.
Therefore, we can define the interpretation $\acta^{\A'}$ of action $\acta$ such that $[t_{0}]\stackrel{\acta^{\A'}}{\Longrightarrow}[t_{1}]\,:\Leftrightarrow\,t_{0}\stackrel{\acta}{\Longrightarrow}t_{1}\in\Gamma_{\alpha}$.
This does not lead to a contradiction. In other words, operators other than $*$ will be interpreted in the same way as before.
For example, in the case of $\acta_{0}\comp\acta_{1}$, we can say $t_{0}\stackrel{\acta_{0}\comp\acta_{1}}{\Longrightarrow}t_{1}\leftrightarrow\Exists{x}t_{0}\stackrel{\acta_{0}}{\Longrightarrow}x\wedge x\stackrel{\acta_{1}}{\Longrightarrow}t_{1}$.
Therefore, from the above construction, we can say $t_{0}\stackrel{\acta_{0}}{\Longrightarrow}x,x\stackrel{\acta_{1}}{\Longrightarrow}t_{1}\in\Gamma_{\alpha}$.
In other words, we can say that there exists an element $a$ such that $[t_{0}]\stackrel{\acta_{0}^{\A'}}{\Longrightarrow}a, a\stackrel{\acta_{1}^{\A'}}{\Longrightarrow}[t_{1}]$, so the interpretation of $(\comp)$ is the same.
Furthermore, since the congruence of actions is preserved even when using $*$, the operation $*_{ss}^{\A'}:\A'_{ss}\to\A'_{ss}$ on $\A'^{\A'}\mid\acta\in L_{\T(\Sigma_{\alpha}),ss}\}$ is naturally determined.
Adding this, $\A'$ becomes a model of $\TA_{k}$.
{\par}
Since $\Gamma_{\alpha}$ is a maximally consistent set with witness,
$\A'\models\gamma\,(\text{for each }\gamma\in\Gamma_{\alpha})$ and $\A'\not\models\delta\,(\text{for each }\delta\in\Delta_{\alpha})$
can be confirmed by induction on the inference rules and sentence structure.
{\par}
Finally, reducing this to $\Sigma$ gives us the model $\A :=\A\red_{\iota:\Sigma\to\Sigma_{\alpha}}$ of $\pos{\Gamma,\Sigma,\Delta}$.
From the satisfaction condition, we can say that $\A\models\gamma\,(\text{for each }\gamma\in\Gamma)$ and $\A\not\models\delta\,(\text{for each }\delta\in\Delta)$.
That is, we can say that $\Gamma\not\vDash^{\TA_{k}}_{\Sigma}\Delta$.
\end{proofcases}
\end{proof}

\section{Proof of the results presented in Section~\ref{sec:interpolation}}

\begin{proof}[Proof of Lemma~\ref{lemma:1234}]\
\begin{proofcases}
\item[$1\rightarrow2$]
For satisfiable theories $T_{i}\subseteq\Sen(\Sigma_{i})\,(i\in2)$, we take models such that $\A_{i}\models t$ for all $t\in T_{i}$.
Since "$\chi_{0}^{-1}(T_{0})=\chi_{1}^{-1}(T_{1})$ and these are complete theories",
"$\A_{0}\red_{\chi_{0}}$ and $\A_{1}\red_{\chi_{1}}$ are elementary equivalent" holds true.
{\par}
From $1$, there is $\A'\in|\Mod^{\TA_{k}}(\Sigma')|$ such that "$\A'\red_{\chi'_{i}}$ and $\A_{i}$ are elementary equivalent for each $i\in2$".
From elementary equivalence, $\A'\models t$ for each $t\in\cup_{i\in2}\chi'_{i}(T_{i})$, so it is satisfyable.
\item[$1\leftarrow2$]
Let $T_{\A}$ denote the set of all sentences that satisfy model $\A$.
If "$\A_{0}\red_{\chi_{0}}$ and $\A_{1}\red_{\chi_{1}}$ are elementary equivalent",
"$\chi_{0}^{-1}(T_{\A_{0}})=\chi_{1}^{-1}(T_{\A_{1}})$ and these are complete theories".
{\par}
From $2$, $\cup_{i\in2}\chi'_{i}(T_{i})$ is satisfyable and has a model $\A'$.
Since $T_{\A_{i}}\,(i\in2)$ are complete, $\A'\red_{\chi'_{i}}$ and $\A_{i}$ are elementary equivalent for each $i\in2$.
\item[$3\leftrightarrow4$]
Let $\Gamma_{i},\Delta_{i}\subseteq\Sen(\Sigma_{i})\,(i\in2)$.
We compare the two conditions.
\begin{itemize}
\item[$(3)$]
If "For each $\phi\in\Sen(\Sigma)$,
$\Gamma_{i}\not\vDash^{\TA_{k}}_{\Sigma_{i}}\Delta_{i}\,\chi_{i}(\phi)\text{ for each }i\in2$ or
$\chi_{i}(\phi)\,\Gamma_{i}\not\vDash^{\TA_{k}}_{\Sigma_{i}}\Delta_{i}\text{ for each }i\in2$",
\\ then
$\cup_{i\in 2}\chi'_{i}(\Gamma_{i})\not\vDash^{\TA_{k}}_{\Sigma'}\cup_{i\in 2}\chi'_{i}(\Delta_{i})$.
\item[$(4)$]
If "For each $\phi\in\Sen(\Sigma)$,
$\Gamma_{0}\not\vdash^{\Ind}_{\Sigma_{0}}\Delta_{0}\,\chi_{0}(\phi)$ or
$\chi_{1}(\phi)\,\Gamma_{1}\not\vdash^{\Ind}_{\Sigma_{1}}\Delta_{1}$",
\\ then
$\cup_{i\in 2}\chi'_{i}(\Gamma_{i})\not\vdash^{\Ind}_{\Sigma'}\cup_{i\in 2}\chi'_{i}(\Delta_{i})$.
\end{itemize}
Note that from the completeness theorem, $\vDash^{\TA_{k}}$ and $\vdash^{\Ind}$ are interchangeable, and the conclusion is the same, so we confirm that the premises are equal.
First, we show that $\Gamma_{i}\not\vdash_{\Sigma_{i}}\Delta_{i}\,(i\in2)$ holds true regardless of whether we assume premises $3$ or $4$.
\begin{itemize}
\item[$(3)$]
It is obvious from $\Sen(\Sigma)\neq\emptyset$.
\item[$(4)$]
If $\Gamma_{0}\not\vdash_{\Sigma_{0}}\Delta_{0}$ does not hold true, then
$\Gamma_{0}\not\vdash_{\Sigma_{0}}\Delta_{0}\,\bot$ does not hold true, so from the assumption,
$\bot\,\Gamma_{1}\not\vdash_{\Sigma_{1}}\Delta_{1}$ should hold true instead that is a contradiction.
If $\Gamma_{1}\not\vdash_{\Sigma_{1}}\Delta_{1}$ does not hold true, then
$\top\,\Gamma_{1}\not\vdash_{\Sigma_{1}}\Delta_{1}$ does not hold true, so from the assumption,
$\Gamma_{0}\not\vdash_{\Sigma_{0}}\Delta_{0}\,\top$ should hold true instead that is a contradiction.
\end{itemize}
If we abbreviate the assumption of $3$,
($\not\vDash\chi_{0}(\phi)$ and $\not\vDash\chi_{1}(\phi)$) or ($\chi_{0}(\phi)\not\vDash$ and $\chi_{1}(\phi)\not\vDash$).
By distributive law,
\[
(\not\vDash\chi_{0}(\phi) \text{ or } \chi_{0}(\phi)\not\vDash) \text{ and }
(\not\vDash\chi_{0}(\phi) \text{ or } \chi_{1}(\phi)\not\vDash) \text{ and }
(\not\vDash\chi_{1}(\phi) \text{ or } \chi_{0}(\phi)\not\vDash) \text{ and }
(\not\vDash\chi_{1}(\phi) \text{ or } \chi_{1}(\phi)\not\vDash)
\]
By $\Gamma_{i}\not\vdash_{\Sigma_{i}}\Delta_{i}\,(i\in2)$, the two outermost statements are tautologies, so
\[
(\not\vDash\chi_{0}(\phi) \text{ or } \chi_{1}(\phi)\not\vDash) \text{ and }
(\not\vDash\chi_{1}(\phi) \text{ or } \chi_{0}(\phi)\not\vDash)
\]
Using $\neg$, "$\not\vDash\chi_{0}(\phi)$ or $\chi_{1}(\phi)\not\vDash$ for all $\phi\in\Sen(\Sigma)$" is sufficient.
In other words, it is the same as the premise in $4$.
\item[$2\rightarrow3$ ($\Sigma$ is countable)]
Suppose that
$\Gamma_{i}\not\vDash^{\TA_{k}}_{\Sigma_{i}}\Delta_{i}\,\chi_{i}(\phi)\text{ for each }i\in2$ or
$\chi_{i}(\phi)\,\Gamma_{i}\not\vDash^{\TA_{k}}_{\Sigma_{i}}\Delta_{i}\text{ for each }i\in2$ holds true
for each $\phi\in\Sen(\Sigma)$.
Then for each $n\in\omega$ and $\phi\in\Sen(\Sigma)^{n}$, there is $v\in2^{n}$ such that
\begin{align*}
\Gamma_{v,i}\,\Gamma_{i}\not\vDash^{\TA_{k}}_{\Sigma_{i}}\Delta_{i}\,\Delta_{v,i}\text{ for each }i\in2
\end{align*}
where $\Gamma_{v,i}:=\chi_{i}\circ\phi\circ v^{-1}(0)$, and $\Delta_{v,i}:=\chi_{i}\circ\phi\circ v^{-1}(1)$.
In other words, if we assign truth or falsity appropriately, it can be satisfied in both cases.
Assume it is not true and show a contradiction.
\begin{itemize}
\item[]
In this case, for any assignment $v$, there exists a function $i\in2$ such that $\Gamma_{v,i}\,\Gamma_{i}\vDash^{\TA_{k}}_{\Sigma_{i}}\Delta_{i}\,\Delta_{v,i}$.
In other words,
$\chi_{0}(\wedge_{j\in n}\neg^{v(j)}\phi(j))\,\Gamma_{0}\vDash^{\TA_{k}}_{\Sigma_{0}}\Delta_{0}$ or
$\chi_{1}(\wedge_{j\in n}\neg^{v(j)}\phi(j))\,\Gamma_{1}\vDash^{\TA_{k}}_{\Sigma_{1}}\Delta_{1}$ holds true.
To put it another way, this means that
$\chi_{0}(\vee_{v\in A}\wedge_{j\in n}\neg^{v(j)}\phi(j))\,\Gamma_{0}\vDash^{\TA_{k}}_{\Sigma_{0}}\Delta_{0}$ and
$\chi_{1}(\vee_{v\in A^{c}}\wedge_{j\in n}\neg^{v(j)}\phi(j))\,\Gamma_{1}\vDash^{\TA_{k}}_{\Sigma_{1}}\Delta_{1}$
hold true if we choose $A \subseteq 2^{n}$ appropriately.
Note that
$\vDash^{\TA_{k}}_{\Sigma_{0}}\chi_{0}(\vee_{v\in 2^{n}}\wedge_{j\in n}\neg^{v(j)}\phi(j))$
we can say that
$\Gamma_{0}\vDash^{\TA_{k}}_{\Sigma_{0}}\Delta_{0}\,\chi_{0}(\vee_{v\in A^{c}}\wedge_{j\in n}\neg^{v(j)}\phi(j))$.
This contradicts the "prerequisite" for $4$ that we wrote earlier in the verification of $(3\leftrightarrow4)$.
Since the "prerequisites" for $3$ and $4$ were equivalent, this contradicts the assumptions of this discussion.
\end{itemize}
Therefore, in the finite case, satisfiability can be maintained by carefully distributing truth values.
In the countable case, it holds true by Konig's Lemma.
{\par}
In other words, if the signature $\Sigma$ is countable, then $\Gamma_{i},\Delta_{i}\,(i\in2)$ can be extended to a theory $\hat{\Gamma}_{i},\hat{\Delta}_{i}\,(i\in2)$ that determines the truth value of all elements of $\Sen(\Sigma)$ while preserving satisfiability.
By applying $2$ to $T_{i}:=\hat{\Gamma}_{i}\cup\{\neg\delta\mid\delta\in\hat{\Delta}_{i}\}\,(i\in2)$,
we get $\cup_{i\in 2}\chi'_{i}(\Gamma_{i})\not\vDash^{\TA_{k}}_{\Sigma'}\cup_{i\in 2}\chi'_{i}(\Delta_{i})$.
\item[$2\leftarrow3$]
Let $T_{i}\subseteq\Sen(\Sigma_{i})\,(i\in2)$ are satisfiable theories such that
$\chi_{0}^{-1}(T_{0})=\chi_{1}^{-1}(T_{1})$ and these are complete theories.
Since $T_{i}\,(i\in2)$ knows the truth value of all sentences on $\Sigma$, if we set $\Gamma_{i}:=T_{i}$ and $\Delta_{i}=\emptyset$, these satisfy condition (3), and we get $\cup_{i\in 2}\chi'_{i}(\Gamma_{i})\not\vDash^{\TA_{k}}_{\Sigma'}\cup_{i\in 2}\chi'_{i}(\Delta_{i})$.
In other words, there exists a model in which all the conditions of $\cup_{i\in 2}\chi'_{i}(\Gamma_{i})$ are true and all the conditions of $\cup_{i\in 2}\chi'_{i}(\Delta_{i})$ are false.
Therefore, $\cup_{i\in2}\chi'_{i}(T_{i})$ is satisfyable.
\end{proofcases}
\end{proof}

\begin{proof}[Proof of Lemma~\ref{fact:exactness}]
By the commutativity of the square, if the above statement holds true, the below statement also holds.
Assume the above holds true.
For simplicity, we treat signatures as just a set of symbols.
We define the relation $R$ on $\Sigma_{0}\sqcup\Sigma_{1}:=(\Sigma_{0}\times\{0\})\cup(\Sigma_{1}\times\{1\})$ as follows
\[
\pos{a,i}R\pos{b,j}
:\iff
\chi_{i}(c)=a\text{ and }\chi_{j}(c)=b\text{ for some }c\in\Sigma\text{ and }i,j\in2
\]
Since $R$ is symmetric, the transitive closure $R^{*}$ is an equivalence relation.
By dividing the direct sum using this equivalence relation, $(\Sigma_{0}\sqcup\Sigma_{1})/R^{*}$ is isomorphic to the pushout $\Sigma'$.
We denote the set of all things equivalent to the symbol $s$ as $[s]$.
We define $\A'$ as follows
\begin{itemize}
\item $\A'_{[\pos{s,i}]}:={\A_{i}}_{s}\Space[\pos{\sigma,i}]^{\A'}:=\sigma^{\A_{i}}\Space[\pos{\lambda,i}]^{\A'}:=\lambda^{\A_{i}}$
\item $\A'_{\pos{[\pos{s,i}],[\pos{s,i}]}}:=\{\pi\mid \pi\in{\A_{i}}_{ss}\}$
\item $*^{\A'}_{\pos{[\pos{s,i}],[\pos{s,i}]}}(\pi):=*^{\A_{i}}_{s}(\pi)$ $(\pi\in\A'_{ss})$
\end{itemize}
This is well-defined.
For example let $\pos{a,i}R\pos{b,j}$ holds true for sorts $a,b$.
By the definition, $\chi_{i}(c)=a\text{ and }\chi_{j}(c)=b\text{ for some }c\in\Sigma$.
Therefore from $\A_{i}\red_{\chi_{i}}=\A_{j}\red_{\chi_{j}}$, ${\A_{i}}_{a}={\A_{i}}_{\chi_{i}(c)}={\A_{i}\red_{\chi_{i}}}_{c}={\A_{j}\red_{\chi_{j}}}_{c}={\A_{j}}_{\chi_{j}(c)}={\A_{j}}_{b}$ hold true.
Since $R^{*}$ is a repetition, it does not matter which equivalent element we choose as the representative.
\end{proof}

\begin{proof}[Proof of Theorem~\ref{theorem:CISQ}]\
\begin{proofcases}
\item[$\mathrm{RJC}$] We first show $\mathrm{RJC}$.
We fix the signature morphisms $\chi_{i}:\Sigma\to\Sigma_{i} \ (i\in2)$,
colimit $\{\chi'_{i}:\Sigma_{i}\to\Sigma'\}_{i\in\{0,1\}}$,
satisfiable theories $T_{i}\subseteq\Sen(\Sigma_i) \ (i\in2)$
which satisfy the completeness and equality of $\chi_{i}^{-1}(T_{i})\,(i\in2)$.
\begin{proofcases}
\item[Preparation of symbols]
Let $\alpha:=\max\{\mathrm{card}(\Sen(i))\mid i=\Sigma,\Sigma_0,\Sigma_1\}^+$.
Let $\mathcal{C}=\{\mathcal{C}_s\}_{s\in S\sqcup S^{=}}$ be a $S\sqcup S^{=}$-sorted disjoint set of new constant (and relation) symbols for $\Sigma$, $\Sigma_0$, and $\Sigma_1$.
Let $\mathcal{C}_{i}=\{\mathcal{C}_{i,s}\}_{s\in S_{i}\sqcup S^{=}_{i}} \ (i\in2)$ be $S_{i}\sqcup S^{=}_{i}$-sorted sets of new (It does not overlap with $\mathcal{C}$ also.) constant (and relation) symbols that satisfy $\mathrm{card}(\mathcal{C}_{i,s})=\alpha$ for all $i\in2$ and $s\in S_{i}$.
We will denote the addition of a constant like $\Sigma(C):\Sigma\to\Sigma[C]\,(C\subseteq\mathcal{C})$.
Let $\chi_{i}[C]:\Sigma[C]\to\Sigma_{i}[C]\,(C\subseteq\mathcal{C})$ of $\chi_{i}$ by constants satisfies $\chi_{i}[C](c)=c$ for all $c\in C$.
For simplicity, we may not distinguish between the extension and the original signature morphisms.
\item[Henkin method]
The goal is create a model such that $\A_{0}'\upharpoonright_{\chi_{0}}=\A_{1}'\upharpoonright_{\chi_{1}}$ and $\A'_{i}\models^{\TA_{k}}\gamma\,(\gamma\in T_{i})$ hold true, and then apply Lemma~\ref{fact:exactness} to it.
We extend the theories alternately to achieve elementary equivalence, adding constants so that their intersection coincide.
Let $U_{i}:=\{\chi_{i}[C]:\Sigma[C]\to\Sigma_{i}[C][C_{i}]\mid C\subseteq_{\alpha}\mathcal{C},\,C_{i}\subseteq_{\alpha}\mathcal{C}_{i}\}$.
We prepare surjective functions $\pos{\phi_{\Sigma'',i}:\alpha\to\Sen(\Sigma'')}_{\Sigma''\in \Cod U_{i}}$ as choice function.
And prepare a bijective $pair:\alpha\times\alpha\to\alpha$ that satisfies $pair(m,n)\geq m$.
\\ \ \\
We show that we can define
an expanding sequence of signatures and satisfyable theories
$\pos{\BS{T}_{i}(n),\BS\chi_{i}(n):\BS\Sigma(n)\to\BS\Sigma_{i}(n)\in U_{i}}_{i\in2,n\in\alpha}$
that satisfies the following conditions recursively for $n \in \alpha$.
\begin{description}
\item[$(T')$] for each $n'<n$, $t\in F_{\T_{\BS\Sigma_{i}(n')}}$, there is $c\in\BS\Sigma(n)\cap\{\mathcal{C}_{s}\}_{s\in S}$ such that $(\chi_{i}(n)(c)=t)\in\BS{T}_{i}(n)$
\item[$(A')$] for each $n'<n$, $\acta\in L_{\T(\BS\Sigma_{i}(n'))}$, there is $\pi\in\BS\Sigma(n)\cap\{\mathcal{C}_{\BS{s}}\}_{\BS{s}\in S^{=}}$ such that $"\chi_{i}(n)(\pi)\equiv\acta"\in\BS{T}_{i}(n)$
\item[$(R)$] $\BS{T}_{1}(n)\vDash^{\TA_{k}}_{\BS\Sigma_{1}(n)}\BS\chi_{1}(n)(\varphi) \text{ if }\BS{T}_{0}(n)\vDash^{\TA_{k}}_{\BS\Sigma_{0}(n)}\BS\chi_{0}(n)(\varphi) \text{ for all } \varphi\in\Sen(\BS\Sigma(n))$
\item[$(L')$] $\BS{T}_{0}(n)\vDash^{\TA_{k}}_{\BS\Sigma_{0}(n)}\BS\chi_{0}(n)(\varphi) \text{ if }\BS{T}_{1}(n')\vDash^{\TA_{k}}_{\BS\Sigma_{1}(n')}\BS\chi_{1}(n')(\varphi) \text{ for all } n'<n \text{ and } \varphi\in\Sen(\BS\Sigma(n'))$
\end{description}
\begin{proofcases}
\item[$n=0$] $\pos{\BS{T}_{i}(n),\BS\chi_{i}(n):\BS\Sigma(n)\to\BS\Sigma_{i}(n)}:=\pos{T_{i},\chi_{i}:\Sigma\to\Sigma_{i}}$.
\item[$n=n'+1$] Let $pair(l,m)=n'$.
\begin{proofcases}
\item[Step $odd$]
To reduce the number of cases to consider, we use the following symbols.
\[
\pos{\phi,\psi,X}:=
\begin{cases}
\pos{\phi_{\BS\Sigma_{1}(l),1}(m),\phi',X'}&(\{\phi_{\BS\Sigma_{1}(l),1}(m)\}\cup\BS{T}_{1}(n')\not\vDash^{\TA_{k}}_{\BS\Sigma_{1}(n')}\emptyset\,\,\&\,\,\phi_{\BS\Sigma_{1}(l),1}(m)=\Exists{X'}\phi')\\
\pos{\phi_{\BS\Sigma_{1}(l),1}(m),\phi_{\BS\Sigma_{1}(l),1}(m),\emptyset}&(\{\phi_{\BS\Sigma_{1}(l),1}(m)\}\cup\BS{T}_{1}(n')\not\vDash^{\TA_{k}}_{\BS\Sigma_{1}(n')}\emptyset\,\,\&\,\,\phi_{\BS\Sigma_{1}(l),1}(m)=else)\\
\pos{\neg\phi_{\BS\Sigma_{1}(l),1}(m),\neg\phi_{\BS\Sigma_{1}(l),1}(m),\emptyset}&(\{\phi_{\BS\Sigma_{1}(l),1}(m)\}\cup\BS{T}_{1}(n')\vDash^{\TA_{k}}_{\BS\Sigma_{1}(n')}\emptyset)\\
\end{cases}
\]
First, we take a bijection $\theta:X\to C_{1}$ that associates the unused symbols $C_{1}\subseteq\mathcal{C}_{1}$ with $X$.
$\{\theta(\psi)\}\cup\BS\Sigma_{1}(l)(C_{1})(\{\phi\}\cup\BS{T}_{1}(n'))$ is satisfyable.
Next, we take a bijection $f:C\to F_{T(\BS\Sigma_{1}(l)[C_{1}])}\sqcup L_{\T(\BS\Sigma_{1}(l)[C_{1}])}$
 that associates the unused symbols $C\subseteq\mathcal{C}$ with $F_{T_{\BS\Sigma_{1}(l)[C_{1}]}}$ and $L_{\T(\BS\Sigma_{1}(l)[C_{1}])}$.
(Note that $\card(C\cup C_{1})<\alpha$ holds true.)
\begin{align*}
&\BS\chi_{1}(n)^{-}:\BS\Sigma(n)^{-}\to\BS\Sigma_{1}(n)^{-}:=\BS\chi_{1}(n')[C]:\BS\Sigma(n')[C]\to\BS\Sigma_{1}(n')[C_{1}][C]\,(\in U_{1})&\\
&\BS{T}_{1}(n)^{-}:=\{\psi\}\cup\BS\Sigma_{1}(l)(C)(\{\phi\}\cup\BS{T}_{1}(n'))\cup&\\
&\Space\{t=f(t)\mid t\in \{C_{s}\}_{s\in S}\}\cup\{\pi\equiv f(\pi)\mid\pi\in\{C_{\BS{s}}\}_{\BS{s}\in S^{=}}\}\\%
&\BS\chi_{0}(n)^{-}:\BS\Sigma(n)^{-}\to\BS\Sigma_{0}(n)^{-}:=\BS\chi_{0}(n')[C]:\BS\Sigma(n')[C]\to\BS\Sigma_{0}(n')[C]\,(\in U_{0})&\\
&\BS{T}_{0}(n)^{-}:=\BS\Sigma_{0}(n')(C)(\BS{T}_{0}(n'))\cup\Phi(n)^{-}&\\
&\Space\Phi(n)^{-}:=\{\BS\chi_{0}(n)^{-}(\varphi)\mid\BS{T}_{1}(n)^{-}\vDash^{\TA_{k}}_{\BS\Sigma(n)^{-}}\BS\chi_{1}(n)^{-}(\varphi), \varphi\in\Sen(\BS\Sigma(n)^{-})\}&
\end{align*}
$\BS{T}_{1}(n)^{-}$ is satisfyable because it only provides definitions to the symbol.
\\ 
$\BS{T}_{0}(n)^{-}$ is also satisfyable.
\begin{itemize}
\item[]
If not, then due to its compactness
we can take finite part
$\iota:C_{f}\subseteq_{\omega}C$ and $\Phi_{f}\subseteq_{\omega}\iota^{-1}(\Phi(n)^{-})$ such that
$\BS\Sigma_{0}(n')(C_{f})(\BS{T}_{0}(n'))\cup\Phi_{f}\vDash^{\TA_{k}}_{\BS\Sigma_{0}(n')[C_{f}]}\emptyset$.
From proof rules $\BS{T}_{0}(n')\vDash^{\TA_{k}}_{\BS\Sigma_{0}(n')}\neg\Exists{C_{f}}\wedge\Phi_{f}$ holds true.
(In reality, $C_{f}$ is replaced with a variable block $X_{f}$.)
From the recursion hypothesis we can say $\BS{T}_{1}(n')\vDash^{\TA_{k}}_{\BS\Sigma_{1}(n')}\neg\Exists{C_{f}}\wedge\Phi_{f}$ from $(R)$.
But since $\Phi_{f}\subseteq_{\omega}\iota^{-1}(\Phi(n)^{-})$,
$\BS{T}_{1}(n)^{-}\vDash^{\TA_{k}}_{\BS\Sigma_{1}(n)^{-}}\Exists{C_{f}}\wedge\Phi_{f}$ holds true,
and therefore,
that contradicts the fact that $\BS{T}_{1}(n)^{-}$ is satisfyable.
\item[]
\end{itemize}
Note that the definitions here satisfy the following conditions.
\begin{description}
\item[$(T_{1})$] for each $t\in F_{\T(\BS\Sigma_{1}(n'))}$, there is $c\in\BS\Sigma(n)^{-}\cap\{\mathcal{C}_{s}\}_{s\in S}$ such that $(\chi_{1}(n)^{-}(c)=t)\in\BS{T}_{1}(n)^{-}$
\item[$(A_{1})$] for each $\acta\in L_{\T(\BS\Sigma_{1}(n'))}$, there is $\pi\in\BS\Sigma(n)^{-}\cap\{\mathcal{C}_{\BS{s}}\}_{\BS{s}\in S^{=}}$ such that $"\chi_{1}(n)^{-}(\pi)\equiv\acta"\in\BS{T}_{1}(n)^{-}$
\item[$(L)$] $\BS{T}_{0}(n)^{-}\vDash^{\TA_{k}}_{\BS\Sigma_{0}(n)^{-}}\BS\chi_{0}(n)^{-}(\varphi) \text{ if }\BS{T}_{1}(n)^{-}\vDash^{\TA_{k}}_{\BS\Sigma_{1}(n)^{-}}\BS\chi_{1}(n)^{-}(\varphi) \text{ for all } \varphi\in\Sen(\BS\Sigma(n)^{-})$
\item[]
\end{description}
\item[Step $even$]
Perform (Step $odd$) in reverse. Just as we were able to derive $(L)$ from $(R)$ in the definition above, we can derive $(R)$ from $(L)$.
The method is clearly the same, so we will omit the details.
\end{proofcases}
\item[$n$ is a limit]
\begin{align*}
&\BS\chi_{i}(n):\BS\Sigma(n)\to\BS\Sigma_{i}(n):=\cup_{n'<n}\BS\chi_{i}(n'):\cup_{n'<n}\BS\Sigma(n')\to\cup_{n'<n}\BS\Sigma_{i}(n')\,(\in U_{i})&\\
&\BS{T}_{i}(n):=\cup_{n'<n}\BS\chi_{i}(n'\leq n)(\BS{T}_{i}(n'))&
\end{align*}
Note that since $\alpha$ is a regular cardinal, even in the limit, we never run out of elements in $\mathcal{C},\mathcal{C}_{i}\,(i\in2)$, and the cardinality of $\Sen$ is also less than $\alpha$.
Furthermore, due to compactness, the limit is still satisfyable.
\end{proofcases}
The resulting $\pos{\BS{T}_{i}(n),\BS\chi_{i}(n):\BS\Sigma(n)\to\BS\Sigma_{i}(n)}$ satisfies four conditions.
{\par}
$\pos{\BS{T}_{i}(\alpha),\BS\chi_{i}(\alpha):\BS\Sigma(\alpha)\to\BS\Sigma_{i}(\alpha)}:=\pos{\cup_{n<\alpha}\BS\chi(n\leq\alpha)(\BS{T}_{i}),\cup_{n<\alpha}\BS{\chi}_{i}(n):\cup_{n<\alpha}\BS\Sigma(n)\to\cup_{n<\alpha}\BS\Sigma_{i}(n)}$
has a maximal consistent set $\BS{T}_{i}(\alpha)$ with witnesses and the complete theory $\BS{T}(\alpha):=\BS{\chi}_{i}(\alpha)^{-1}(\BS{T}_{i}(\alpha))$.
$\A_{i}\in|\Mod(\BS\Sigma_{i}(\alpha))|$ can be constructed by quotient such that
$\A_{0}\red_{\BS\chi_{0}(\alpha)}$ and $\A_{1}\red_{\BS\chi_{1}(\alpha)}$ are elementary equivalent and
$\A_{i}\models^{\TA_{k}}\phi$ hold true for all $\phi\in\BS{T}_{i}(\alpha)$.
From the first two of the four conditions, we can see that $\BS\Sigma(\alpha)$ covers the actions and terms that appear in $\BS\Sigma_{0}(\alpha)$ and $\BS\Sigma_{1}(\alpha)$.
(That is, all elements of $\{{\A_{i}}_{s}\}_{s\in S}$ can be expressed as $\BS\Sigma(\alpha)$ terms, and the elements of $\{{\A_{i}}_{\BS{s}}\red_{\chi_{i}}\}_{\BS{s}\in S^{=}}$ also be in the image of $L_{\T(\BS\Sigma(\alpha))}$.)
Therefore, these elements have a one-to-one correspondence.
Then, if we can appropriately replace the names so that $\A_{0}'\upharpoonright_{\chi_{0}}=\A_{1}'\upharpoonright_{\chi_{1}}$ holds true, we can apply Lemma~\ref{fact:exactness}.
{\par}
Let $M_i:=\{s\in S\mid \text{ there is a sort } s'\neq s \text{ which satisfies } \chi_i(s)=\chi_i(s') \}\ (i\in2)$, and $M:=S\setminus(M_0\cup M_1)$.
Since $\chi_0$ and $\chi_1$ are disjoint (about sort confluence), $S=M\sqcup M_0\sqcup M_1$.
We make new models $\A_{0}'$ and $\A_{1}'$ by replacing the elements (and replacing functions, relations, and $*$ to match the replacement of elements):
\begin{itemize}
\item $c^{\A_{i}'}=c^{\A_{0}}$ if $c\in\mathcal{C}_s\subseteq\mathcal{C}_{\chi_0(s)}$ and $s\in M\sqcup M_0$
\item $c^{\A_{i}'}=c^{\A_{1}}$ if $c\in\mathcal{C}_s\subseteq\mathcal{C}_{\chi_1(s)}$ and $s\in M_1$
\end{itemize}
Note that nothing is replaced above for the merging sorts.
For example, if $i=0$ we only change elements of $s\in M_1$.
The elements of $M_1$ does not merged by $\chi_0$.
For parts that do not merge, the sort correspondence is one-to-one.
By the definition, $\A_{0}'\upharpoonright_{\chi_{0}}=\A_{1}'\upharpoonright_{\chi_{1}}$ holds true.
By the Lemma~\ref{fact:exactness}, there is a $\Sigma'$-model $\A'$ such that $\A'\upharpoonright_{\chi'_0}=\A'_{1}$ and $\A'\upharpoonright_{\chi'_1}=\A'_{0}$ hold true.
By satisfaction condition, $\A'\models \phi$ holds true or each $\phi\in\chi'_1(T_0)\cup\chi'_0(T_1)$.
Therefore, $(\Sigma',\chi'_1(T_0)\cup\chi'_0(T_1))$ is consistent.
\end{proofcases}
\item[$\mathrm{CI}$]
Since the equivalence of $\mathrm{CI}$ and $\mathrm{RJC}$ has only been proven for countable squares, in order to prove $\mathrm{CI}$ using $\mathrm{RJC}$, we need to restrict the squares to countable.
{\par}
Suppose $\Gamma_{i},\Delta_{i}\subseteq\Sen(\Sigma_{i})\,(i\in2)$ satisfy
$\Gamma'\vdash^{\Ind}_{\Sigma'}\Delta'$ $(\Gamma':=\chi_{0}(\Gamma_{0})\cup\chi_{1}(\Gamma_{1}),\,\Delta':=\chi_{0}(\Delta_{0})\cup\chi_{1}(\Delta_{1}))$.
{\par}
By the compactness,
there exists a finite signature $\Sigma'_{f}\subseteq\Sigma'$,
inclusion $\iota':\Sigma'_{f}\to\Sigma'$ and
sets of sentences $\Gamma'_{f}\subseteq\iota'^{-1}(\Gamma')$, $\Delta'_{f}\subseteq\iota'^{-1}(\Delta')$ such that
$\Gamma'_{f}\vdash^{\Ind}_{\Sigma'_{f}}\Delta'_{f}$ holds true.
{\par}
We choose the following
\[
f'_{M,i}\in\prod_{\sigma\in\Sigma'_{f}\cap\chi'_{i}(\Sigma_{i})}{\chi'_{i}}^{-1}(\sigma)\Space
f'_{L,i}\in\prod_{\gamma\in\Gamma'_{f}\cap\chi'_{i}(\Gamma_{i})}{\chi'_{i}}^{-1}(\gamma)\Space
f'_{R,i}\in\prod_{\delta\in\Delta'_{f}\cap\chi'_{i}(\Delta_{i})}{\chi'_{i}}^{-1}(\delta)
\]
Let $\Sigma_{f,i}$ be the smallest signature containing the symbols appearing in $f'_{L,i}(\Gamma'_{f}),\,\,f'_{M,i}(\Sigma'_{f}),\,\,f'_{R,i}(\Delta'_{f})$, and let $\iota_{i}:\Sigma_{f,i}\to\Sigma_{i}$ be the inclusion.
We can take finite sets $\Gamma_{f,i},\,\Delta_{f,i}\subseteq\Sen(\Sigma_{f,i})$ such that
$\cup_{i\in2}(\chi'_{i}\circ\iota_{i})(\Gamma_{f,i})=\iota'(\Gamma'_{f})$,
$\cup_{i\in2}(\chi'_{i}\circ\iota_{i})(\Delta_{f,i})=\iota'(\Delta'_{f})$.
Let $\chi'_{f,i}:\Sigma_{f,i}\to\Sigma'_{f}$ be the restriction of $\chi'_{i}$.
By definition, $\cup_{i\in2}\chi'_{f,i}(\Gamma_{f,i})=\Gamma'_{f}$ and $\cup_{i\in2}\chi'_{f,i}(\Delta_{f,i})=\Delta'_{f}$ hold true.
Let $\Sigma_{f}$ be an empty signature, and
$\iota:\Sigma_{f}\to\Sigma$ be the inclusion.
Let $\chi_{f,i}:\Sigma_{f}\to\Sigma_{f,i}$ be the restriction of $\chi_{i}$.
{\par}
Up to here, we have obtained the following subsquare and finite sets $\Gamma_{f,i},\,\Delta_{f,i}\subseteq\Sen(\Sigma_{f,i})$ that satisfy
$\cup_{i\in2}\chi_{f,i}(\Gamma_{f,i})\vdash^{\Ind}_{\Sigma'_{f}}\cup_{i\in2}\chi_{f,i}(\Delta_{f,i})$.
\[
\def\labelstyle{\normalsize}
\xymatrix@R=5pt@C=35pt{
&\Sigma_{0}\ar[rrr]|{\phantom{\int}{\chi'_{0}}\phantom{\int}}&&&\Sigma'\\
\Sigma\ar[ur]|{\phantom{\int}{\chi_{0}}\phantom{\int}}\ar[rrr]|{\phantom{\int}{\chi_{1}}\phantom{\int}}&&&\Sigma_{1}\ar[ur]|{\phantom{\int}{\chi'_{1}}\phantom{\int}}&\\
&&&&\\
&\ar[uuu]|{\phantom{\int}{\iota_{0}}\phantom{\int}}\Sigma_{f,0}\ar[rrr]|{\phantom{\int}{\chi'_{f,0}}\phantom{\int}}&&&\ar[uuu]|{\phantom{\int}{\iota'}\phantom{\int}}\Sigma'_{f}\\
\ar[uuu]|{\phantom{\int}{\iota}\phantom{\int}}\Sigma_{f}\ar[ur]|{\phantom{\int}{\chi_{f,0}}\phantom{\int}}\ar[rrr]|{\phantom{\int}{\chi_{f,1}}\phantom{\int}}&&&\ar[uuu]|{\phantom{\int}{\iota_{1}}\phantom{\int}}\Sigma_{f,1}\ar[ur]|{\phantom{\int}{\chi'_{f,1}}\phantom{\int}}&
}\]
However, this does not necessarily mean it will be a pushout, so we will add symbols to make it so.
\\ \ {\par}
For simplicity, we treat signatures as just a set of symbols.
We define the relation $R$ on $\Sigma_{0}\sqcup\Sigma_{1}:=(\Sigma_{0}\times\{0\})\cup(\Sigma_{1}\times\{1\})$ as follows
\[
\pos{a,i}R\pos{b,j}
:\iff
\chi_{i}(c)=a\text{ and }\chi_{j}(c)=b\text{ for some }c\in\Sigma\text{ and }i,j\in2
\]
Since $R$ is symmetric, the transitive closure $R^{*}$ is an equivalence relation.
By dividing the direct sum using this equivalence relation, we get $(\Sigma_{0}\sqcup\Sigma_{1})/R^{*}$ becomes (isomorphic to) $\Sigma'$.
{\par}
We choose functions
\begin{itemize}
\item $\pos{f_{\Sigma''}:\omega\twoheadrightarrow\{\pos{a,b}\mid aR^{*}b,\,a,b\in\Sigma''\}}_{\Sigma''\subseteq\Sigma_{0}\sqcup\Sigma_{1}}$
\item $e:\{\pos{d,d'}\mid dRd',\,d,d'\in\Sigma_{0}\sqcup\Sigma_{1}\}\to\Sigma$ such that $\chi_{i}(c)=d,\,\chi_{j}(c)=d'\,(e(d,d')=c)$.
\item $e^{*}:\{\pos{a,b}\mid aR^{*}b,\,a,b\in\Sigma_{0}\sqcup\Sigma_{1}\}\to(\Sigma_{0}\sqcup\Sigma_{1})^{*}$ such that $a=d_{0}Rd_{1}\cdots d_{n-2}Rd_{n-1}=b\,(e^{*}(a,b)=\pos{d_{i}}_{i\in n})$.
\end{itemize}
Let $\text{pair}:\omega\times\omega\to\omega$ be a bijective function defined by $\text{pair}(i,j):=\bigl((i+j)(i+j+1)+2j\bigr)/2$ for all $i,j\in\omega$.
We define $\pos{\BS\chi_{f,i}(n):\BS\Sigma_{f}(n)\to\BS\Sigma_{f,i}(n)}_{n\in\omega}$ as follows
\begin{proofcases}
\item[$n=0$]
$\BS\chi_{f,i}(n):\BS\Sigma_{f}(n)\to\BS\Sigma_{f,i}(n):=\chi_{f,i}:\Sigma_{f}\to\Sigma_{f,i}$
\item[$n=n'+1$] Let $n'=pair(l,m)$.
For $a,b$ such that $f_{\BS\Sigma_{f,0}(l)\sqcup\BS\Sigma_{f,1}(l)}(m)=\pos{a,b}$,
we add the witness $e^{*}(\pos{a,b})=\{d_{i}\}_{i\in N}$ of $aR^{*}b$ to $\BS\Sigma_{f,0}(n')\sqcup\BS\Sigma_{f,1}(n')$.
If $N>0$, we further add the witness $e(d_{i},d_{i+1})=c_{i}\in\Sigma\,(i<N-1)$ of $d_{i}Rd_{i+1}$ to $\BS\Sigma_{f}(n')$.
We also add $\pos{\chi_{i}(c_{i}),i}\,(i\in2)$ to $\BS\Sigma_{f,0}(n')\sqcup\BS\Sigma_{f,1}(n')$.
For function symbols and labels, we also add the associated sorts.
\end{proofcases}
Let $\chi_{c,i}:\Sigma_{c}\to\Sigma_{c,i}\,(i\in2)$ be the morphisms obtained as the limit of the expansion sequence.
These satisfy the disjoint property.
The morphism $\chi'_{c,i}:\Sigma_{c,i}\to\Sigma'_{c}$ obtained by restricting $\chi'_{i}$ to $\Sigma_{c,i}$ is a pushout, and each signature is countable.
Therefore, from the result of \(\mathrm{RJC}\), this becomes \(\mathrm{CI}\).

Let $\iota^{c}_{f,i}:\Sigma_{f,i}\to\Sigma_{c,i}$ and $\iota_{c,i}:\Sigma_{c,i}\to\Sigma_{i}$ be inclusions.
Let $\Gamma_{c,i}:=\iota^{c}_{f,i}(\Gamma_{f,i})$, $\Delta_{c,i}:=\iota^{c}_{f,i}(\Delta_{f,i})$.
These satisfy $\cup_{i\in2}\chi_{c,i}(\Gamma_{c,i})\vdash^{\Ind}_{\Sigma'_{c}}\cup_{i\in2}\chi_{c,i}(\Delta_{c,i})$.
Therefore there is a sentence $\phi\in\Sen(\Sigma_{c})$ that satisfies
$\Gamma_{c,0}\vdash^{\Ind}_{\Sigma_{c,0}}\Delta_{c,0}\,\chi_{c,0}(\phi)$ and
$\chi_{c,1}(\phi)\,\Gamma_{c,1}\vdash^{\Ind}_{\Sigma_{c,1}}\Delta_{c,1}$.
Let $\iota_{c}:\Sigma_{c}\to\Sigma$ be the inclusion.
Then $\iota_{c}(\phi)$ becomes an interpolation of the original $\Gamma_{i},\Delta_{i}\subseteq\Sen(\Sigma_{i})\,(i\in2)$.
\end{proofcases}
Conversely if there is overlap in the sort merging, it will always include a sort merging like Example~\ref{example:2sort} or Example~\ref{example:3sort}.
In other words, if we write these square condition for push outs using only the sort condition, the "disjoint" is necessary.
\end{proof}

\begin{example}[when same sorts merge]\label{example:2sort}\small
Let $\Sigma=(\{s',s''\},\{d':\to s',d'':\to s''\},\emptyset)$, $\Sigma_0=\Sigma_1=\Sigma'=(\{s\},\{d',d'':\to s\},\emptyset)$.
$\Sigma'$ is a push out of same signature morphisms $\chi=\chi_{0}=\chi_{1}:\Sigma\to\Sigma_0,\Sigma_1$ $(\chi(s')=\chi(s'')=s, \chi(d')=d', \chi(d'')=d'')$.
Consider the following:
\begin{itemize}%
\item $T=\{\text{" } \texttt{s} \text{ has exactly two elements"}, \text{" } \text{ "}\A_{\texttt{ss}}=\{0,1,0^{c},1^{c}\}\text{"}\mid\texttt{s}\in\{s',s''\}\}$ \par
In $\FOL$ this is a complete theory.
By the properties of relation algebras and Kleene algebras $0^{*}=1^{*}=1,(1^{c})^{*}=(0^{c})^{*}=0^{c}$.
These four $\{0,1,0^{c},1^{c}\}$ are closed under other operations as well.
So this is a complete theory in $\TA_{k}$.
\item $T_0=\{\text{" } s \text{ has exactly two elements"}, d'=d'', \text{ "}\A_{ss}=\{0,1,0^{c},1^{c}\}\text{"}\}$. 
(This is satisfiable.)
\item $T_1=\{\text{" } s \text{ has exactly two elements"}, d'\neq d'', \text{ "}\A_{ss}=\{0,1,0^{c},1^{c}\}\text{"}\}$. 
(This is satisfiable.)
\end{itemize}
In this case clearly the merging causes confusion.
\end{example}
\begin{example}[when the merging targets overlap]\label{example:3sort}\small\
\begin{itemize}
\item $\Sigma=(\{s',s'',s'''\},\{d':\to s', d'':\to s'', d'''\to s'''\},\emptyset)$,\\
$T=\{\text{" } \texttt{s} \text{ has exactly two elements"}, \text{ "}\A_{\texttt{ss}}=\{0,1,0^{c},1^{c}\}\text{"}\mid\texttt{s}\in\{s',s'',s'''\}\}$,
\item $\Sigma_0=(\{s, s'''\},\{d',d'':\to s, d'''\to s'''\},\emptyset)$,\\
$T_0=\{\text{" } \texttt{s} \text{ has exactly two elements"}, d'\neq d'', \text{ "}\A_{\texttt{ss}}=\{0,1,0^{c},1^{c}\}\text{"}\mid\texttt{s}\in\{s,s'''\}\}$,
\item[] $\chi_0:\Sigma\to\Sigma_0$ $(\chi_0(s')=\chi_0(s'')=s, \chi(s''')=s''', \chi_0(d')=d', \chi_0(d'')=d'', \chi_0(d''')=d''')$,
\item $\Sigma_1=(\{s', s\},\{d'\to s' ,d'',d'''\to s\},\emptyset)$,\\
$T_1=\{\text{" } \texttt{s} \text{ has exactly two elements"}, d''\neq d''', \text{ "}\A_{\texttt{ss}}=\{0,1,0^{c},1^{c}\}\text{"}\mid\texttt{s}\in\{s',s\}\}$,
\item[] $\chi_1:\Sigma\to\Sigma_1$ $(\chi_1(s)=s, \chi_1(s'')=\chi_1(s''')=s, \chi_1(d')=d', \chi_1(d'')=d'', \chi_1(d''')=d''')$,
\item $\Sigma'=(\{s\},\{d',d'',d'''\to s\},\emptyset)$ ($\Sigma'$ is a push out of signature morphisms).
\end{itemize}
In this case clearly the merging causes confusion.
\end{example}

\end{document}